\documentclass[11pt]{article}

\usepackage[margin=1in]{geometry}
\usepackage[T1]{fontenc}
\usepackage{lmodern}
\usepackage{microtype}
\usepackage{amsmath,amssymb,amsthm,mathtools}
\usepackage{booktabs}
\usepackage{array}
\usepackage{enumitem}
\usepackage[dvipsnames]{xcolor}

\usepackage{float}
\usepackage[hypertexnames=false]{hyperref}
\usepackage{algorithm}
\usepackage{algpseudocode}

\usepackage[capitalise,noabbrev]{cleveref}

\makeatletter
\AddToHook{env/algorithmic/begin}{\def\@currentcounter{ALG@line}}
\makeatother
\crefalias{ALG@line}{line}

\usepackage{thm-restate}

\usepackage[dvipsnames]{xcolor}
\usepackage{xparse}
\usepackage{tikz}
\usetikzlibrary{calc,hobby}
\usepackage{subcaption}

\definecolor{linkblue}{RGB}{25,80,135}
\hypersetup{colorlinks=true,linkcolor=linkblue,citecolor=linkblue,urlcolor=linkblue}
\setlist{nosep,leftmargin=*}
\allowdisplaybreaks

\newcommand{\paperTitle}{Symmetric Submodular Minimization from Comparisons}

\newtheorem{theorem}{Theorem}[section]
\newtheorem{lemma}[theorem]{Lemma}
\newtheorem{corollary}[theorem]{Corollary}

\theoremstyle{definition}
\newtheorem{definition}[theorem]{Definition}
\newtheorem{remark}[theorem]{Remark}

\numberwithin{equation}{section}



\newcommand{\cutval}{\operatorname{cut}}
\newcommand{\Var}{\operatorname{Var}}

\DeclareMathOperator*{\argmin}{arg\,min}
\DeclareMathOperator*{\argmax}{arg\,max}
\DeclareMathOperator{\poly}{poly}

\newcommand{\weightedMinCutUpperBound}{\softO\!\left(n^2 + \min\!\left\{m\B,\,  n\B^2\right\} \right)}

\newcommand{\weightedEdgeWeights}{\{0,1,\ldots,\B\}} % For the bounded-weight min-cut algorithms, to consistently include 0 or not

\NewDocumentCommand{\ex}{o m}{%
  \IfNoValueTF{#1} % Check if the optional argument is provided
    {\ensuremath{\mathbb{E}\left[ #2 \right]}} % Case 1: No optional argument, E[b]
    {\ensuremath{\mathbb{E}_{#1}\left[ #2 \right]}} % Case 2: With optional argument, E_a[b]
} 
\NewDocumentCommand{\pr}{o m}{%
  \IfNoValueTF{#1} % Check if the optional argument is provided
    {\ensuremath{\Pr\left[ #2 \right]}} % Case 1: No optional argument, Pr[b]
    {\ensuremath{\Pr_{#1}\left[ #2 \right]}} % Case 2: With optional argument, Pr_a[b]
}

\NewDocumentCommand{\var}{o m}{%
  \IfNoValueTF{#1} % Check if the optional argument is provided
    {\ensuremath{\Var\left[ #2 \right]}} % Case 1: No optional argument, Var[b]
    {\ensuremath{\Var_{#1}\left[ #2 \right]}} % Case 2: With optional argument, Var_a[b]
}

\newcommand{\indicator}[1]{\mathbf{1}\!\left[#1\right]}
\newcommand{\innerProd}[2]{\langle #1,#2\rangle} 
\newcommand{\setSeparator}{\,:\,}

\newcommand{\R}{\mathbb{R}} % Real numbers
\newcommand{\F}{\mathbb{F}} % Finite field
\newcommand{\Z}{\mathbb{Z}} % Integers 
\newcommand{\N}{\mathbb{N}} % Natural numbers {0,1,..}

\newcommand{\cI}{\mathcal{I}} % Downward-closed set
\newcommand{\downwardClosedDomain}{2^V\setminus\{V\}}

\newcommand{\hashFamily}{\mathcal{F}}
\newcommand{\cP}{\mathcal{P}}

\newcommand{\softO}{\widetilde{O}}

\newcommand{\cutGraph}[1]{\cutval_{#1}}
\newcommand{\cutGraphSet}[2]{\cutGraph{#1}(#2)}

\newcommand{\degreeGraph}[2]{\deg_{#2}(#1)}

\newcommand{\minimumGraphDegree}[1]{\min_v \{\degreeGraph{v}{#1}\}}

\newcommand{\totalWeight}{W}
\newcommand{\B}{B}

\newcommand{\edgesBetween}[3]{E_{#1}(#2,\,#3)}

\newcommand{\edgeWeightE}[1]{w_{#1}}

\newcommand{\edgeWeight}[2]{\edgeWeightE{\{#1,#2\}}}

\newcommand{\weightVertexSet}[2]{w(#1,#2)}

\newcommand{\supp}{H}
\newcommand{\suppEdges}{E_\supp}
\newcommand{\uEdges}{E_u}
\newcommand{\supportVerticesOne}{X_1}
\newcommand{\supportVerticesTwo}{X_2}
\newcommand{\missingWeightSum}[1]{\sum_{v\in V\setminus\{u\} \,:\,  \{u,v\}\notin #1 }
 \edgeWeight{u}{v}} % Missing weight in support graph

\newcommand{\expanderConstant}{c}
\newcommand{\expanderDemand}{D}
\newcommand{\contractionSet}[1]{K_{#1}}
\newcommand{\uncontractedVertices}{\bar K}
\newcommand{\crossingWeight}{W_{\mathrm{cross}}}

\newcommand{\hashQuerySet}[1]{X_{#1}}
\newcommand{\sampledIncidentWeight}{Y_h} %{\weightVertexSet{v}{\hashQuerySet{h}}} 
\newcommand{\externalIncidentWeight}{\weightVertexSet{v}{V\setminus V_i}}%externalIncidentWeight\coloneqq\weightVertexSet{v}{V\setminus V_i}} d_{\mathrm{out}}

\newcommand{\partition}{\Pi}
\newcommand{\contractedG}{G_{\partition}}
\newcommand{\partitionIndex}{t}

\newcommand{\gslDimension}{N}
\newcommand{\gslWeight}{\beta}  
\newcommand{\gslDistinctValues}{R}
\newcommand{\gslSubspace}{\mathcal{L}}
\newcommand{\cutIndicator}[1]{\chi_{#1}} % Indicator vector for a cut

\newcommand{\sfmFunction}{f}
\newcommand{\mcSize}{n}
\newcommand{\mcRoot}{r}
\newcommand{\mcPrefix}[1]{V_{#1}}
\newcommand{\sfmRootedFunction}{g}
\newcommand{\sfmActiveBlocks}{\mathcal{A}}
\newcommand{\sfmRootBlock}{R}
\newcommand{\sfmBlockCount}{k}
\newcommand{\sfmContractedFunction}{\sfmFunction_{\sfmRootBlock,\sfmActiveBlocks}}

\newcommand{\sfmMergedBlock}{B_{\sfmBlockCount-1}\cup B_{\sfmBlockCount}}

\newcommand{\sfmDummyElement}{\rho}
\newcommand{\sfmBestSet}{T}
\newcommand{\sfmPenultimateBlock}{B_{\sfmBlockCount-1}}
\newcommand{\sfmLastBlock}{B_{\sfmBlockCount}}
\newcommand{\sfmDummyV}{V_{\sfmDummyElement}}

\newcommand{\unif}{\operatorname{Unif}}
\newcommand{\uniformSample}[2]{#1 \sim \unif(#2)}

\newcommand{\lbTargetCut}{S}
\newcommand{\lbOrderings}{\mathcal{H}}
\newcommand{\lbPathEdge}[2]{e^{#1}_{#2}}
\newcommand{\lbPrefix}[2]{P^{#1}_{#2}}
\newcommand{\lbBoundary}[2]{\partial_{#1}(#2)}

\newcommand{\lbSpacing}{d} % Spacing between perturbed edges
\newcommand{\lbPerturbIndices}{I} 
\newcommand{\lbPerturbIndicesLower}{I_1} 
\newcommand{\lbPerturbIndicesUpper}{I_2} 
\newcommand{\lbSharedPrefixEvent}{\mathcal{E}_{\lbOtherOrdering}}

\newcommand{\lbWeightScale}{K} 
\newcommand{\lbBaseGraph}[1]{G_{#1}}
\newcommand{\lbPerturbedGraph}[2]{G_{#1}^{#2}}
\newcommand{\lbEdgeWeight}[2]{w_{#1, #2}} % Edge weights for the lower bound graphs

\newcommand{\lbDistribution}{\mathcal{D}} 
\newcommand{\lbOtherOrdering}{\tau'}
\newcommand{\alg}{\mathcal{A}}

\newcommand{\lbBoundarySet}{\mathcal{B}}
\newcommand{\lbOtherBoundarySet}{\mathcal{C}}
\newcommand{\lbBoundaryVector}[1]{z(#1)} % Signed indicator for set #1
\newcommand{\lbBoundaryVectorEntry}[2]{\lbBoundaryVector{#1}_{#2}}% The #2-nd entry of the signed indicator for set #1  

\newcommand{\lbPairMatrix}[1]{Z^{#1}}
\newcommand{\lbPairMatrixEntry}[3]{\lbPairMatrix{#1}_{#2 #3}}

\newcommand{\lbSingletonMatrix}{H} 
\newcommand{\lbSingletonMatrixEntry}[2]{\lbSingletonMatrix_{#1 #2}} 

\newcommand{\lbProjectionMatrix}{\Gamma}
\newcommand{\lbProjectionCoeff}[2]{\gamma_{#1 #2}}

\newcommand{\lbQuadraticForm}[2]{\psi_{#1}(#2)} % #1 is the matrix, #2 is a vector 

\newcommand{\lbCertificateCount}{N} 
\newcommand{\lbEncodingMatrix}{M} 
\newcommand{\lbInitialEncodingMatrix}{\lbEncodingMatrix'} 
\newcommand{\lbEncodingScale}{\Lambda}

\newcommand{\lbCertificate}[1]{\Phi_{#1}}
\newcommand{\lbCertificatePair}[1]{\bigl(\lbPairMatrix{#1},\lbQuadraticForm{\lbSingletonMatrix}{#1}\bigr)}

\newcommand{\lbBoundaryCount}{t_{rs}}

\newcommand{\maxNorm}[1]{\lVert #1\rVert_{\max}}  
 
\newcommand{\lbWeightRange}{n^{O(1)}2^{O(\ell)}}  

\title{\vspace{-2.5em}%
  \makebox[\textwidth][c]{%
    \resizebox{1.10\textwidth}{!}{\bfseries \paperTitle}
  }}

\author{James Fox\thanks{EPFL, Switzerland. Part of this work was carried out while affiliated with the Max Planck Institute for Informatics, Germany. Email: james.fox@epfl.ch.} 
\and 
David P. Woodruff\thanks{Carnegie Mellon University and Google Research, USA. Email: dwoodruf@cs.cmu.edu.} 
}
\date{October 7, 2026}

\begin{document}

\maketitle
\vspace{2em}

\begin{abstract} 
Given value-oracle access to a symmetric submodular function $f:2^V\to\mathbb{R}$ with $|V|=n$, a nontrivial minimizer can be found using $O(n^3)$ value queries \cite{Queyranne}. We study the weaker \emph{comparison model}, in which a query on $S,T\subseteq V$ reveals only whether $f(S)$ is smaller than, equal to, or larger than $f(T)$. 

We give a deterministic polynomial-time algorithm that finds a nontrivial minimizer of any symmetric submodular function using $O(n^3)$ comparisons, matching the best-known deterministic value-oracle bound despite not knowing the function values. More generally, the same $O(n^3)$-comparison bound holds for minimization over the nonempty members of any downward-closed family. Our algorithm combines the minimum-capacity ordering recently introduced by Iwata and Konno \cite{IwataKonno} with the contraction framework of Goemans and Soto \cite{GoemansSoto}.

Applying this result to weighted graph cut functions resolves the main open question of Cohen-Addad et al. \cite{comparison}, who gave an $\widetilde{O}(n^3)$-comparison algorithm that runs in exponential time and asked whether a weighted minimum cut can be found in polynomial time using comparisons. For graphs with $m$ edges of integer weight at most $B$, we also give a deterministic polynomial-time algorithm that finds a minimum cut using
\[
\widetilde{O}\!\left(n^2+\min\!\left\{mB,\,nB^2\right\}\right)
\]
comparisons, improving on the $O(n^3)$ bound when $B$ is small. Finally, we show that every randomized algorithm that outputs a minimum cut with probability at least $2/3$ makes $\Omega(n \log n)$ expected comparisons in the worst case. Under the stronger assumption that all edge weights are polynomially bounded integers, we obtain an $\Omega(n \log \log n)$ expected comparison lower bound. These bounds contrast with the value-oracle model, where no $\omega(n)$ lower bound is known even for deterministic algorithms.
\end{abstract}

\clearpage
% \setcounter{page}{1} % Reset page numbering to avoid counting cover page

%%%%%%%%%%%%%%%
% Contents
%%%%%%%%%%%%%%%

%%%%%%%%%%%%%%%%%%%
%%%%%%%%%%%%%%%%%%%
\section{Introduction}\label{sec:intro}
%%%%%%%%%%%%%%%%%%%
Suppose that we want to minimize a set function~${f:2^V\to\R}$ on an~$n$-element ground set~$V$. 
In combinatorial optimization, access to~$f$ is often modeled by a \emph{value oracle}: for any query~${S\subseteq V}$, the oracle returns the exact value~${f(S)}$.    

Cohen-Addad et al.~\cite{comparison} studied a weaker \emph{comparison-oracle} model. 
In this model, a comparison query on~${S,T\subseteq V}$ only returns whether~${f(S)}$ is smaller than, equal to, or larger than~${f(T)}$, but does not reveal either function value. 
One problem they considered is finding a minimum cut in a graph when only comparisons between cut values are available.

Let~${G=(V,w)}$ be an undirected graph with nonnegative edge weights.\footnote{We assume throughout that graphs are undirected and loopless and have nonnegative edge weights.}
The \emph{cut function} of~$G$ is
\[
\cutGraphSet{G}{S}
\;\coloneqq\;
\sum_{\{u,v\}\in \binom{V}{2}}
\edgeWeight{u}{v}\,\indicator{|S\cap\{u,v\}|=1}. 
\]
Equivalently,~${\cutGraphSet{G}{S}}$ is the total weight of the edges crossing between~$S$ and~${V\setminus S}$. A \emph{minimum~cut} is any set that minimizes~${\cutGraph{G}}$ over the nonempty proper subsets of~$V$. 
The graph is~\emph{unweighted} if~${\edgeWeight{u}{v}\in\{0,1\}}$ for all~${u,v\in V}$.

For unweighted graphs, Cohen-Addad et al. gave a randomized polynomial-time algorithm that finds a minimum cut with high probability using~${\softO(n)}$ comparison queries~\mbox{\cite[Theorem~1.6]{comparison}}.\footnote{The section and theorem numbering differ between the STOC and arXiv versions of~\cite{comparison}. We refer to~\href{https://arxiv.org/abs/2511.15142v2}{arXiv~v2} throughout.} 
For weighted graphs, however, they obtain polynomial-time results only in special cases; their general optimization algorithm makes~${\softO(n^3)}$ queries but requires exponential time~\cite[Theorems~1.1 and~1.8]{comparison}.
They asked whether a weighted minimum cut can be found in polynomial time using polynomially many queries~\cite[Section~5]{comparison}; this is also the first open problem in a talk at the STOC~2026 workshop ``Algorithmic Frontiers of Graph and Hypergraph Problems via Global Queries.''\footnote{``Combinatorial Optimization Using Comparison Oracles,'' \url{https://sagnikm.github.io/workshop-stoc.html}, accessed October 7, 2026.}

We answer this question using submodular function minimization~(SFM). 
A function~${f:2^V\to\R}$ is \emph{submodular} if 
\[
f(S)+f(T)
\;\ge\;
f(S\cup T)+f(S\cap T)
\qquad\text{for all }S,T\subseteq V,
\]
and is \emph{symmetric} if~${f(S)=f(V\setminus S)}$ for every~${S\subseteq V}$.
If~$f$ is symmetric and submodular, then~${f(\emptyset)=f(V)\le f(S)}$ for every~${S\subseteq V}$, so~$\emptyset$ and~$V$ trivially minimize~$f$. 
In \emph{symmetric~SFM}, we therefore seek a \emph{nontrivial minimizer} in
\[
\argmin_{S \,\in\, 2^V\setminus\{\emptyset,V\}}\, f(S).
\] 
Note that we require a minimizer but not the corresponding minimum value, as comparison queries do not reveal the function values. 
The cut function of an undirected graph with nonnegative edge weights is always symmetric and submodular, so the weighted minimum-cut problem is a special case of symmetric SFM.   

As in~\cite{comparison}, we forbid comparison queries involving the trivial minimizers~$\emptyset$ and~$V$, and all queries made by our algorithms respect this convention. 
Otherwise, these sets could serve as reference thresholds. 
For example, if it were the case that~${f(S)=f(\emptyset)}$ for some~${S\notin\{\emptyset,V\}}$, then comparing~$S$ with~$\emptyset$ would immediately reveal that~$S$ is a minimizer.

\subsection{Results}\label{sec:results}
 
Our first result is a deterministic comparison algorithm for symmetric~SFM.

\begin{restatable}[Symmetric submodular minimization using comparisons]{theorem}{SymmetricSFM}\label{thm:intro-sfm} 
Let~${f:2^V\to\R}$ be symmetric and submodular, and let~${n=|V| \geq 2}$. 
There is a deterministic polynomial-time algorithm that, given comparison-oracle access to~$f$, finds a nontrivial minimizer of~$f$ using~${O(n^3)}$ comparisons.

\end{restatable}

This~${O(n^3)}$ comparison bound matches the best-known deterministic query complexity of symmetric SFM in the stronger value-oracle model~\cite{Queyranne,Nagamochi10,IwataKonno}. 
Since each comparison can be simulated by two value-oracle queries, any asymptotic improvement to the bound in~\Cref{thm:intro-sfm} would give the same improvement in the value-oracle model. 

We obtain~\Cref{thm:intro-sfm} as a corollary of the following more general result. 
A family~${\cI\subseteq2^V}$ is \emph{downward closed} if, for all~${S\in\cI}$, every subset of~$S$ is also in~$\cI$. 
Such families are also called \emph{hereditary}~\cite{GoemansSoto}. 
Examples of downward-closed families include the independent sets of a matroid and the family~${\{S\subseteq V : |S|\leq k\}}$ of sets subject to a cardinality constraint. 
Since~$\emptyset$ and~$V$ trivially minimize every symmetric submodular function, we assume that~${V\notin \cI}$ and minimize the function over~${\cI\setminus\{\emptyset\}}$. 
Such a minimizer is \emph{inclusionwise minimal} if none of its strict nonempty subsets is also a minimizer.

\begin{restatable}[Symmetric~SFM over a downward-closed family]{theorem}{DownwardClosedSFM}
\label{thm:intro-constrained}
Let~${f:2^V\to\R}$ be symmetric and submodular, let~${n=|V| \geq 2}$, and let~${\cI\subseteq\downwardClosedDomain}$ be downward closed with~${|\cI|\geq 2}$.  
There is a deterministic polynomial-time algorithm that, given a comparison oracle for~$f$ and a membership oracle for~$\cI$, finds an inclusionwise-minimal set in
\[ 
 \argmin_{S \,\in\, \cI\setminus\{\emptyset\}} \, f(S)
\]
using~${O(n^3)}$ comparisons and~${O(n)}$ membership queries.
\end{restatable}

Again, this~${O(n^3)}$ bound matches the best-known bound in the value-oracle model~\cite{GoemansSoto}. 
Setting~${\cI \coloneqq 2^V\setminus\{V\}}$ in~\Cref{thm:intro-constrained} then gives the symmetric SFM algorithm in~\Cref{thm:intro-sfm}.   

Since graph cut functions are symmetric and submodular,~\Cref{thm:intro-sfm} also allows us to find a global minimum cut of a graph with arbitrary nonnegative edge weights using~${O(n^3)}$ cut comparisons and polynomial time. 
When the edge weights are integers in~${\weightedEdgeWeights}$, we obtain the following parameterized bounds, which improve on the~${O(n^3)}$ bound for small~$\B$.
We assume that~${\B\in\N_{\geq 1}}$ throughout and use the notation~${\softO(\cdot)}$ to suppress polylogarithmic factors in~$n$ and~$\B$.  
All running-time bounds are polynomial in~$n$ and the size of the input parameters.

\begin{restatable}[Weighted minimum cut with comparisons]{theorem}{BoundedWeightMinCutUpperBound}
\label{thm:bounded-weight-min-cut-upper-bound} 
Let~$G$ be an undirected graph on~${n\geq 2}$ vertices with integer edge weights in~${\weightedEdgeWeights}$, and let~$m$ be the number of positive-weight edges.
There is a deterministic polynomial-time algorithm that finds a minimum cut of~$G$ using
\[
    \weightedMinCutUpperBound
\]
comparison queries. 
For the~$n\B^2$ term, the algorithm needs to know~$\B$.
\end{restatable}

A reader used to the cut-query model may find the~${n^2}$ term in~\Cref{thm:bounded-weight-min-cut-upper-bound} suspicious. 
In the \emph{cut-query} model, each query~$S$ returns the exact value~${\cutGraphSet{G}{S}}$, which allows us to reconstruct the entire graph in~${O(n^2)}$ queries using the identity 
\[
 \edgeWeight{u}{v}
 \;=\;
 \tfrac12\bigl(
 \cutGraphSet{G}{\{u\}}+\cutGraphSet{G}{\{v\}}-\cutGraphSet{G}{\{u,v\}}
 \bigr).
\]
In the comparison model, we can reconstruct an \emph{unweighted} graph using~${O(n^2)}$ cut comparisons for any~${n\geq4}$~\cite[Theorem~1.7]{comparison}. 
However, \emph{weighted} graphs generally cannot be recovered at all, even up to scaling, and even if we compare all possible pairs of cuts~\cite[Appendix~D.4.2]{comparison}.

Our final result is a comparison-query lower bound for weighted minimum cut. For deterministic algorithms, an~${\Omega(n\log n)}$ lower bound follows from the two-party communication complexity of graph connectivity~\cite{hajnal1988onthecommunicationcomplexity}, even for integer edge weights bounded by~${O(n^2)}$.\footnote{% Start footnote
For the communication reduction, let~${E_A}$ and~${E_B}$ be the edges held by Alice and Bob, respectively. 
Let~${M\coloneqq \binom{n}{2}+1}$, and let~$G$ be the graph with edge weights~${\edgeWeightE{e}\coloneqq M\indicator{e\in E_A}+\indicator{e\in E_B}}$.
To simulate a comparison between~$S$ and~$T$, Alice locally computes the contribution of her edges to~${\Delta\coloneqq\cutGraphSet{G}{S}-\cutGraphSet{G}{T}}$, and Bob does the same for his edges. 
Since~${M>\binom{n}{2}}$, Alice's contribution determines the sign of~$\Delta$ unless it is zero, in which case Bob's contribution determines the sign. 
Thus, each comparison can be simulated using~${O(1)}$ bits of communication.
The minimum cut in~$G$ has value~$0$ exactly when~${E_A\cup E_B}$ is disconnected, so the~${\Omega(n\log n)}$ communication lower bound for connectivity~\cite{hajnal1988onthecommunicationcomplexity} gives an~${\Omega(n\log n)}$ deterministic comparison lower bound for weighted minimum cut.
% End footnote
}
For bounded-error randomized algorithms, we obtain the following result, which ranges between~${\Omega(n)}$ and~${\Omega(n\log n)}$ depending on the magnitude of the edge weights.

\begin{restatable}[Randomized lower bound for weighted minimum cut]{theorem}{randomizedlowerbound}
\label{thm:lower}
Let~$n$ be even and sufficiently large, and let~$\ell$ be an integer with~${2\leq\ell<n/2}$. 
Let~$\mathcal{G}$ be the class of complete~$n$-vertex graphs with positive integer edge weights bounded by~${\lbWeightRange}$.
Every randomized comparison algorithm that outputs a global minimum cut with probability at least~${2/3}$ on every graph in~$\mathcal{G}$ has worst-case expected query complexity
\[
  \Omega(n\log\ell).
\]
\end{restatable}

Setting~${\ell\coloneqq \lceil\log n\rceil}$ in~\Cref{thm:lower} gives an~${\Omega(n\log\log n)}$-comparison lower bound for minimum cut when the edge weights are bounded by~${n^{O(1)}}$.
Moreover, for every fixed~${\varepsilon\in(0,1)}$, setting~${\ell\coloneqq \lfloor(n/2)^\varepsilon\rfloor}$ gives an~${\Omega_\varepsilon(n\log n)}$ lower bound when the edge weights are bounded by~${2^{O_\varepsilon(n^\varepsilon)}}$.\footnote{The notation~$O_\varepsilon(\cdot)$ and~${\Omega_\varepsilon(\cdot)}$ allows the hidden constants to depend on~$\varepsilon$.} 
These results contrast with the cut-query model, where no~${\omega(n)}$ lower bound is known even for deterministic algorithms. 

We prove the lower bound by constructing a hard family of graphs parameterized by~$\ell$. 
The restriction~${\ell<n/2}$ comes from this construction: we select~$\ell$ edges from each half of a Hamiltonian path and perturb their weights to create new minimum cuts while preserving the answers to all previous strict comparisons. 
Consequently, the lower bound in~\Cref{thm:lower} is always at most~$O(n\log n)$ and cannot be made arbitrarily large by increasing~$\ell$.

\subsection{Technical overview}\label{sec:overview}

We outline the proofs of~\Cref{thm:intro-sfm,thm:intro-constrained,thm:bounded-weight-min-cut-upper-bound,thm:lower} below.

\paragraph{Symmetric SFM (\Cref{thm:intro-sfm}).} 
Our comparison algorithm follows the contraction approach used by several value-oracle algorithms for symmetric~SFM~\cite{Queyranne,Nagamochi10,IwataKonno}: greedily construct an ordering of the ground set, contract the final two elements of the ordering into one element, and repeat.
The orderings in these three algorithms are chosen so that either a singleton is a minimizer, or every minimizer contains both final elements or neither.  
Storing the original~$n$ singletons and at most one new singleton per contraction gives~${O(n)}$ distinct candidate minimizers, which we compare at the end to return a minimizer. 
However, comparison queries do not reveal the function values, so we cannot directly evaluate sums or differences between function values.

For an ordering~${(v_1,\ldots,v_{n})}$ of~$V$, write~${\mcPrefix{0}\coloneqq\emptyset}$ and~${\mcPrefix{i}\coloneqq\{v_1,\ldots,v_i\}}$.
Suppose that~${v_1,\ldots,v_{i-1}}$ have already been chosen.
Each of the three algorithms above chooses the next element~$v_i$ from~${V\setminus\mcPrefix{i-1}}$ using a different greedy rule.
Queyranne's algorithm~\cite{Queyranne} uses a \emph{maximum-adjacency}~(MA) ordering, which selects
\[
v_i
\;\in\;
\argmax_{v\in V\setminus\mcPrefix{i-1}}\,
\bigl(\sfmFunction(\mcPrefix{i-1})+\sfmFunction(\{v\})-\sfmFunction(\mcPrefix{i-1}\cup\{v\})\bigr).
\] 
We cannot evaluate this rule directly using comparison queries, since the oracle does not reveal the function values needed for the sums and differences. 
Nagamochi~\cite{Nagamochi10} later introduced a \emph{minimum-degree}~(MD) ordering, which instead selects
\[
v_i
\;\in\;
\argmin_{v\in V\setminus\mcPrefix{i-1}}\,
\bigl(\sfmFunction(\mcPrefix{i-1}\cup\{v\})+\sfmFunction(\{v\})\bigr).
\]
Again, a comparison oracle cannot directly compare these sums.
We instead use the \emph{minimum-capacity}~(MC) ordering recently introduced by Iwata and Konno~\cite{IwataKonno}, which selects
\[
v_i
\;\in\;
\argmin_{v\in V\setminus\mcPrefix{i-1}}\,
\sfmFunction(\mcPrefix{i-1}\cup\{v\}).
\]
This selection rule is a plain minimum of function values (without sums or differences), so we can compute~$v_i$ using~${n-i}$ comparisons by scanning over~${V\setminus\mcPrefix{i-1}}$ and keeping the best candidate element seen so far. 
We can therefore construct a full MC ordering using~${(n-1) + (n-2) + \ldots + 1 = O(n^2)}$ comparisons without knowing any function values.

Iwata and Konno~\cite{IwataKonno} show that every set~$S$ containing exactly one of the final two elements of an MC ordering satisfies~${f(S) \geq \min_{v\in V}f(\{v\})}$. 
We may therefore contract the final pair and recurse on a ground set that is one element smaller until only two elements remain, keeping track of the best candidate minimizer so far at each step. 
This algorithm uses~${O(n)}$ contractions, each requiring~${O(n^2)}$ comparisons, giving a total of~${O(n^3)}$ comparisons.

\paragraph{Symmetric SFM over a downward-closed family (\Cref{thm:intro-constrained}).} 
For minimization over a downward-closed family~$\cI$, we cannot directly use the contraction approach above, as the union of the last two elements in an MC ordering might no longer lie in~$\cI$.  
Instead, following the approach of Goemans and Soto~\cite[Algorithm~2]{GoemansSoto}, we first check whether the union of the final two elements still lies in~$\cI$; if so, then we contract them as before, and otherwise, we add both elements to a \emph{root block} of discarded elements, which we contract into a single element~$\mcRoot$ that is fixed at the start of every subsequent ordering.

There are two further complications. 
First, Goemans and Soto use an MA~ordering, which we cannot construct using comparisons, as discussed above. 
Second, we cannot simply substitute an~MC ordering into the Goemans--Soto algorithm, as their algorithm requires the root to be fixed first in the ordering for each contraction, whereas the first element in an~MC~ordering must correspond to a singleton that minimizes~$\sfmFunction$. 
We resolve this by defining a modified MC~ordering: we call an ordering~${(v_1,\ldots,v_{n})}$ an~\emph{$\mcRoot$-rooted MC ordering} for~$\sfmFunction$ if~${v_1=\mcRoot}$ and, for every~${i \geq 2}$, the element~$v_i$ satisfies the MC ordering inequality.  
Using Iwata and Konno's inequality for~MC orderings~\mbox{\cite[Lemma~1]{IwataKonno}}, we show that if~${(v_1,\ldots,v_{n})}$ is an~${\mcRoot}$-rooted MC ordering for~$\sfmFunction$, then the contraction property we need still holds for all sets not containing~$\mcRoot$~(see~\Cref{cor:rooted-mc}). 
This inequality and the choice of ordering allow us to adapt the Goemans--Soto algorithm to use comparisons~(see~\Cref{alg:downward-closed-sfm}).

\paragraph{Weighted minimum cut with bounded weights 
(\Cref{thm:bounded-weight-min-cut-upper-bound}).}
For graphs with edge weights in~$\weightedEdgeWeights$, we show how to compute a minimum cut using~$\weightedMinCutUpperBound$ queries. We prove the~${m\B}$ and~${n\B^2}$ bounds separately, and then run the corresponding algorithms in parallel.

For the first bound, let~$\totalWeight$ be the total edge weight in the graph~$G$, and represent each cut~${S\subseteq V}$ by the indicator vector~${\cutIndicator{S}\in\{0,1\}^{\binom{n}{2}}}$ of its crossing edges. 
Then we can write~${\cutGraphSet{G}{S}=\innerProd{\cutIndicator{S}}{\gslWeight}}$, where~${\gslWeight\in\R^{\binom{n}{2}}}$ is a vector of the unknown edge weights and~${\innerProd{\cdot}{\cdot}}$ is the standard inner product. 
If two cuts~${S,T}$ satisfy~${\cutGraphSet{G}{S}=\cutGraphSet{G}{T}}$, then~${\innerProd{\gslWeight}{\cutIndicator{S}-\cutIndicator{T}}=0}$. 
Therefore each time that the oracle confirms that two cuts have the same value, this gives us a linear constraint on~${\gslWeight}$ that we may store and reuse without making any additional queries. 

Following the Global Subspace Learning method of Cohen-Addad et al.~\cite[Theorem~4.2]{comparison}, we store one representative set for each cut value discovered so far and build a subspace of~${\gslWeight^\perp}$ spanned by the linear constraints obtained from each cut equality.  
For each newly queried cut~$S$, we first check whether the stored constraints already imply that~${\cutGraphSet{G}{S}}$ is equal to the cut value of a representative. 
If not, we use binary search among the representatives; this either confirms that the value~${\cutGraphSet{G}{S}}$ has not been seen yet or discovers a new constraint on~$\gslWeight$.  
There are at most~${\totalWeight+1}$ distinct cut values and at most~${\binom{n}{2}}$ independent constraints, so this allows us to simulate the~${O(n^3)}$-query algorithm from~\Cref{thm:intro-sfm} using~${O((n^2+\totalWeight)\log(\totalWeight+2))}$ comparisons~(see~\Cref{cor:weighted-min-cut-total-weight}). 
The inequality~${\totalWeight\le m\B}$ gives the bound~${\softO(n^2+m\B)}$.

When the graph is dense, the term~${m\B}$ can be as large as~${\Theta(n^2\B)}$.
To prove the second upper bound, we replace~$G$ by a contracted graph~$\contractedG$ with total edge weight~${\softO(n\B^2)}$. We then run the previous algorithm on~$\contractedG$, obtaining an upper bound of~${\softO(n^2+n\B^2)}$, which improves the query complexity when~${\B= o(m/n)}$. 
Although graph contractions and expander decompositions have been used throughout the cut-query literature~\cite{RSW,StarContraction,AnandSaranurakWang,KennethMordochKrauthgamer2025,KennethMordochKrauthgamer2026a,KennethMordochKrauthgamer2026b}, prior algorithms either assume that the graph is unweighted or that we know the exact weighted cut values. 
The challenge is finding a contraction that preserves the nonsingleton minimum cuts in~$G$ without knowing the edge weights.

Fix~${v \in V}$ and~${S\subseteq V \setminus \{v\}}$. 
Write~$\degreeGraph{v}{G}$ for the weighted degree of~$v$ and~${\weightVertexSet{v}{S}}$ for the sum of edge weights between~$v$ and~$S$. 
Although comparison queries do not reveal~$\weightVertexSet{v}{S}$, the identity
\[ 
 \cutGraphSet{G}{S\cup\{v\}}-\cutGraphSet{G}{S}
 \;=\;
 \degreeGraph{v}{G}-2\weightVertexSet{v}{S}
\]
lets us use a single comparison to test whether edges to~${S}$ account for more or less than half of the weight incident to~${v}$. 
In~\Cref{subsec:recovering-support-graph}, we use a greedy algorithm built on this test to recover an unweighted graph~$H$ such that each edge in~$H$ has positive weight in~$G$, and such that the unrecovered weight at each vertex is at most~${\B}$. 

We then compute an \emph{expander decomposition} of~$H$. 
This gives a partition of~$V$ into blocks such that each nontrivial split of a block is crossed by at least~${\Omega(\B)}$ recovered edges per vertex on the smaller side, while only~${\softO(n\B)}$ recovered edges join different blocks~(see~\Cref{thm:expander}). 
Choosing the parameters of the expander decomposition is the only step of the algorithm that needs to know~$\B$.
Finally, for each nonsingleton block~$V_i$, we identify a set~${\contractionSet{i}\subseteq V_i}$ of vertices whose incident weight lies mostly inside~$V_i$~(\Cref{lemma:external}), and contract each nonempty~$\contractionSet{i}$ into a single vertex. 
We show in~\Cref{lemma:contraction} that, unless a singleton is already a minimum cut, every minimum cut is preserved by these contractions, and that the resulting contracted graph has total edge weight~${\softO(n\B^2)}$.

\paragraph{The randomized lower bound~(\Cref{thm:lower}).} 
Our lower bound uses a perturbation argument: we first construct a family of graphs whose minimum cuts encode the prefixes of a hidden ordering of the vertices.
We then show that we can perturb individual edge weights so that any correct deterministic algorithm must identify many of these prefixes. 
We conclude by choosing a hard distribution over these graphs and applying Yao's principle. 

Let~${n=2k}$ be even, consider the vertex set~${V\coloneqq[n]=\{1,2,\ldots,n\}}$, and let~$\ell$ be an integer with~${2\leq\ell<k}$.
Fix~${\lbTargetCut\subseteq V}$ with~${|\lbTargetCut|=k}$, and let~$\lbOrderings$ be the family of all orderings~${\tau=(v_1,\ldots,v_n)}$ of~$V$ such that~${(v_1,\ldots,v_k)}$ is an ordering of~$\lbTargetCut$ and~${(v_{k+1},\ldots,v_n)}$ is an ordering of~${V\setminus\lbTargetCut}$.
The algorithm may know~$\lbTargetCut$, but it does not know~$\tau$ initially.

We view~$\tau$ as a Hamiltonian path that visits every vertex in~$\lbTargetCut$ before every vertex in~${V\setminus\lbTargetCut}$.
Write~${\lbPrefix{\tau}{i}\coloneqq\{v_1,\ldots,v_i\}}$ for the prefix of the path~$\tau$ at position~${i\in[n-1]}$, and call the cuts represented by these sets \emph{prefix cuts}.
For each~${\tau\in\lbOrderings}$, we construct a complete \emph{base graph}~$\lbBaseGraph{\tau}$ whose minimum cuts are exactly these prefix cuts, and whose edge weights are positive even integers bounded by~$\lbWeightRange$~(see~\Cref{lemma:base_graph_properties}).

Next, we perturb each base graph~$\lbBaseGraph{\tau}$ so that distinguishing it from a perturbed graph requires querying a particular prefix cut.
In~\Cref{subsec:path-construction}, we choose a set~${\lbPerturbIndices\subseteq[n-1]}$ with~${|\lbPerturbIndices|=2\ell}$, independently of~$\tau$.
For each~${i\in\lbPerturbIndices}$, we obtain the \emph{perturbed graph}~$\lbPerturbedGraph{\tau}{i}$ from~$\lbBaseGraph{\tau}$ by decreasing the weight of the path edge~${\{v_i,v_{i+1}\}}$ by~$1$ and leaving all other edge weights unchanged.
Among the prefix cuts, only the cut represented by~$\lbPrefix{\tau}{i}$ crosses this edge, so it is the \emph{unique} minimum cut of~$\lbPerturbedGraph{\tau}{i}$~(up to complements).
Moreover, each cut value in~$\lbPerturbedGraph{\tau}{i}$ decreases by at most~$1$ relative to the base graph.
Since all cut values in~$\lbBaseGraph{\tau}$ are even, it follows that for all~${X,Y\subseteq V}$, if~${\cutGraphSet{\lbBaseGraph{\tau}}{X}<\cutGraphSet{\lbBaseGraph{\tau}}{Y}}$, then~${\cutGraphSet{\lbPerturbedGraph{\tau}{i}}{X}<\cutGraphSet{\lbPerturbedGraph{\tau}{i}}{Y}}$.
That is, the perturbation preserves all strict comparisons, so an algorithm can only distinguish~$\lbBaseGraph{\tau}$ from~$\lbPerturbedGraph{\tau}{i}$ by querying a pair~${X,Y\subseteq V}$ such that~${\cutGraphSet{\lbBaseGraph{\tau}}{X}=\cutGraphSet{\lbBaseGraph{\tau}}{Y}}$ but~${\cutGraphSet{\lbPerturbedGraph{\tau}{i}}{X}\neq\cutGraphSet{\lbPerturbedGraph{\tau}{i}}{Y}}$.
We prove in~\Cref{lemma:one-edge-perturbation} that for any such pair~$X,Y$, both sets represent prefix cuts of~$\tau$, and one of the sets belongs to~${\{\lbPrefix{\tau}{i},\,V\setminus\lbPrefix{\tau}{i}\}}$.

Informally, these properties ensure that a deterministic algorithm that is correct on a sufficiently large proportion of the perturbed graphs must make comparisons involving a large proportion of the prefix cuts~${\lbPrefix{\tau}{i}}$ or their complements. 
To prove the lower bound, we define a hard distribution~$\lbDistribution$ over these graphs.
The distribution~$\lbDistribution$ samples an ordering~$\tau$ uniformly at random from~$\lbOrderings$, then returns the base graph~$\lbBaseGraph{\tau}$ with probability~${1/2}$, or a uniformly chosen perturbed graph~$\lbPerturbedGraph{\tau}{i}$ with~${i\in\lbPerturbIndices}$ otherwise.
In~\Cref{lemma:common-cuts}, we show that for any fixed~${\tau\in\lbOrderings}$, at most~${|\lbOrderings|2^{-\Omega(n\log\ell)}}$ distinct orderings in~$\lbOrderings$ share at least~$\ell$ prefix cuts at positions in~$\lbPerturbIndices$ with~$\tau$.
Using this bound and counting comparison transcripts, we show in~\Cref{lemma:distributional-lower-bound} that every deterministic algorithm that outputs a minimum cut with probability at least~${9/10}$ for graphs drawn from~$\lbDistribution$ makes~${\Omega(n\log\ell)}$ comparisons in the worst case.
The randomized lower bound in~\Cref{thm:lower} follows by Yao's principle, using a constant number of repetitions and Markov's inequality to bound the worst-case expected query complexity.

\subsection{Related work}\label{sec:related}
 
\paragraph{Symmetric SFM in the value-oracle model.}  
Queyranne~\cite{Queyranne} extended the minimum-cut algorithm of Nagamochi and Ibaraki~\cite{NagamochiIbaraki92} to symmetric submodular functions, giving a deterministic algorithm using~${O(n^3)}$ value queries.
As described above, Queyranne's algorithm uses maximum-adjacency~(MA) orderings, while Nagamochi~\cite{Nagamochi10} later obtained the same query bound using minimum-degree~(MD) orderings.
Goemans and Soto~\cite{GoemansSoto} used an MA ordering to minimize a symmetric submodular function over the nonempty members of a downward-closed family, also using~${O(n^3)}$ value queries. 
Iwata and Konno recently introduced a framework that unifies~MA and~MD orderings~\cite{IwataKonno}. 
They give a family of orderings parameterized by~${\alpha\in[-1,1]}$, with MA and MD orderings corresponding to~${\alpha=-1}$ and~${\alpha=1}$, respectively. 
Their MC ordering, which we use in this paper, corresponds to~${\alpha=0}$.

No deterministic algorithm for symmetric~SFM using~${o(n^3)}$ value-oracle queries is known, but allowing randomization or a dependence on the range does help. 
For symmetric submodular functions taking integer values in~$[-M,M]$, Chekuri and Quanrud~\cite[Corollary~3]{ChekuriQuanrud2021} give a randomized algorithm that finds a nontrivial minimizer with high probability using~${\softO(n^2\log(nM))}$ queries and~${\poly(n,\log M)}$ additional time.
Their algorithm extends the isolating-cuts framework of Li and Panigrahi~\cite{lipanigrahi2022deterministicmincut} to submodular functions.
Alternatively, applying Chekuri and Quanrud's reduction from symmetric SFM to general SFM~\cite[Theorem~2]{ChekuriQuanrud2021} with Jiang's exponential-time~${O(n^2\log n)}$-query algorithm for general~SFM~\cite{Jiang22} gives a randomized~${\softO(n^2)}$-query algorithm for symmetric~SFM, with exponential running time.

\paragraph{The cut-query model.} 
In the \emph{cut-query model}, each query is a set of vertices in a graph, and an oracle returns the value of the cut. 
This is a special case of the value-oracle model in which the unknown function is a cut function.
Rubinstein, Schramm, and Weinberg~\cite{RSW} studied this model and gave a randomized algorithm that finds a minimum cut of an unweighted graph using~${\softO(n)}$ queries, even though reconstructing the graph requires~${\Theta(n^2/\log n)}$ cut queries in the worst case. 
Apers~et~al.~\cite{StarContraction} later improved the bounded-error randomized query complexity of unweighted minimum cut to~${O(n)}$, and Anand, Saranurak, and Wang~\cite{AnandSaranurakWang} recently gave the first deterministic algorithm that beats reconstructing the graph, computing an unweighted minimum cut using~${\softO(n^{5/3})}$ cut queries. 
For weighted graphs, Mukhopadhyay and Nanongkai~\cite{MukhopadhyayNanongkai} gave a randomized~${\softO(n)}$-query algorithm for computing a minimum cut, but no subquadratic deterministic algorithms are known.

For lower bounds, Graur et al.~\cite{GraurPollnerRamaswamyWeinberg20} defined the \emph{cut dimension} of a graph as the dimension of the space spanned by the indicator vectors of its minimum cuts, and used this to prove a deterministic lower bound of~${3n/2-2}$ queries for computing a weighted minimum-cut value. Lee et al.~\cite{LLSZ} showed that every graph has cut dimension at most~${2n-3}$, but improved the lower bound to~${2n-2}$ queries using a generalization of cut dimension. This is also the best-known deterministic lower bound for the more general problem of symmetric~SFM.  
For zero-error randomized algorithms, Auza and Lee~\cite{AuzaLee21} proved that every such algorithm that decides whether a weighted graph is connected makes at least~${n/4}$ expected queries in the worst case in the linear query model, which implies the same lower bound in the cut-query model. 
For bounded-error algorithms, an~$\Omega(n/\log n)$ lower bound follows from the communication complexity of deciding connectivity~\cite{babai1986complexity}.

\paragraph{The comparison model.} 
Pairwise comparisons of combinatorial functions were studied before their use for exact optimization; for example, Balcan, Vitercik, and White~\cite{BalcanVitercikWhite} learn approximate comparison oracles for several function classes, including submodular functions. 
Cohen-Addad et al.~\cite{comparison} formulated the exact comparison-based optimization model studied here.
Their general linear-optimization result gives an~${\softO(n^3)}$-comparison algorithm for weighted minimum cut, but its running time can be exponential~\cite[Theorem~1.1]{comparison}.
For unweighted graphs, they adapted the cut-query approach of Rubinstein, Schramm, and Weinberg~\cite{RSW} to obtain a randomized polynomial-time algorithm using~${\softO(n)}$ comparisons~\cite[Theorem~1.6]{comparison}.
 For graphs with integer weights in~${\weightedEdgeWeights}$, they obtained a comparison-query bound of~${(n\B)^{O(r^2\log \B)}}$ when the graph has~$r$ distinct weighted degrees~\cite[Theorem~1.8]{comparison}.

\subsection{Organization}\label{sec:organization}
The paper is organized as follows. 
We prove the SFM upper bounds~(\Cref{thm:intro-sfm,thm:intro-constrained}) in~\Cref{sec:symmetric}.
For weighted minimum cut, we prove the upper bound~(\Cref{thm:bounded-weight-min-cut-upper-bound}) in~\Cref{sec:bounded} and the lower bound~(\Cref{thm:lower}) in~\Cref{sec:lower}.
We discuss open questions in~\Cref{sec:discussion}.
\Cref{app:weighted-graph-contraction} contains the proof of the weighted-graph contraction lemma~(\Cref{lemma:contraction}), and~\Cref{app:lower_bound_proofs} contains details omitted from the lower-bound proof.

%%%%%%%%%%%%%%%%%%%
\section{Comparison-based symmetric submodular minimization}\label{sec:symmetric} 
%%%%%%%%%%%%%%%%%%% 

\subsection{Minimum-capacity orderings}
We first recall the minimum-capacity~(MC) ordering of~Iwata and~Konno~\cite{IwataKonno}.
We then adapt their key inequality to allow minimizing a symmetric submodular function over an arbitrary downward-closed family. 

\begin{definition}[MC ordering and~${\mcRoot}$-rooted MC ordering]\label{def:mc} 
Let~${\sfmFunction:2^V\to\R}$ with~${\mcSize=|V|\geq 2}$, and let~${(v_1,\ldots,v_{\mcSize})}$ be an ordering of~$V$.
Write~${\mcPrefix{0}\coloneqq \emptyset}$ and~${\mcPrefix{i}\coloneqq\{v_1,\ldots,v_i\}}$ for each~${i\in[n]}$. 
\begin{itemize}
    \item The ordering~${(v_1,\ldots,v_{\mcSize})}$ is a \emph{minimum-capacity~(MC) ordering} for~$\sfmFunction$ if, for every~${i\in[\mcSize]}$,
    \begin{align}
    v_i
    \;\in\;
    \argmin_{v\in V\setminus\mcPrefix{i-1}}
    \sfmFunction\bigl(\mcPrefix{i-1}\cup\{v\}\bigr).
    \label{eq:mc-ordering-inequality}
    \end{align}
    \item Fix~${\mcRoot\in V}$. The ordering~${(v_1,\ldots,v_{\mcSize})}$ is an~\emph{${\mcRoot}$-rooted MC ordering} for~$\sfmFunction$ if~${v_1=\mcRoot}$ and~\eqref{eq:mc-ordering-inequality} holds for every~${i\in\{2,\ldots,\mcSize\}}$.
\end{itemize}
\end{definition}

An~${\mcRoot}$-rooted MC~ordering need not be an~MC ordering, as its first element is always~$\mcRoot$ rather than being chosen from~${\argmin_{v\in V} \sfmFunction\bigl(\{v\}\bigr)}$. However, using the following result of Iwata and Konno, we obtain the analogous inequality for an~$\mcRoot$-rooted MC ordering for sets that do not contain~$\mcRoot$.

\begin{lemma}[Iwata--Konno MC inequality~{\cite[Lemma~1]{IwataKonno}}]\label{lemma:mc}
Let~${\sfmFunction:2^V\to\R}$ be submodular, let~${\mcSize=|V|\geq2}$, and let~${(v_1,\ldots,v_{\mcSize})}$ be an MC ordering for~$\sfmFunction$.
If~${S\subsetneq V}$ satisfies~${|S\cap\{v_{\mcSize-1},v_{\mcSize}\}|=1}$, then
\[
\sfmFunction(S)+\sfmFunction(V\setminus S)
\;\ge\;
\sfmFunction(\mcPrefix{\mcSize-1})
+
\min_{v\in S}\sfmFunction(\{v\}).
\]
\end{lemma}

\begin{corollary}[Rooted MC ordering inequality]\label{cor:rooted-mc}  

Let~${\sfmFunction:2^V\to\R}$ be symmetric and submodular with~${\mcSize=|V|\ge2}$, and let~${(v_1,\ldots,v_{\mcSize})}$ be an~${\mcRoot}$-rooted MC ordering for~$\sfmFunction$. 
If a set~${S\subseteq V\setminus\{\mcRoot\}}$ satisfies~${|S\cap\{v_{\mcSize-1},v_{\mcSize}\}|=1}$, then
\[
2\sfmFunction(S)
\;\ge\;
\sfmFunction(\{v_{\mcSize}\})
+
\min_{v\in S}\sfmFunction(\{v\}).
\]
\end{corollary}

\begin{proof}
Let~${c \coloneqq \min_{v\in V\setminus\{\mcRoot\}} \sfmFunction(\{v\}) - \sfmFunction(\{\mcRoot\})}$ and define~${\sfmRootedFunction:2^V\to\R}$ by
\[
\sfmRootedFunction(X) 
\;\coloneqq\;
\sfmFunction(X) + c\,\indicator{\mcRoot\in X}.
\] 
The second term is modular, so~$\sfmRootedFunction$ is submodular. 
We claim that~${(v_1,\ldots,v_{\mcSize})}$ is an MC ordering for~$\sfmRootedFunction$. 
Since~${v_1=\mcRoot}$, we have~${\sfmRootedFunction(\{v_1\})=\sfmRootedFunction(\{\mcRoot\})=\sfmFunction(\{\mcRoot\})+c = \min_{v\in V\setminus\{\mcRoot\}}
\sfmFunction(\{v\})}$. As~${\sfmRootedFunction(\{v\})=\sfmFunction(\{v\})}$ for all~${v\in V\setminus\{\mcRoot\}}$, it follows that
\[
\sfmRootedFunction(\{v_1\})  
\;=\;
\min_{v\in V\setminus\{\mcRoot\}} \sfmRootedFunction(\{v\})
\;=\;
\min_{v\in V}
\sfmRootedFunction(\{v\}),
\] 
where the final step holds as~${v_1=\mcRoot}$. 
Hence~$v_1$ is a valid first element of an MC ordering for~$\sfmRootedFunction$.  

Let~${i\in\{2,\ldots,\mcSize\}}$. Then~${\mcRoot\in\mcPrefix{i-1}}$, so~${\sfmRootedFunction\bigl(\mcPrefix{i-1}\cup\{v\}\bigr) = \sfmFunction\bigl(\mcPrefix{i-1}\cup\{v\}\bigr)+c}$ for all~${v\in V\setminus\mcPrefix{i-1}}$. 
Adding a constant~$c$ to every function value does not change the minimizers, so
\[   
    \argmin_{v\in V\setminus\mcPrefix{i-1}}
    \sfmRootedFunction\bigl(\mcPrefix{i-1}\cup\{v\}\bigr)
    \;=\;
    \argmin_{v\in V\setminus\mcPrefix{i-1}}
    \sfmFunction\bigl(\mcPrefix{i-1}\cup\{v\}\bigr).
\]
Since~${(v_1,\ldots,v_{\mcSize})}$ is an~$\mcRoot$-rooted MC ordering for~$\sfmFunction$, 
\[
v_i 
\;\in\;
\argmin_{v\in V\setminus\mcPrefix{i-1}}
    \sfmRootedFunction\bigl(\mcPrefix{i-1}\cup\{v\}\bigr)
    \qquad\text{for all }~i\in\{2,\ldots,\mcSize\}.
\]
Therefore~${(v_1,\ldots,v_{\mcSize})}$ is an MC ordering for~$\sfmRootedFunction$. 

Now suppose that~${S\subseteq V\setminus\{\mcRoot\}}$ satisfies~${|S\cap\{v_{\mcSize-1},v_{\mcSize}\}|=1}$. By~\Cref{lemma:mc},
\begin{align}
\label{eq:mc-inequality-rooted-function}
\sfmRootedFunction(S)
+
\sfmRootedFunction(V\setminus S)
\;\geq\;
\sfmRootedFunction(\mcPrefix{\mcSize-1})
+
\min_{v\in S}\sfmRootedFunction(\{v\}).
\end{align}
Since~${\mcRoot\notin S}$ and~${\sfmFunction}$ is symmetric, we have~${\sfmRootedFunction(S)=\sfmFunction(S)}$ and~${\sfmRootedFunction(V \setminus S)=\sfmFunction(V \setminus S) + c=\sfmFunction(S) + c}$.
Moreover, as~${\mcRoot\in \mcPrefix{\mcSize-1}=V\setminus\{v_n\}}$, we have~${\sfmRootedFunction(\mcPrefix{\mcSize-1})=\sfmFunction(\mcPrefix{\mcSize-1})+c=\sfmFunction(\{v_n\})+c}$. 
Finally, as~${\mcRoot\notin S}$, we have~${\sfmRootedFunction(\{v\})=\sfmFunction(\{v\})}$ for all~${v\in S}$. 
Substituting these equalities into~\eqref{eq:mc-inequality-rooted-function} gives
\[
2\sfmFunction(S) + c
\;\geq\;
\sfmFunction(\{v_n\})+c
+\min_{v\in S}\sfmFunction(\{v\}),
\]
and the result follows by subtracting~$c$ from both sides. 
\end{proof}

Next, we show how to compute an MC ordering (or an ${\mcRoot}$-rooted MC ordering) using comparisons, without knowing the function values and without making a query involving~$\emptyset$ or~$V$.

\begin{lemma}[Computing an MC ordering with comparisons]
\label{lemma:computing-mc-ordering} 
Let~${\sfmFunction:2^V\to\R}$ be a set function with~${|V|=\mcSize\geq2}$, and let~${\mcRoot\in V}$. 
We can compute an MC ordering for~$\sfmFunction$, or an~${\mcRoot}$-rooted MC ordering, using~${O(\mcSize^2)}$ comparison queries. 
\end{lemma}

\begin{proof} 
We construct the ordering one element at a time.
Suppose that~${v_1,\ldots,v_{i-1}}$ have already been chosen. 
We find~${v_i \in \argmin_{v\in V\setminus\mcPrefix{i-1}} \sfmFunction\bigl(\mcPrefix{i-1}\cup\{v\}\bigr)}$ by comparing the values~${\sfmFunction(\mcPrefix{i-1}\cup\{v\})}$ for the remaining elements~${v\in V\setminus\mcPrefix{i-1}}$ one by one, and keeping an element with the smallest value found so far. This uses~${\mcSize-i}$ comparisons.

To construct an MC ordering, we apply this procedure for~${i \in [\mcSize-1]}$, and then take~$v_{\mcSize}$ to be the remaining element. For an~${\mcRoot}$-rooted MC ordering, we instead set~${v_1=\mcRoot}$ and apply the same procedure for~${i\in\{2,\ldots,\mcSize-1\}}$. 
The total number of comparisons is at most~$\sum_{i=1}^{\mcSize-1} (\mcSize-i)=O(\mcSize^2)$, and every comparison is between two nonempty proper subsets of~$V$. 
\end{proof}

\subsection{Minimization over a downward-closed family}\label{subsec:minimization-over-a-downward-closed-family} 
Recall that a family~${\cI\subseteq2^V}$ is \emph{downward closed} if, for all~${S\in\cI}$, every subset of~$S$ is also in~$\cI$. 
Fix a symmetric submodular function~${\sfmFunction:2^V\to\R}$ and a downward-closed family~${\cI\subseteq\downwardClosedDomain}$ with~${|\cI|\geq 2}$. 
We now explain how to find a minimizer of~$f$ over~${\cI\setminus\{\emptyset\}}$ using comparisons. 
At a high level, we adapt~\cite[Algorithm~2]{GoemansSoto} to use the~$\mcRoot$-rooted MC ordering above so each contraction step can be implemented using comparisons; the full algorithm is given in~\Cref{alg:downward-closed-sfm}.

Fix a dummy element~${\sfmDummyElement\notin V}$, and let~${\sfmDummyV\coloneqq V \cup\{\sfmDummyElement\}}$. 
We model contracting the ground set by partitioning~${\sfmDummyV}$ into \emph{blocks}. 
Each block is a nonempty set~${B\subseteq \sfmDummyV}$, and the block is \emph{feasible} if~${B\in\cI}$. 
The algorithm maintains a \emph{root block}~$\sfmRootBlock$ with~${\sfmDummyElement \in \sfmRootBlock}$ and a collection~${\sfmActiveBlocks\subseteq\cI\setminus\{\emptyset\}}$ of \emph{active blocks} that partition~${V\setminus\sfmRootBlock}$.  
The dummy element~$\sfmDummyElement$ ensures that~$\sfmRootBlock$ is always nonempty so that we can use an~$\sfmRootBlock$-rooted MC ordering in each iteration. This is only a bookkeeping device: one could instead omit~$\sfmDummyElement$ and use an ordinary MC~ordering while~${\sfmRootBlock=\emptyset}$, then switch to an~$\sfmRootBlock$-rooted MC ordering as soon as~$\sfmRootBlock$ becomes nonempty.

The collection~${\{\sfmRootBlock\}\cup\sfmActiveBlocks}$ is a partition of~$\sfmDummyV$, and each set in this partition represents an element of the contracted ground set.
Define the \emph{contracted function}~${\sfmContractedFunction:2^{\{\sfmRootBlock\}\cup\sfmActiveBlocks}\to\R}$ by
\begin{align}
\label{eq:sfm-contracted-function-definition}
\sfmContractedFunction(\mathcal S)
\;\coloneqq\;
\sfmFunction\left(\left(\bigcup_{B\in\mathcal S}B\right)\setminus\{\sfmDummyElement\} \right)
\qquad\text{for every }\mathcal S\subseteq\{\sfmRootBlock\}\cup\sfmActiveBlocks.
\end{align} 
Since~$f$ is submodular and symmetric, the contracted function~$\sfmContractedFunction$ is also submodular and symmetric. 
We can simulate a comparison oracle for~$\sfmContractedFunction$ using a comparison oracle for~$\sfmFunction$ by comparing the corresponding unions of blocks after deleting~$\sfmDummyElement$.  
Thus by~\Cref{lemma:computing-mc-ordering}, given~${\sfmRootBlock}$ and~${\sfmActiveBlocks}$, we can compute an~$\sfmRootBlock$-rooted MC ordering for~${\sfmContractedFunction}$ using~${O(\mcSize^2)}$ comparisons to~$\sfmFunction$, without making any comparisons involving~$\emptyset$ or~$V$.

In~\Cref{alg:downward-closed-sfm}, we initially take~$\sfmActiveBlocks$ to be the set of all singletons in~$\cI$. 
There exists at least one such singleton because~$\cI$ is downward closed and has cardinality~${|\cI|\geq 2}$. 
The algorithm maintains a set~${\sfmBestSet\in\cI\setminus\{\emptyset\}}$ that achieves the smallest feasible value of~${\sfmFunction}$ seen so far, and repeatedly contracts the ground set. 

In each iteration, we use the simulation described above to construct an~$\sfmRootBlock$-rooted MC ordering~${(\sfmRootBlock,B_1,\ldots,B_{\sfmBlockCount})}$ for the contracted function~${\sfmContractedFunction}$, where~${\sfmBlockCount=|\sfmActiveBlocks|}$. 
If~${\sfmMergedBlock\in\cI}$, then we \emph{contract} these two blocks together by replacing them by their union. If~${\sfmFunction(\sfmMergedBlock)<\sfmFunction(\sfmBestSet)}$, then we also update~${\sfmBestSet}$ to~${\sfmMergedBlock}$; this requirement for a strict improvement ensures that the final set~$\sfmBestSet$ is inclusionwise minimal. 
Otherwise, if~${\sfmMergedBlock\notin\cI}$, then we discard the blocks by moving them into~$\sfmRootBlock$. 
We prove that the algorithm satisfies the properties of~\Cref{thm:intro-constrained} below.

\begin{algorithm}[H]%[t]
\caption{Symmetric SFM over a downward-closed family}\label{alg:downward-closed-sfm} 

\algrenewcommand\algorithmicensure{\textbf{Returns:}}

\begin{algorithmic}[1]
\addtolength{\itemsep}{0.05ex} % Slightly increase line spacing

\Require A comparison oracle for a symmetric submodular function~${\sfmFunction:2^V\to\R}$, and a membership oracle for a downward-closed family~${\cI\subseteq\downwardClosedDomain}$ with~${|\cI|\geq 2}$.

\Ensure An inclusionwise-minimal minimizer~$\sfmBestSet$ of~$\sfmFunction$ over~${\cI\setminus\{\emptyset\}}$.

\State Initialize the active blocks~${\sfmActiveBlocks\gets\bigl\{\{v\}\setSeparator v\in V\text{ and }\{v\}\in\cI\bigr\}}$\label{line:sfm-initialize-active}

\State Fix a dummy element~${\sfmDummyElement\notin V}$, and initialize the root block~${\sfmRootBlock\gets\{\sfmDummyElement\}\cup\left(V\setminus\bigcup_{A\in\sfmActiveBlocks}A\right)}$

\State Choose~${\sfmBestSet\in\argmin_{A\in\sfmActiveBlocks}\sfmFunction(A)}$\label{line:sfm-initialize-best}

\While{${\sfmBlockCount\coloneqq |\sfmActiveBlocks|\geq 2}$}

  \State Construct an~${\sfmRootBlock}$-rooted MC ordering~${(\sfmRootBlock,B_1,\ldots,B_{\sfmBlockCount})}$ for the contracted function~$\sfmContractedFunction$\label{line:sfm-mc-ordering}
  
  \State ${\sfmActiveBlocks\gets\sfmActiveBlocks\setminus\{\sfmPenultimateBlock,\sfmLastBlock\}}$\label{line:remove-blocks-from-active-blocks}
  
  \If{${\sfmMergedBlock\in\cI}$}\label{line:sfm-membership}
  
    \State ${\sfmActiveBlocks\gets\sfmActiveBlocks\cup\{\sfmMergedBlock\}}$\label{line:sfm-merge} \Comment{Contract~$\sfmPenultimateBlock$ and~$\sfmLastBlock$ in~$\sfmActiveBlocks$}
    
    \If{${\sfmFunction(\sfmMergedBlock)<\sfmFunction(\sfmBestSet)}$}\label{line:sfm-improvement}
    
      \State ${\sfmBestSet\gets\sfmMergedBlock}$\label{line:sfm-update-best} \Comment{Update the best candidate minimizer}
      
    \EndIf
  \Else
    \State ${\sfmRootBlock\gets\sfmRootBlock\cup\sfmMergedBlock}$\label{line:sfm-discard} \Comment{If~${\sfmMergedBlock\notin\cI}$, then discard the blocks}
  \EndIf
\EndWhile
\State \Return $\sfmBestSet$
\end{algorithmic} 
\end{algorithm}

At the start of each iteration of the while loop in~\Cref{alg:downward-closed-sfm}, and when the algorithm terminates,~${\{\sfmRootBlock\}\cup\sfmActiveBlocks}$ is a partition of~$\sfmDummyV$.
Therefore the contracted function~$\sfmContractedFunction$ used in each iteration is symmetric and submodular.

Also, at the start of each iteration and when the algorithm terminates, we have~${\sfmActiveBlocks\subseteq\cI\setminus\{\emptyset\}}$.
Since~$\sfmBestSet$ is initially chosen from~$\sfmActiveBlocks$ and every subsequent assignment~${\sfmBestSet\gets\sfmMergedBlock}$ requires that~${\sfmMergedBlock\in\cI}$, the returned set also satisfies~${\sfmBestSet\in\cI\setminus\{\emptyset\}}$. 
We prove below that this set is an inclusionwise-minimal minimizer over~${\cI\setminus\{\emptyset\}}$.

\begin{lemma}[The stored set~$\sfmBestSet$ is no worse than any active block] 
\label{lemma:stored-set}
At the start of each iteration of the while loop in~\Cref{alg:downward-closed-sfm}, and when the algorithm terminates,~${\sfmFunction(\sfmBestSet)\leq\sfmFunction(A)}$ for every~${A\in\sfmActiveBlocks}$.  
\end{lemma}

\begin{proof} 
Initially, we choose~${\sfmBestSet\in\argmin_{A\in\sfmActiveBlocks}\sfmFunction(A)}$, so the claim holds.
Suppose that the claim holds at the start of an iteration of the while loop.
During this iteration,~${\sfmFunction(\sfmBestSet)}$ does not increase, as we update~$\sfmBestSet$ only if the strict inequality in~\Cref{line:sfm-improvement} holds.
Thus, after the iteration, we still have~${\sfmFunction(\sfmBestSet)\leq\sfmFunction(A)}$ for every~$A$ that belonged to~${\sfmActiveBlocks\setminus\{\sfmPenultimateBlock,\sfmLastBlock\}}$ at the start of the iteration.

If~${\sfmMergedBlock\in\cI}$, then we add~$\sfmMergedBlock$ to~$\sfmActiveBlocks$, and the comparison in~\Cref{line:sfm-improvement} ensures that~${\sfmFunction(\sfmBestSet)\leq\sfmFunction(\sfmMergedBlock)}$ at the end of the iteration.
Otherwise, no set is added to~$\sfmActiveBlocks$.
Therefore the claim holds after the iteration in either case.
\end{proof}

\begin{lemma}[Sets better than~$\sfmBestSet$ are unions of active blocks] 
\label{lemma:downward-invariant} 
At the start of each iteration of the while loop in~\Cref{alg:downward-closed-sfm}, and when the algorithm terminates, every set~${S\in\cI\setminus\{\emptyset\}}$ with~${\sfmFunction(S)<\sfmFunction(\sfmBestSet)}$ is a union of sets in~$\sfmActiveBlocks$.
\end{lemma}

\begin{proof} 
We prove the claim by induction on the iterations.
Before the first iteration, we initialize~$\sfmActiveBlocks$ to contain every singleton in~$\cI$.
Therefore, as~$\cI$ is downward closed, every~${S\in\cI\setminus\{\emptyset\}}$ is a union of singletons in~$\sfmActiveBlocks$, and the claim holds. 

Suppose that the claim holds at the start of an iteration, and let~$\sfmBestSet$ be the currently stored set, $\sfmActiveBlocks$ the collection of active blocks, and~$\sfmRootBlock$ the root block. 
Let~$\sfmBestSet'$ be the stored set at the end of the iteration.  
Fix~${S\in\cI\setminus\{\emptyset\}}$ with~${\sfmFunction(S)<\sfmFunction(\sfmBestSet')}$.
Since the value of the stored set does not increase during the iteration,~${\sfmFunction(S)< \sfmFunction(\sfmBestSet')\leq\sfmFunction(\sfmBestSet)}$.  

Let~${\mathcal S\coloneqq\{A\in\sfmActiveBlocks\setSeparator A\subseteq S\}}$ be the set of blocks contained in~$S$. 
Then~${\sfmRootBlock\notin\mathcal S}$ and, by the inductive hypothesis,~$S$ is a union of sets in~$\sfmActiveBlocks$, so~${S=\bigcup_{A\in\mathcal S}A}$.  
Let~${(\sfmRootBlock,B_1,\ldots,B_{\sfmBlockCount})}$ be the~$\sfmRootBlock$-rooted MC ordering computed in~\Cref{line:sfm-mc-ordering}, and suppose for contradiction that~${|\mathcal S\cap\{\sfmPenultimateBlock,\sfmLastBlock\}|=1}$.
Then applying~\Cref{cor:rooted-mc} to~$\sfmContractedFunction$ and~$\mathcal S$ gives
\[
2\sfmContractedFunction(\mathcal S)
\;\geq\;
\sfmContractedFunction(\{\sfmLastBlock\})
+\min_{A\in\mathcal S}\sfmContractedFunction(\{A\}).
\]
Thus, by the definition of~$\sfmContractedFunction$ in~\eqref{eq:sfm-contracted-function-definition}, we have~${2\sfmFunction(S) \geq \sfmFunction(\sfmLastBlock) +\min_{A\in\mathcal S}\sfmFunction(A) \geq 2\sfmFunction(\sfmBestSet)}$, where the second inequality follows by~\Cref{lemma:stored-set}.
This contradicts our assumption that~${\sfmFunction(S)<\sfmFunction(\sfmBestSet)}$, so it follows that~${|\mathcal S\cap\{\sfmPenultimateBlock,\sfmLastBlock\}|\neq 1}$.

Now let~$\sfmActiveBlocks'$ be the collection of active blocks at the end of the iteration.
First, suppose that~${|\mathcal S\cap\{\sfmPenultimateBlock,\sfmLastBlock\}|=0}$.
Since the iteration removes only~$\sfmPenultimateBlock$ and~$\sfmLastBlock$ from~$\sfmActiveBlocks$, we have~${\mathcal S\subseteq\sfmActiveBlocks'}$, so~$S$ is a union of sets in~$\sfmActiveBlocks'$. 
Otherwise, suppose that~${|\mathcal S\cap\{\sfmPenultimateBlock,\sfmLastBlock\}|=2}$.
Then~${\sfmMergedBlock\subseteq S}$.
Since~${S\in\cI}$ and~$\cI$ is downward closed, we have~${\sfmMergedBlock\in\cI}$.
Hence~${\sfmMergedBlock\in\sfmActiveBlocks'}$ by~\Cref{line:sfm-merge}.
We also have~${\mathcal S\setminus\{\sfmPenultimateBlock,\sfmLastBlock\}\subseteq\sfmActiveBlocks'}$, since no other blocks are removed. 
Thus~$S$ is a union of sets in~$\sfmActiveBlocks'$.
The claim therefore holds at the end of the iteration, completing the induction.
\end{proof}

\DownwardClosedSFM*

\begin{proof} 
Each iteration of~\Cref{alg:downward-closed-sfm} decreases~${|\sfmActiveBlocks|}$ by at least one.
Since initially~${|\sfmActiveBlocks|\leq n}$, the algorithm terminates after at most~${n-1}$ iterations.
Let~$\sfmBestSet'$ be the returned set and~$\sfmActiveBlocks'$ the collection of active blocks at termination.
Then~${|\sfmActiveBlocks'|\leq1}$ and~${\sfmBestSet'\in\cI\setminus\{\emptyset\}}$. 
Suppose for contradiction that some~${S\in\cI\setminus\{\emptyset\}}$ satisfies~${\sfmFunction(S)<\sfmFunction(\sfmBestSet')}$.
By~\Cref{lemma:downward-invariant},~$S$ is a union of sets in~$\sfmActiveBlocks'$.
Since~${S\neq\emptyset}$ and~${|\sfmActiveBlocks'|\leq1}$, we have~${\sfmActiveBlocks'=\{S\}}$.
But~\Cref{lemma:stored-set} gives~${\sfmFunction(\sfmBestSet')\leq\sfmFunction(S)}$, a contradiction.
Therefore~$\sfmBestSet'$ minimizes~$\sfmFunction$ over~${\cI\setminus\{\emptyset\}}$.

We next prove that~$\sfmBestSet'$ is inclusionwise minimal.
If~${|\sfmBestSet'|=1}$, then the claim holds.
Otherwise, since the initially stored set is a singleton, consider the last iteration that updates the stored set.
Let~$\sfmBestSet$ be the stored set and~$\sfmActiveBlocks$ the collection of active blocks at the start of this iteration, and let~$\sfmPenultimateBlock$ and~$\sfmLastBlock$ be the two blocks merged in this iteration.
Then~${\sfmBestSet'=\sfmMergedBlock}$ and~${\sfmFunction(\sfmBestSet')<\sfmFunction(\sfmBestSet)}$ by~\Cref{line:sfm-improvement,line:sfm-update-best}. 
Suppose for contradiction that some nonempty~${X\subsetneq\sfmBestSet'}$ satisfies~${\sfmFunction(X)=\sfmFunction(\sfmBestSet')}$.
Since~${\sfmBestSet'\in\cI}$ and~$\cI$ is downward closed, we have~${X\in\cI}$.
As~${\sfmFunction(X)=\sfmFunction(\sfmBestSet')<\sfmFunction(\sfmBestSet)}$, it follows by~\Cref{lemma:downward-invariant} that~$X$ is a union of sets in~$\sfmActiveBlocks$.
Since~${\sfmBestSet'=\sfmMergedBlock}$ and the sets in~$\sfmActiveBlocks$ are pairwise disjoint, only~$\sfmPenultimateBlock$ and~$\sfmLastBlock$ can appear in this union, so it follows that~${X\in\{\sfmPenultimateBlock,\sfmLastBlock\}}$.
Applying~\Cref{lemma:stored-set} at the start of this iteration gives~${\sfmFunction(\sfmBestSet)\leq\sfmFunction(X)}$, a contradiction.
Thus~$\sfmBestSet'$ is inclusionwise minimal.

For the query complexity, the initialization before the while loop uses~${O(n)}$ comparisons and~${O(n)}$ membership queries. 
Each of the~${O(n)}$ iterations computes a rooted MC ordering on at most~${n+1}$ blocks using~${O(n^2)}$ comparisons by~\Cref{lemma:computing-mc-ordering}.  
It also makes one membership query and at most one further comparison. 
Thus the algorithm uses~${O(n^3)}$ comparisons and~${O(n)}$ membership queries in total, and all operations take polynomial time. 
\end{proof}

Setting~${\cI\coloneqq2^V\setminus\{V\}}$ in~\Cref{thm:intro-constrained} gives the following result.

\SymmetricSFM*

The~${O(n^3)}$ comparison bound in~\Cref{thm:intro-sfm} matches the best-known deterministic value-oracle bound for symmetric~SFM~\cite{Queyranne,IwataKonno}.
Since two value-oracle queries simulate one comparison, an~${o(n^3)}$-query comparison algorithm for symmetric~SFM would also give an~${o(n^3)}$-query value-oracle algorithm. 
In the next section, we consider the special case of weighted minimum cut.

%%%%%%%%%%%%%%%%%%%
\section{Weighted minimum-cut upper bounds}
\label{sec:bounded}
%%%%%%%%%%%%%%%%%%%

The cut function of an undirected graph with nonnegative edge weights is symmetric and submodular. 
We can therefore use~\Cref{thm:intro-sfm} to find a minimum cut in polynomial time using~${O(n^3)}$ comparisons, even when the edge weights are unbounded. 
We now show how to improve this bound for graphs whose edge weights are bounded integers.
 
Throughout this section, let~${G=(V,w)}$ be an undirected graph on~$n$ vertices with edge weights in~${\weightedEdgeWeights}$, where~$\B$ is a known positive integer. 
An edge that is not present in the graph has weight~$0$.
Let~$m$ be the number of positive-weight edges in~$G$, and let~${\totalWeight\coloneqq\sum_e\edgeWeightE{e}}$ be their total weight. 
The parameters~$m$ and~$\totalWeight$ are used only in the analysis and need not be known by the algorithm.

In~\Cref{sec:online}, we show how to compute a minimum cut using~${ O\!\left((n^2+\totalWeight)\log(\totalWeight+2)\right)}$ comparisons. Since the total edge weight in~$G$ is at most~${m\B}$, this gives an upper bound of~${\softO(n^2+m\B)}$ comparisons. 
In~\Cref{sec:support}, we use an expander decomposition of~$G$ to create a contracted graph with total edge weight~${\softO(n\B^2)}$ while preserving the nonsingleton minimum cuts of~$G$. 
Applying the algorithm from~\Cref{sec:online} to the contracted graph gives an upper bound of~${\softO(n^2+n\B^2)}$ comparisons, which is stronger when~$G$ is dense and~${\B=o(n)}$. 

Combining both bounds proves~\Cref{thm:bounded-weight-min-cut-upper-bound}: we can compute a weighted minimum cut using
\[
\weightedMinCutUpperBound
\]
comparisons and polynomial time.

\subsection{Weighted minimum cut in graphs with bounded total weight}
\label{sec:online} 
The following algorithm is an online form of the Global Subspace Learning~(GSL) algorithm of \mbox{Cohen-Addad}~et~al.~\cite[Theorem~4.2]{comparison}. 
We first prove the result for a general linear function, then specialize to weighted cut functions below.

Let~${\gslDimension\in\N_{\geq 1}}$, and let~${\cP\subseteq\{0,1\}^{\gslDimension}}$ be nonempty. 
Suppose that we are given comparison-query access to a function~${f:\cP\to\R}$ defined by~${f(x)\coloneqq\langle \gslWeight,x\rangle}$, where~${\gslWeight\in\R^{\gslDimension}}$ is unknown and~${\innerProd{\gslWeight}{x}=\gslWeight^\top x}$ is the standard inner product. 
For any vectors~${x,y\in\cP}$, if~${f(x)=f(y)}$, then~${\innerProd{\gslWeight}{x-y}=0}$. 
Thus, if a comparison query returns equality, then the difference of the queried points is orthogonal to~$\gslWeight$. 

The~GSL algorithm uses this observation to sort the vectors~${x\in\cP}$ in increasing order of the values~${f(x)}$ using only~${O((\gslDimension+\gslDistinctValues)\log( \gslDistinctValues+1))}$ comparisons, where~$\gslDistinctValues$ is the number of distinct values taken by~$f$ on~$\cP$~\cite[Theorem~4.2]{comparison}.  
However, the implementation in~\cite[Section~4.1]{comparison} processes the entire set~$\cP$, which can take time exponential in~$\gslDimension$.  

In the next result, we adapt this subspace-learning approach to \emph{simulate another algorithm}. 
The key modification is that the simulation only processes vectors that are actually queried by the algorithm. 
Hence, if the original algorithm makes~$\poly(\gslDimension)$ queries, then the simulation increases the running time by at most a polynomial in~$\gslDimension$. 
In~\Cref{cor:weighted-min-cut-total-weight}, we use this result to solve weighted minimum cut in polynomial time by simulating the~\mbox{$O(n^3)$-comparison} algorithm from~\Cref{subsec:minimization-over-a-downward-closed-family}.

\begin{theorem}[Simulating comparison-query access to a linear function]
\label{thm:online}
Let~${\gslDimension\in\N_{\geq 1}}$, and let~${\cP\subseteq\{0,1\}^{\gslDimension}}$ be nonempty. 
Let~${\gslWeight\in\R^{\gslDimension}}$ be unknown, and define~${f:\cP\to\R}$ by~${f(x)\coloneqq\langle \gslWeight,x\rangle}$. 
Let~${\gslDistinctValues\coloneqq\bigl|\{f(x):x\in\cP\}\bigr|}$ be the number of distinct values taken by~$f$ on~$\cP$.

Suppose that an adaptive algorithm~$\alg$ has comparison-query access to~$f$.
Then we can simulate~$\alg$ using
\[
  O\bigl(  (\gslDimension+\gslDistinctValues)\log(\gslDistinctValues+1) \bigr)
\]
comparisons. 
The simulation does not need to know~$\gslDistinctValues$. 
If~$\alg$ makes~$q$ comparisons when run on~$f$, then the simulation requires~${\poly(\gslDimension,q)}$ additional time beyond the running time of~$\alg$.
\end{theorem}

\begin{proof}
For each distinct value of~$f$ that we have discovered, store one representative point attaining that value. Label these representatives~${r_1,r_2,r_3,\ldots}$, ordered so that~${f(r_1)<f(r_2)<f(r_3)<\ldots}$. 
Since~$f$ takes at most~$\gslDistinctValues$ distinct values, there are at most~$\gslDistinctValues$ such representatives. 
We also maintain a subspace~${\gslSubspace\subseteq\gslWeight^\perp}$ orthogonal to~$\gslWeight$. 

The simulation works as follows. Initially,~${\gslSubspace=\{0\}}$. 
When a point~${x\in\cP}$ first appears in a comparison query, we check whether~${x-r_i\in \gslSubspace}$ for some~$i$. If~${x-r_i\in \gslSubspace}$, then~${f(x)=f(r_i)}$, so the simulation knows that this equality holds without making a comparison query.

Otherwise, we use binary search to determine where~$f(x)$ lies among the ordered values~${f(r_1),f(r_2),f(r_3),\ldots}$, using~${O(\log(\gslDistinctValues+1))}$ comparison queries.
If~$f(x)$ differs from every previously discovered value, then we insert~$x$ into the list of representatives in the appropriate position.
Otherwise,~${f(x)=f(r_i)}$ for some~$i$, so~${x-r_i\in\gslWeight^\perp}$.  
We therefore replace~$\gslSubspace$ by the span of~${\gslSubspace\cup\{x-r_i\}}$.
Since we already checked that~${x-r_i\notin \gslSubspace}$, this increases~${\dim \gslSubspace}$ by one.

Thus, whenever a binary search is required, either we add a new representative or increase~${\dim \gslSubspace}$ by one. We can add at most~$\gslDistinctValues$ representatives, and~${\dim \gslSubspace\le\gslDimension}$, so binary search is required for at most~${\gslDimension+\gslDistinctValues}$ points. Since each binary search uses~${O(\log(\gslDistinctValues+1))}$ comparison queries, the total number of comparison queries is~${O\!\left((\gslDimension+\gslDistinctValues)\log(\gslDistinctValues+1)\right)}$.

It remains to bound the additional running time used by the simulation.  
Every vector added to~$\gslSubspace$ is the difference of two points in~${\{0,1\}^{\gslDimension}}$, so it lies in ${\{-1,0,1\}^{\gslDimension}}$.
We can therefore maintain~$\gslSubspace$ and test membership in~$\gslSubspace$ in time polynomial in~$\gslDimension$ and in the number of distinct queried points.   
\end{proof}

We now use~\Cref{thm:online} to solve weighted minimum cut by simulating the symmetric~SFM algorithm from~\Cref{sec:symmetric}.

\begin{corollary}[Total-weight bound for minimum cut]
\label{cor:weighted-min-cut-total-weight}
Let~${G=(V,w)}$ be an undirected graph on~${n\geq 2}$ vertices with nonnegative integer edge weights, and let~$\totalWeight$ be its total edge weight. 
There is a deterministic polynomial-time algorithm that finds a minimum cut of~$G$ using
\[
 O\!\left((n^2+\totalWeight)\log(\totalWeight+2)\right)
\]
cut comparisons. The algorithm does not need to know~$\totalWeight$.
\end{corollary}

\begin{proof} 
Let~${\gslDimension=\binom{n}{2}}$. Fix an arbitrary ordering of the unordered pairs of vertices, and define~${\gslWeight\in\R^{\gslDimension}}$ so that the coordinate corresponding to~${\{u,v\}}$ is~${\edgeWeight{u}{v}}$. 
For each~${S\subseteq V}$, let~${\cutIndicator{S}\in\{0,1\}^{\gslDimension}}$ be the indicator vector of which edges cross the cut~$S$. Then~${\cutGraphSet{G}{S}=  \innerProd{\gslWeight}{\cutIndicator{S}}}$.

We can therefore use~\Cref{thm:online} to simulate the symmetric~SFM algorithm from~\Cref{thm:intro-sfm} on~$\cutGraph{G}$.
Every cut value is in~${\{0,1,\ldots,\totalWeight\}}$, so~${\gslDistinctValues \leq \totalWeight+1}$ and the simulation makes 
\[
O\bigl((\gslDimension+\gslDistinctValues)\log(\gslDistinctValues+1)\bigr)
\;=\;
O\bigl((n^2+\totalWeight)\log(\totalWeight+2)\bigr)
\] 
comparisons. 
The algorithm from~\Cref{thm:intro-sfm} uses~${O(n^3)}$ queries and~$\poly(n)$ time, so the simulation also runs in~$\poly(n)$ time. 
\end{proof}

If every edge weight is at most~$\B$, then the total graph weight satisfies~${\totalWeight\leq m\B}$, and~\Cref{cor:weighted-min-cut-total-weight} gives an~${\softO(n^2+m\B)}$ comparison bound for finding a minimum cut. The dependence on~$m$ is useful when~$G$ is sparse, but can be quadratic in~$n$ when~$G$ is dense. We next show how to replace the~${m\B}$ term by~${n\B^2}$, which improves the query complexity when~$m$ is large and~$\B=o(n)$.

\subsection{Weighted minimum cut in dense graphs}\label{sec:support}

In this subsection, we show that we can compute a weighted minimum cut using~${\softO(n^2+n\B^2)}$ comparisons and polynomial time.
The main step is replacing~$G$ by a contracted graph with total edge weight~${\softO(n\B^2)}$ while preserving the nonsingleton minimum cuts. The result then follows by applying~\Cref{cor:weighted-min-cut-total-weight} to the contracted graph.

For a vertex~${v\in V}$, denote its weighted degree by ${\degreeGraph{v}{G}\coloneqq\cutGraphSet{G}{\{v\}}}$.
For ${S\subseteq V\setminus\{v\}}$, write~${\weightVertexSet{v}{S}\coloneqq\sum_{x\in S}\edgeWeight{v}{x}}$ for the sum of edge weights between~$v$ and~$S$.

We use the following standard identity for weighted graph cuts.  
If~$S$ and~${S\cup\{v\}}$ are nonempty proper subsets of~$V$, then this allows us to determine whether~${\weightVertexSet{v}{S} < \degreeGraph{v}{G}/2}$ using one comparison query, which we use to construct the contracted graph.

\begin{lemma}[Cut marginal identity]
\label{lemma:cut-marginal} 
Let~${G=(V,w)}$ be an undirected weighted graph, and let~${S\subseteq V}$. For each~${v\in V\setminus S}$,
\[
 \cutGraphSet{G}{S\cup\{v\}}-\cutGraphSet{G}{S}
 \;=\;
 \degreeGraph{v}{G}-2\weightVertexSet{v}{S}.
\] 
\end{lemma}

\begin{proof}
Fix~${v\in V\setminus S}$. 
When~$v$ is added to~$S$, the edges from~$v$ to~$S$ stop crossing the cut, decreasing the weight by~${\weightVertexSet{v}{S}}$, while the edges from~$v$ to~${V\setminus(S\cup\{v\})}$ start crossing the cut, increasing the weight by~${\degreeGraph{v}{G}-\weightVertexSet{v}{S}}$. Adding these contributions gives the result.
\end{proof}

If there is a vertex~${u}$ with weighted degree~${\degreeGraph{u}{G}=0}$, then the singleton~${\{u\}}$ is a minimum cut; we call such a vertex~\emph{isolated}. 
If we could compare~${\cutGraphSet{G}{\{u\}}}$ with~${\cutGraphSet{G}{\emptyset}=0}$, then we could decide whether~$u$ is isolated using one comparison. 
Even when trivial comparisons are not allowed, the following observation of Cohen-Addad et al.~\cite[Lemma~D.4]{comparison} allows us to decide whether~$u$ is isolated using three comparisons. 
The restriction to~${n\geq 4}$ is harmless for the asymptotic bound: if~${n<4}$, then we can compare all pairs of nontrivial cuts and output a minimum cut.

\begin{lemma}[Isolated vertex test]
\label{lemma:deciding-whether-min-degree-is-0-isolated}
Let~${G=(V,w)}$ be an undirected weighted graph on~${n\ge4}$ vertices with nonnegative edge weights, and let~${u \in V}$.
We can determine whether~${\degreeGraph{u}{G}=0}$ using three cut comparisons.  
\end{lemma}

\begin{proof} 
Choose distinct~${a,b\in V\setminus\{u\}}$, and query the pairs
\[
 \bigl(\{a,u\},\, \{a\}\bigr),
 \qquad
 \bigl(\{b,u\},\, \{b\}\bigr), 
 \qquad\text{and}\qquad
 \bigl(\{a,b,u\},\, \{a,b\}\bigr).
\]
These are valid cut comparisons because~${n\ge4}$. 
If~${\degreeGraph{u}{G}=0}$, then every edge incident to~$u$ has weight~$0$, so all three comparisons return equality.
Conversely, if all three comparisons return equality, then by~\Cref{lemma:cut-marginal},~${\edgeWeight{u}{a}=\edgeWeight{u}{b}=\edgeWeight{u}{a}+\edgeWeight{u}{b}=\degreeGraph{u}{G}/2}$, and it follows that~${\degreeGraph{u}{G}=0}$. 
Thus,~${\degreeGraph{u}{G}=0}$ if and only if all three comparisons return equality. 
\end{proof}

We can use~\Cref{lemma:deciding-whether-min-degree-is-0-isolated} to decide whether there is an isolated vertex in~$G$ using~${O(n)}$ comparisons. If so, then the singleton containing this vertex is a minimum cut. Otherwise, if no vertices are isolated, then we~\emph{contract}~$G$, in the sense explained below.

A \emph{partition}~$\partition$ of~$V$ is a collection~${\partition=(V_1,\ldots,V_{\partitionIndex})}$ of pairwise disjoint nonempty subsets of~$V$ whose union is~$V$. 
Each~$V_i$ is a~\emph{block} of the partition.
The~\emph{contracted graph}~$\contractedG$ contains one vertex for each block of the partition; we call these~\emph{supervertices} to distinguish them from the original vertices of~$G$.
For distinct blocks~$V_i$ and~$V_j$, the edge between the corresponding supervertices in~${\contractedG}$ has weight equal to the total weight between~$V_i$ and~$V_j$ in~$G$. 
Edges with both endpoints in the same block are removed. 

Even though we do not know the edge weights of~$G$, we can simulate comparison-query access to~${\cutGraph{\contractedG}}$ as follows.

\begin{remark}[Querying the contracted graph]\label{rem:querying-cut-function-of-contracted-graph}
Given a known partition~$\partition=(V_1,\ldots,V_{\partitionIndex})$, we can simulate a comparison query to~${\cutGraph{\contractedG}}$ using one comparison query to~${\cutGraph{G}}$, without knowing any edge weights.
This holds because a cut in~$\contractedG$ corresponds exactly to the cut in~$G$ obtained by taking the union of the blocks on each side of the cut. 
\end{remark}

Next, we construct a partition~$\partition$ such that contracting the graph according to this partition does not destroy the minimum-cut structure of~$G$. 
The total edge weight of~$G$ can be as large as~${O(n^2 \B)}$, but the following result ensures that the contracted graph~$\contractedG$ has total edge weight bounded by~${\softO(n \B^2)}$, even though we do not know any individual edge weights. This allows us to reduce the quadratic dependence on~$n$ when~${\B=o(n)}$.

\begin{restatable}[Minimum-cut-preserving contraction]{lemma}{GraphContractionUsingComparisons}
\label{lemma:contraction}
Let~${G=(V,w)}$ be an undirected weighted graph on~${n\geq 2}$ vertices with edge weights in~${\weightedEdgeWeights}$ for a known parameter~$\B$. 
Suppose that~${\minimumGraphDegree{G}>0}$. 
Then using~${O(n^2)}$ comparisons and polynomial time, we can construct a partition~$\partition$ of~$V$ with the following properties.
\begin{enumerate}[label=(\roman*)] 
\item\label{item:contraction-preserves-min-cut} Either some singleton is a minimum cut of~$G$, or every minimum cut of~$G$ is a union of blocks of~$\partition$. 
\item\label{item:contraction-total-weight} The total edge weight of the contracted graph~${\contractedG}$ is~${\softO(n\B^2)}$.
\end{enumerate}
\end{restatable}
\begin{proof}
Deferred to~\Cref{app:weighted-graph-contraction}.
\end{proof}

We now use~\Cref{lemma:contraction} to create a contracted graph~$\contractedG$ of total edge weight~${\softO(n\B^2)}$, then apply the minimum-cut algorithm from~\Cref{cor:weighted-min-cut-total-weight}.

\begin{theorem}[Weighted minimum cut in dense graphs]
\label{thm:weighted-min-cut-nB2}
   Let~${G=(V,w)}$ be an undirected graph on~${n\geq 2}$ vertices with edge weights in~$\weightedEdgeWeights$ for a known parameter~$\B$. 
   There is a deterministic polynomial-time algorithm that finds a minimum cut of~$G$ using
\[
 \softO\!\left(n^2+nB^2\right)
\]
cut comparisons.
\end{theorem}

\begin{proof}
If~${n<4}$, then we can compare all of the nontrivial cuts using~${O(1)}$ comparisons and return a minimum cut. Suppose therefore that~${n \geq 4}$. 
First, find a vertex~$u$ of minimum weighted degree. 
This is possible using~${n-1}$ comparisons by comparing the singleton cuts one by one, each time keeping a vertex with the smallest cut value seen so far.

Next, use~\Cref{lemma:deciding-whether-min-degree-is-0-isolated} to check whether~${\degreeGraph{u}{G}=0}$. If so, then~$\{u\}$ is a minimum cut and we return~${\{u\}}$. Otherwise,~${\degreeGraph{u}{G}=\minimumGraphDegree{G}>0}$, so we can use~\Cref{lemma:contraction} to compute a partition~$\partition$ of~$V$. By~\Cref{lemma:contraction}\ref{item:contraction-preserves-min-cut}, either~${\{u\}}$ is a minimum cut of~$G$, or every minimum cut of~$G$ is a union of blocks of~$\partition$.  
If~${\partition=\{V\}}$, then the second alternative is impossible, so we return~${\{u\}}$.

Otherwise,~$\partition$ contains at least two blocks. 
In this case, run the weighted minimum-cut algorithm from~\Cref{cor:weighted-min-cut-total-weight} on the contracted graph~$\contractedG$, and let~${S\subseteq V}$ be the union of the blocks returned by the algorithm. 
By~\Cref{rem:querying-cut-function-of-contracted-graph}, we can implement each comparison in~$\contractedG$ using one comparison in~$G$, without knowing the edge weights of~$\contractedG$.

Finally, compare~${\cutGraphSet{G}{S}}$ with~${\cutGraphSet{G}{\{u\}}}$ and return whichever of~$S$ and~$\{u\}$ has the smaller cut value, breaking ties arbitrarily. 
By~\Cref{lemma:contraction}\ref{item:contraction-preserves-min-cut}, at least one of~${\{u\}}$ and~$S$ is a minimum cut, so the algorithm always outputs a minimum cut.

It remains to bound the number of comparisons.
We can find~$u$ in~${O(n)}$ comparisons, then check whether~${\degreeGraph{u}{G}=0}$ in three comparisons using~\Cref{lemma:deciding-whether-min-degree-is-0-isolated}.
Computing the partition~$\partition$ with~\Cref{lemma:contraction} uses~${O(n^2)}$ comparisons and polynomial time.
By~\Cref{lemma:contraction}\ref{item:contraction-total-weight}, the total edge weight of~$\contractedG$ is~${\softO(n\B^2)}$.
The graph~$\contractedG$ has at most~$n$ supervertices, so~\Cref{cor:weighted-min-cut-total-weight} uses~${\softO(n^2+n\B^2)}$ comparisons on~$\contractedG$.
All other operations use~${O(n^2)}$ comparisons and polynomial time, so the entire algorithm has the claimed query complexity and running time. 
\end{proof}

\subsection{Combining the bounds for sparse and dense graphs}

The upper bound in~\Cref{thm:bounded-weight-min-cut-upper-bound} is now almost immediate by combining the bounds above.

\BoundedWeightMinCutUpperBound*

\begin{proof} 
Run the algorithms from~\Cref{thm:weighted-min-cut-nB2,cor:weighted-min-cut-total-weight} in parallel, alternating oracle queries, and return the minimum cut found by whichever algorithm terminates first.
The respective algorithms run in polynomial time, so the combined algorithm runs in polynomial time. 
This parallel execution increases the smallest query count by at most a factor of two. 
The total edge weight in~$G$ satisfies~${\totalWeight\leq m\B}$, so the algorithm in~\Cref{cor:weighted-min-cut-total-weight} makes~${O\!\left((n^2+\totalWeight)\log(\totalWeight+2)\right)=\softO\!\left(n^2 + m\B\right)}$ comparisons. The algorithm in~\Cref{thm:weighted-min-cut-nB2} makes~${\softO(n^2+n\B^2)}$ comparisons, and the result follows by taking the minimum of these two values. 
\end{proof}

\section{Randomized minimum-cut lower bound}
\label{sec:lower}

In this section, we prove the following lower bound.

\randomizedlowerbound*

To prove~\Cref{thm:lower}, we encode a hidden ordering~${\tau=(v_1,\ldots,v_n)}$ of the vertices into the edge weights of a complete graph.
For each~$\tau$, we construct a base graph~$\lbBaseGraph{\tau}$ whose minimum cuts are precisely the prefixes~${\{v_1,\ldots,v_i\}}$~(up to complements).  
We then choose~${2\ell}$ path edges whose weights may be decreased by~$1$ to create a new unique minimum cut while leaving every comparison that was already strict unchanged. Finally, the lower bound follows by choosing a uniformly random ordering~$\tau$ and counting transcripts. 
The following consequences of~\Cref{thm:lower} are immediate.

\begin{corollary}[Minimum-cut lower bound with polynomial weights]
\label{cor:poly}
There is an absolute constant~${C>0}$ such that, for every sufficiently large even~$n$, every randomized comparison algorithm that outputs a minimum cut with probability at least~${2/3}$ on every complete~$n$-vertex graph with positive integer edge weights at most~${n^{C}}$ has worst-case expected query complexity~${\Omega(n\log\log n)}$.
\end{corollary}

\begin{proof}
Apply~\Cref{thm:lower} with~${\ell=\lceil\log n\rceil}$. 
\end{proof}

\begin{corollary}[Minimum-cut lower bound with arbitrarily small exponential weights]
\label{cor:epsilon} 
For every fixed~${\varepsilon\in(0,1)}$, there are constants~${c_\varepsilon,C_\varepsilon>0}$ such that, for every sufficiently large even~$n$, every randomized comparison algorithm that outputs a minimum cut with probability at least~${2/3}$ on every complete~$n$-vertex graph with positive integer edge weights at most~${2^{C_\varepsilon n^\varepsilon}}$ has worst-case expected query complexity at least~${c_\varepsilon n\log n}$. 
\end{corollary}

\begin{proof}
Apply~\Cref{thm:lower} with~${\ell=\lfloor(n/2)^\varepsilon\rfloor}$.
Then~${\log\ell=\Theta_\varepsilon(\log n)}$ and~${\lbWeightRange=2^{O_\varepsilon(n^\varepsilon)}}$, which gives the result.
\end{proof}

\subsection{Constructing the base graphs}
\label{subsec:path-construction} 
We first define the base graphs. 
Write~${n=2k}$, let~${V\coloneqq [n]}$, and fix a set~${\lbTargetCut\subseteq V}$ with~${|\lbTargetCut|=k}$.
Let~${\lbOrderings}$ be the family of orderings~${\tau=(v_1,\ldots,v_n)}$ of~$V$ in which~${(v_1,\ldots,v_k)}$ is an ordering of~${\lbTargetCut}$ and~${(v_{k+1},\ldots,v_n)}$ is an ordering of~${V\setminus\lbTargetCut}$.
The total number of such orderings is~${|\lbOrderings|=(k!)^2}$.

% Local helpers for this figure.
\newcommand{\lbFigureOrder}{\tau}
\newcommand{\lbFigureSet}{X}

% -----------------------------------------------------------------------------
% Repeated dimensions and style parameters.
% -----------------------------------------------------------------------------
\pgfmathsetmacro{\lbFigXStep}{1.32}
\pgfmathsetmacro{\lbFigYTop}{1.48}
\pgfmathsetmacro{\lbFigYBot}{0.00}
\pgfmathsetmacro{\lbFigDotR}{2.15}
\pgfmathsetmacro{\lbFigGuideTop}{2.20}
\pgfmathsetmacro{\lbFigGuideBot}{-0.40}
\pgfmathsetmacro{\lbFigSideLabelY}{2.42}
\pgfmathsetmacro{\lbFigTopLabelShift}{7.8}
\pgfmathsetmacro{\lbFigBotLabelShift}{-8.8}
\pgfmathsetmacro{\lbFigLeftLabelXShift}{-10.0}
\pgfmathsetmacro{\lbFigRightLabelXShift}{-9.5}
\pgfmathsetmacro{\lbFigEdgeLabelYShift}{-0.5}
\pgfmathsetmacro{\lbFigBaseLineWidth}{1.0}
\pgfmathsetmacro{\lbFigGuideLineWidth}{1.2}
\pgfmathsetmacro{\lbFigBoundaryEdgeWidth}{1.45}
\pgfmathsetmacro{\lbFigBoundaryCurveWidth}{1.30}
\pgfmathsetmacro{\lbFigCurveOffset}{0.34}
\pgfmathsetmacro{\lbFigCurveTailDX}{0.92}
\pgfmathsetmacro{\lbFigCurveTailY}{2.02}

\definecolor{lbFigureBoundaryPurple}{RGB}{122,50,190}
\colorlet{lbFigurePartitionGray}{black!62}
\colorlet{lbFigureMutedGraphGray}{black!28}

\newcommand{\lbFigureEdgeLabelFont}{\small}
\newcommand{\lbFigureVertexLabelFont}{\normalsize}

\tikzset{
  lbPartitionGuide/.style={lbFigurePartitionGray, dotted, line width=\lbFigGuideLineWidth pt},
  lbPartitionLabel/.style={black, font=\normalsize},
  lbBaseEdge/.style={draw=black, line width=\lbFigBaseLineWidth pt, line cap=round},
  lbMutedEdge/.style={draw=lbFigureMutedGraphGray, line width=\lbFigBaseLineWidth pt, line cap=round},
  lbBoundaryEdge/.style={draw=black, line width=\lbFigBoundaryEdgeWidth pt, line cap=round},
  lbBoundaryCurve/.style={draw=lbFigureBoundaryPurple, line width=\lbFigBoundaryCurveWidth pt, line cap=round, line join=round},
  lbDarkVertex/.style={circle, fill=black, inner sep=\lbFigDotR pt},
  lbMutedVertex/.style={circle, fill=lbFigureMutedGraphGray, inner sep=\lbFigDotR pt},
  lbVertexLabelTop/.style={font=\lbFigureVertexLabelFont, inner sep=0pt, yshift=\lbFigTopLabelShift pt},
  lbVertexLabelBot/.style={font=\lbFigureVertexLabelFont, inner sep=0pt, yshift=\lbFigBotLabelShift pt},
  lbLeftLabel/.style={font=\lbFigureEdgeLabelFont, inner sep=0pt, xshift=\lbFigLeftLabelXShift pt, yshift=\lbFigEdgeLabelYShift pt},
  lbRightLabel/.style={font=\lbFigureEdgeLabelFont, inner sep=0pt, xshift=\lbFigRightLabelXShift pt, yshift=\lbFigEdgeLabelYShift pt}
}

\newcommand{\lbFigurePathCoordinates}{%
  \coordinate (v1) at (0,\lbFigYBot);
  \coordinate (v2) at (\lbFigXStep,\lbFigYTop);
  \coordinate (v3) at (2*\lbFigXStep,\lbFigYBot);
  \coordinate (v4) at (3*\lbFigXStep,\lbFigYTop);
  \coordinate (v5) at (4*\lbFigXStep,\lbFigYBot);
  \coordinate (v6) at (5*\lbFigXStep,\lbFigYTop);
  \pgfmathsetmacro{\lbFigureMidGuide}{2.5*\lbFigXStep}
  \pgfmathsetmacro{\lbFigureSLabelX}{\lbFigXStep}
  \pgfmathsetmacro{\lbFigureTLabelX}{4*\lbFigXStep}
}

\newcommand{\lbFigurePartitionGuide}{%
  \draw[lbPartitionGuide] (\lbFigureMidGuide,\lbFigGuideBot) -- (\lbFigureMidGuide,\lbFigGuideTop);
  \node[lbPartitionLabel] at (\lbFigureSLabelX,\lbFigSideLabelY) {$\lbTargetCut$};
  \node[lbPartitionLabel] at (\lbFigureTLabelX,\lbFigSideLabelY) {$V\setminus\lbTargetCut$};
}

\newcommand{\lbDrawLeftPanel}{%
  \begin{tikzpicture}[x=1cm,y=1cm]
    \lbFigurePathCoordinates
    \lbFigurePartitionGuide

    \draw[lbBaseEdge] (v1) -- (v2) -- (v3) -- (v4) -- (v5) -- (v6);
    \foreach \p in {v1,v2,v3,v4,v5,v6}{\node[lbDarkVertex] at (\p) {};}

    \node[lbVertexLabelBot] at (v1) {$v_1$};
    \node[lbVertexLabelTop] at (v2) {$v_2$};
    \node[lbVertexLabelBot] at (v3) {$v_3$};
    \node[lbVertexLabelTop] at (v4) {$v_4$};
    \node[lbVertexLabelBot] at (v5) {$v_5$};
    \node[lbVertexLabelTop] at (v6) {$v_6$};

    \node[lbLeftLabel]  at ($(v1)!0.5!(v2)$) {$\lbPathEdge{\lbFigureOrder}{1}$};
    \node[lbRightLabel] at ($(v2)!0.5!(v3)$) {$\lbPathEdge{\lbFigureOrder}{2}$};
    \node[lbLeftLabel]  at ($(v3)!0.5!(v4)$) {$\lbPathEdge{\lbFigureOrder}{3}$};
    \node[lbRightLabel] at ($(v4)!0.5!(v5)$) {$\lbPathEdge{\lbFigureOrder}{4}$};
    \node[lbLeftLabel]  at ($(v5)!0.5!(v6)$) {$\lbPathEdge{\lbFigureOrder}{5}$};
  \end{tikzpicture}%
}

\newcommand{\lbDrawRightPanel}{%
  \begin{tikzpicture}[x=1cm,y=1cm]
    \lbFigurePathCoordinates
    \lbFigurePartitionGuide

    \draw[lbMutedEdge] (v1) -- (v2) -- (v3) -- (v4) -- (v5) -- (v6);

    \draw[lbBoundaryEdge] (v1) -- (v2);
    \draw[lbBoundaryEdge] (v2) -- (v3);
    \draw[lbBoundaryEdge] (v5) -- (v6);

    \foreach \p in {v1,v3,v4,v5}{\node[lbMutedVertex] at (\p) {};}
    \foreach \p in {v2,v6}{\node[lbDarkVertex] at (\p) {};}

    \node[lbVertexLabelTop] at (v2) {$v_2$};
    \node[lbVertexLabelTop] at (v6) {$v_6$};

    \node[lbLeftLabel]  at ($(v1)!0.5!(v2)$) {$\lbPathEdge{\lbFigureOrder}{1}$};
    \node[lbRightLabel] at ($(v2)!0.5!(v3)$) {$\lbPathEdge{\lbFigureOrder}{2}$};
    \node[lbLeftLabel]  at ($(v5)!0.5!(v6)$) {$\lbPathEdge{\lbFigureOrder}{5}$};

    % Curve representing the cut
    \coordinate (c0) at ($(v2)+(-\lbFigCurveTailDX,{\lbFigCurveTailY-\lbFigYTop})$);
    \coordinate (c1) at ($(v2)+(0,-\lbFigCurveOffset)$);
    \coordinate (c2) at ($(v4)+(0,\lbFigCurveOffset)$);
    \coordinate (c3) at ($(v6)+(0,-\lbFigCurveOffset)$);
    \coordinate (c4) at ($(v6)+(\lbFigCurveTailDX,{\lbFigCurveTailY-\lbFigYTop})$);
    \draw[lbBoundaryCurve, use Hobby shortcut] (c0) .. (c1) .. (c2) .. (c3) .. (c4);
  \end{tikzpicture}%
} 

\begin{figure}[th]
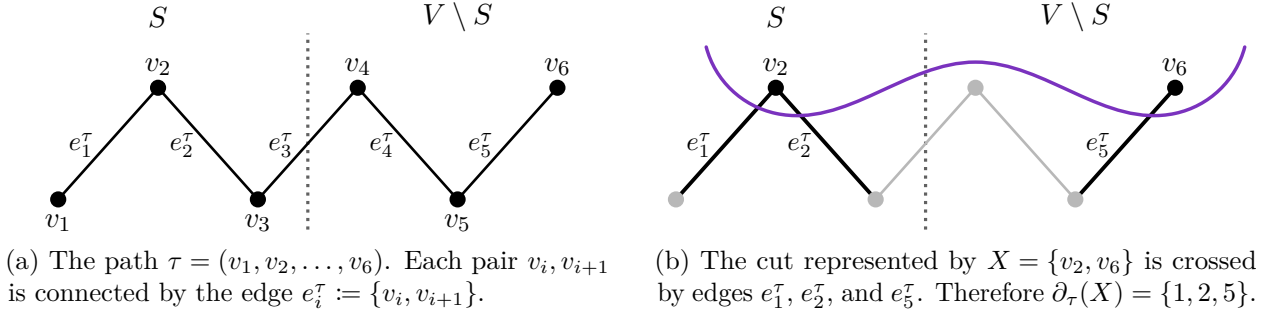

\centering
\begin{subfigure}[t]{0.48\textwidth}
  \centering
  \lbDrawLeftPanel
  \caption{The path~${\lbFigureOrder=(v_1,v_2,\ldots,v_6)}$. Each pair~$v_i,v_{i+1}$ is connected by the edge~${\lbPathEdge{\tau}{i}\coloneqq\{v_i,v_{i+1}\}}$.}
\end{subfigure}
\hfill
\begin{subfigure}[t]{0.48\textwidth}
  \centering
  \lbDrawRightPanel
  \caption{The cut represented by~${\lbFigureSet=\{v_2,v_6\}}$ is crossed by edges~${\lbPathEdge{\lbFigureOrder}{1}}$,~${\lbPathEdge{\lbFigureOrder}{2}}$, and~${\lbPathEdge{\lbFigureOrder}{5}}$.
  Therefore~${\lbBoundary{\lbFigureOrder}{\lbFigureSet}=\{1,2,5\}}$.
}
\end{subfigure}
\caption{Example of the path construction from~\Cref{subsec:path-construction} on~${n=6}$ vertices with~${k=n/2=3}$. Here~${\lbTargetCut=\{v_1,v_2,v_3\}}$, and the vertices in~$\lbTargetCut$ are on the left of the dotted line. 
The cut represented by~${X}$ is a prefix cut if and only if the number of path edges crossed by the cut satisfies~${|\lbBoundary{\tau}{X}|=1}$. }
\label{fig:lb_path_boundaries}
\end{figure}

Fix~${\tau=(v_1,\ldots,v_n)\in\lbOrderings}$.
We view~$\tau$ as a Hamiltonian path, as shown in~\Cref{fig:lb_path_boundaries}.
For~${i\in[n-1]}$, let~${\lbPathEdge{\tau}{i}\coloneqq\{v_i,v_{i+1}\}}$ be the path edge at position~$i$, and let~${\lbPrefix{\tau}{i}\coloneqq\{v_1,\ldots,v_i\}}$ be the corresponding prefix.
We call the cut represented by~${\lbPrefix{\tau}{i}}$ the \emph{prefix cut at position~${i}$}.

An edge~${\{u,v\}}$ \emph{crosses} a set~${X\subseteq V}$ if exactly one of~${u}$ and~${v}$ belongs to~$X$.
For~${X\subseteq V}$, let
\begin{equation*}
\lbBoundary{\tau}{X}
\;\coloneqq\;
\bigl\{i\in[n-1]\setSeparator \lbPathEdge{\tau}{i}\text{ crosses }X\bigr\}.
\end{equation*}
For every nonempty proper set~${X\subsetneq V}$, we consider~$X$ and~${V\setminus X}$ to \emph{represent} the same cut, even though they are not the same set; when referring to cut values, we use both interchangeably.
The cut represented by~$X$ is a prefix cut if and only if~${|\lbBoundary{\tau}{X}|=1}$.

Let~${\lbSpacing\coloneqq\lfloor k/(\ell+1)\rfloor}$. 
Since~${\ell<k}$, we have~${\lbSpacing\geq1}$.
To create the perturbed graphs, we decrease the weight of some path edge~$\lbPathEdge{\tau}{i}$ with~${i \in \lbPerturbIndices}$, where 
\begin{align}
\lbPerturbIndices
\;\coloneqq\;
\{\lbSpacing,\,2\lbSpacing,\,\ldots,\,\ell\lbSpacing\}
\cup
\{n-\ell\lbSpacing,\,\ldots,\,n-2\lbSpacing,\,n-\lbSpacing\}.
\label{eq:perturbation-indices}
\end{align} 
Since~${\ell\lbSpacing<k<n-\ell\lbSpacing}$, it follows that~${|\lbPerturbIndices|=2\ell}$ and~${k\notin\lbPerturbIndices}$.

We first show how to encode which of these~${2\ell}$ path edges are crossed by a given cut using a quadratic form.
The matrices below are indexed by the~${n-1}$ path-edge positions.
Fix~${\lbBoundarySet\subseteq[n-1]}$ and write~${\lbBoundarySet=\{b_1<\ldots<b_m\}}$.
Define its signed indicator vector~${\lbBoundaryVector{\lbBoundarySet}\in\{-1,0,1\}^{n-1}}$ by
\[
\lbBoundaryVectorEntry{\lbBoundarySet}{i}
\;\coloneqq\;
\begin{cases}
(-1)^{r}, & \text{if }i=b_r\text{ for some }r\in[m],\\
0, & \text{if }i\notin\lbBoundarySet.
\end{cases}
\]
For example, if~${\lbBoundarySet=\{1,2,4\}}$, then~$\lbBoundaryVector{\lbBoundarySet}=(-1,1,0,-1,0,\ldots,0)$.   

Denote the maximum absolute value of an entry in a matrix~$\lbEncodingMatrix$ by~${\maxNorm{\lbEncodingMatrix}\coloneqq \max_{i,j}|M_{i,j}|}$. 
For a symmetric~$(n-1)\times(n-1)$ matrix~${\lbEncodingMatrix}$, let
\[
\lbQuadraticForm{\lbEncodingMatrix}{\lbBoundarySet}
\;\coloneqq\;
(\lbBoundaryVector{\lbBoundarySet})^\top
\lbEncodingMatrix
\lbBoundaryVector{\lbBoundarySet}
\;=\;
\sum_{r=1}^{m}
\sum_{s=1}^{m}
(-1)^{r+s}\lbEncodingMatrix_{b_r b_s}.
\]
The next lemma shows that we can choose~$\lbEncodingMatrix$ so that the quadratic form~${\lbQuadraticForm{\lbEncodingMatrix}{\lbBoundarySet}}$ determines~${\lbBoundarySet\cap\lbPerturbIndices}$.

\begin{restatable}[Encoding crossings with the perturbed path edges]{lemma}{encodingCrossingsWithPerturbedPathIndices}
\label{lemma:encoding-perturbed-edges}
There is a symmetric matrix~${\lbEncodingMatrix\in\Z^{(n-1)\times(n-1)}}$ with zero diagonal such that~${1\leq \maxNorm{\lbEncodingMatrix} \leq \lbWeightRange}$ and, for every~${\lbBoundarySet,\lbOtherBoundarySet\subseteq[n-1]}$ with~${|\lbBoundarySet|,|\lbOtherBoundarySet|\geq2}$,
\[
\lbQuadraticForm{\lbEncodingMatrix}{\lbBoundarySet}
\;=\;\lbQuadraticForm{\lbEncodingMatrix}{\lbOtherBoundarySet}
\qquad\implies\qquad
\lbBoundarySet\cap\lbPerturbIndices
\;=\;
\lbOtherBoundarySet\cap\lbPerturbIndices.
\]
Moreover, after setting~${\lbEncodingMatrix_{0j}\coloneqq0}$ and~${\lbEncodingMatrix_{i n}\coloneqq0}$ whenever these entries occur, the quantity
\[
\lbEncodingMatrix_{r,s-1}
+
\lbEncodingMatrix_{r-1,s}
-
\lbEncodingMatrix_{r-1,s-1}
-
\lbEncodingMatrix_{r,s}
\]
is a positive even integer for all~${1\leq r<s-1\leq n-1}$.
\end{restatable}

\begin{proof}
    Deferred to~\Cref{subsec:constructing_the_encoding_matrix_M}.
\end{proof}

We now use the matrix~${\lbEncodingMatrix}$ from~\Cref{lemma:encoding-perturbed-edges} to construct the base graph for the lower bound.

\begin{restatable}[Base graph~$\lbBaseGraph{\tau}$]{definition}{baseGraphDefinition}
\label{def:base_graph}
Let~$n$ be sufficiently large and even, and let~${V=[n]}$.
Let~${\lbEncodingMatrix\in\Z^{(n-1)\times(n-1)}}$ be the matrix from~\Cref{lemma:encoding-perturbed-edges}.
Define~${\lbWeightScale\coloneqq2n^2 \maxNorm{\lbEncodingMatrix}}$ and set~${\lbEncodingMatrix_{0j}\coloneqq0}$ and~${\lbEncodingMatrix_{i n}\coloneqq0}$ whenever these entries occur below.

For a path~${\tau=(v_1,\ldots,v_n)\in\lbOrderings}$, let~${\lbBaseGraph{\tau}=(V,w)}$ be the complete graph with the following weights.
For each nonpath edge~${\{v_r,v_s\}}$ with~${s\geq r+2}$, define
\begin{align}
\lbEdgeWeight{r}{s}
\;&\coloneqq\;
\lbEncodingMatrix_{r,s-1}
+
\lbEncodingMatrix_{r-1,s}
-
\lbEncodingMatrix_{r-1,s-1}
-
\lbEncodingMatrix_{r,s}.
\label{eq:nonpath_edge_weights}
\end{align}
For each path edge~$\{v_i,v_{i+1}\}$ with~${i\in[n-1]}$, define 
\begin{align}
\lbEdgeWeight{i}{i+1}
\;&\coloneqq\;
n^3\lbWeightScale
-
\sum_{r=1}^{i-1}
\sum_{s=i+1}^{n}
\lbEdgeWeight{r}{s}
-
\sum_{s=i+2}^{n}
\lbEdgeWeight{i}{s}.
\label{eq:path_edge_weight}
\end{align}
\end{restatable}

We use the following properties of the base graph~$\lbBaseGraph{\tau}$ to prove the lower bound.

\begin{restatable}[Properties of the base graph]{lemma}{baseGraphProperties}
\label{lemma:base_graph_properties} 
    Let~${\lbBaseGraph{\tau}=(V,w)}$ be the base graph in~\Cref{def:base_graph}. 
    \begin{enumerate}[label=(\roman*)]
        \item \label{item:base_graph_edge_weights} 
        Every edge weight in~$\lbBaseGraph{\tau}$ is a positive even integer at most~${\lbWeightRange}$.
        
        \item \label{item:base_graph_min_cuts} 
        The minimum cuts of~$\lbBaseGraph{\tau}$ are exactly the prefix cuts of~$\tau$.
        
        \item \label{item:base_graph_boundary_intersections} 
        If~${\emptyset\subsetneq X,Y\subsetneq V}$ represent two non-prefix cuts and satisfy~${\cutGraphSet{\lbBaseGraph{\tau}}{X}=\cutGraphSet{\lbBaseGraph{\tau}}{Y}}$, then
        \[
        \lbBoundary{\tau}{X}\cap\lbPerturbIndices
        \;=\;
        \lbBoundary{\tau}{Y}\cap\lbPerturbIndices.
        \]  
    \end{enumerate} 
\end{restatable}

\begin{proof}
    Deferred to~\Cref{subsec:properties_of_the_base_graph}.
\end{proof}

\subsection{Perturbing the base graph}
\label{subsec:perturbing-base-graph}

For~${\tau\in\lbOrderings}$ and~${i\in\lbPerturbIndices}$, let~${\lbPerturbedGraph{\tau}{i}}$ be the \emph{perturbed graph} obtained from~${\lbBaseGraph{\tau}}$ by decreasing the weight of the path edge~${\lbPathEdge{\tau}{i}=\{v_i,v_{i+1}\}}$ by~$1$ and leaving all other weights unchanged.
The next lemma gives the properties of this perturbation that we use in the proof of~\Cref{thm:lower}.

\begin{lemma}[Perturbing one path edge]
\label{lemma:one-edge-perturbation}
Fix~${\tau\in\lbOrderings}$ and~${i\in\lbPerturbIndices}$.
Then~${\lbPerturbedGraph{\tau}{i}}$ has positive integer edge weights bounded by~${\lbWeightRange}$ and satisfies the following properties.
\begin{enumerate}[label=(\roman*)]
    \item \label{item:perturbed_unique_min}
    The prefix cut represented by~${\lbPrefix{\tau}{i}}$ is the unique minimum cut of~${\lbPerturbedGraph{\tau}{i}}$.
    
    \item \label{item:perturbed_strict_comparisons}
    For all~${X,Y\subseteq V}$, if~${\cutGraphSet{\lbBaseGraph{\tau}}{X}<\cutGraphSet{\lbBaseGraph{\tau}}{Y}}$, then~${\cutGraphSet{\lbPerturbedGraph{\tau}{i}}{X}<\cutGraphSet{\lbPerturbedGraph{\tau}{i}}{Y}}$. 
    
    \item \label{item:perturbed_changed_comparisons}  
    For all~${X,Y\subseteq V}$, if~${\cutGraphSet{\lbBaseGraph{\tau}}{X}=\cutGraphSet{\lbBaseGraph{\tau}}{Y}}$ but~${\cutGraphSet{\lbPerturbedGraph{\tau}{i}}{X}\neq\cutGraphSet{\lbPerturbedGraph{\tau}{i}}{Y}}$, then~$X$ and~$Y$ represent prefix cuts of~$\tau$, and one of them belongs to~${\{\lbPrefix{\tau}{i},\,V\setminus\lbPrefix{\tau}{i}\}}$.
\end{enumerate}
\end{lemma}

\begin{proof}
By~\Cref{lemma:base_graph_properties}\ref{item:base_graph_edge_weights}, every edge weight in~${\lbBaseGraph{\tau}}$ is a positive even integer at most~${\lbWeightRange}$. 
Hence decreasing the weight of~${\{v_i,v_{i+1}\}}$ by~$1$ gives a positive integer, and it follows that~${\lbPerturbedGraph{\tau}{i}}$ has positive integer edge weights bounded by~${\lbWeightRange}$. 
Now for every~${X\subseteq V}$,
\begin{align}
    \label{eq:perturbed_cut_values}
    \cutGraphSet{\lbPerturbedGraph{\tau}{i}}{X}
    \;=\;
    \cutGraphSet{\lbBaseGraph{\tau}}{X}
    -
    \indicator{\lbPathEdge{\tau}{i}\text{ crosses }X}.
\end{align} 
By~\Cref{lemma:base_graph_properties}\ref{item:base_graph_min_cuts}, the minimum cuts of~$\lbBaseGraph{\tau}$ are exactly the prefix cuts~${\lbPrefix{\tau}{1},\ldots,\lbPrefix{\tau}{n-1}}$, so every non-prefix cut has strictly larger value than a prefix cut in~${\lbBaseGraph{\tau}}$. 
Since all cut values in~${\lbBaseGraph{\tau}}$ are even, every non-prefix cut has value at least~${2+\cutGraphSet{\lbPerturbedGraph{\tau}{i}}{\lbPrefix{\tau}{i}}}$. The perturbation decreases any cut value by at most~$1$, so every non-prefix cut still has larger value than~${\lbPrefix{\tau}{i}}$ in~${\lbPerturbedGraph{\tau}{i}}$. 
Among the prefix sets, only~${\lbPrefix{\tau}{i}}$ crosses~${\lbPathEdge{\tau}{i}}$, so it follows by~\eqref{eq:perturbed_cut_values} that~${\cutGraphSet{\lbPerturbedGraph{\tau}{i}}{\lbPrefix{\tau}{i}}=\cutGraphSet{\lbBaseGraph{\tau}}{\lbPrefix{\tau}{i}}-1}$, while the cut value of all other prefix sets remains unchanged in the perturbed graph.
Therefore the cut represented by~${\lbPrefix{\tau}{i}}$ is the unique minimum cut of~${\lbPerturbedGraph{\tau}{i}}$, which proves~(i).

Fix~${X,Y\subseteq V}$. 
Suppose first that~${\cutGraphSet{\lbBaseGraph{\tau}}{X}<\cutGraphSet{\lbBaseGraph{\tau}}{Y}}$.
Since every cut value in~${\lbBaseGraph{\tau}}$ is even,
\[
\cutGraphSet{\lbBaseGraph{\tau}}{Y}
-
\cutGraphSet{\lbBaseGraph{\tau}}{X}
\;\geq\;
2.
\]
By~\eqref{eq:perturbed_cut_values}, each of the two cut values decreases by either~$0$ or~$1$, so~${\cutGraphSet{\lbPerturbedGraph{\tau}{i}}{Y}-\cutGraphSet{\lbPerturbedGraph{\tau}{i}}{X}\geq 1}$ and therefore~${\cutGraphSet{\lbPerturbedGraph{\tau}{i}}{X}<\cutGraphSet{\lbPerturbedGraph{\tau}{i}}{Y}}$.  This proves~(ii).

Finally, suppose that~${\cutGraphSet{\lbBaseGraph{\tau}}{X}=\cutGraphSet{\lbBaseGraph{\tau}}{Y}}$ but~${\cutGraphSet{\lbPerturbedGraph{\tau}{i}}{X}\neq\cutGraphSet{\lbPerturbedGraph{\tau}{i}}{Y}}$.
By~\eqref{eq:perturbed_cut_values}, exactly one of~$X$ and~$Y$ crosses~${\lbPathEdge{\tau}{i}}$. 
The sets~$X$ and~$Y$ cannot represent one prefix cut and one non-prefix cut, since the prefix cuts are exactly the minimum cuts of~${\lbBaseGraph{\tau}}$ and therefore have strictly smaller value than every non-prefix cut.
Suppose for contradiction that~$X$ and~$Y$ both represent non-prefix cuts.
Then by~\Cref{lemma:base_graph_properties}\ref{item:base_graph_boundary_intersections},
${\lbBoundary{\tau}{X}\cap\lbPerturbIndices=\lbBoundary{\tau}{Y}\cap\lbPerturbIndices}$.
Since~${i\in\lbPerturbIndices}$, this implies that~${\lbPathEdge{\tau}{i}}$ crosses~$X$ if and only if it crosses~$Y$, a contradiction.

Thus~$X$ and~$Y$ both represent prefix cuts.
Among the prefix cuts of~$\tau$, the prefix cut at position~$i$ is the unique one crossed by~${\lbPathEdge{\tau}{i}}$.
Since exactly one of~$X$ and~$Y$ crosses~${\lbPathEdge{\tau}{i}}$, one of them belongs to~${\{\lbPrefix{\tau}{i},\,V\setminus\lbPrefix{\tau}{i}\}}$, which proves~(iii).   
\end{proof}

\subsection{Counting hidden orderings}
\label{subsec:counting-hidden-orderings}
We prove the lower bound in~\Cref{thm:lower} by choosing~$\tau$ uniformly at random from~$\lbOrderings$ and then possibly decreasing the weight of a single path edge~${\lbPathEdge{\tau}{i}}$ with~${i\in\lbPerturbIndices}$. 
Since the cut represented by~${\lbPrefix{\tau}{i}}$ is the unique minimum cut of~${\lbPerturbedGraph{\tau}{i}}$, distinguishing the base graph from the possible perturbations forces the algorithm to identify many prefix cuts of~$\tau$.
 
The next lemma allows us to count how many orderings in~$\lbOrderings$ can share at least~$\ell$ of these prefix cuts; this is where the~${\log\ell}$ factor in~\Cref{thm:lower} comes from.  

\begin{lemma}[Number of orderings sharing at least~$\ell$ prefix cuts]
\label{lemma:common-cuts}
Fix~${\tau\in\lbOrderings}$.
There are at most~${|\lbOrderings|2^{-\Omega(n\log\ell)}}$ distinct orderings~${\lbOtherOrdering\in\lbOrderings}$ satisfying~${\bigl|\bigl\{i\in\lbPerturbIndices\setSeparator\lbPrefix{\tau}{i}=\lbPrefix{\lbOtherOrdering}{i}\bigr\}\bigr|\geq\ell}$. 
\end{lemma}

\begin{proof}
Sample~$\uniformSample{\lbOtherOrdering}{\lbOrderings}$ uniformly at random, and let~$\lbSharedPrefixEvent$ be the event that
\[
\left|
\bigl\{
i\in\lbPerturbIndices
\setSeparator
\lbPrefix{\tau}{i}
=
\lbPrefix{\lbOtherOrdering}{i}
\bigr\}
\right|
\;\geq\;
\ell.
\]
Since every ordering in~$\lbOrderings$ is equally likely, it is enough to prove~${\pr{\lbSharedPrefixEvent}\leq 2^{-\Omega(n\log\ell)}}$. 

Write~${\lbPerturbIndices=\lbPerturbIndicesLower\cup\lbPerturbIndicesUpper}$ where~${\lbPerturbIndicesLower \;\coloneqq\; \{ \lbSpacing,\,2\lbSpacing,\,\ldots,\,\ell\lbSpacing \}}$ and~${\lbPerturbIndicesUpper \;\coloneqq\; \{n-\ell\lbSpacing,\,\ldots,\,n-2\lbSpacing,\,n-\lbSpacing\}}$, and let~${h\coloneqq\lfloor\ell/2\rfloor}$.
If the event~$\lbSharedPrefixEvent$ occurs, then at least half of the indices in~${\bigl\{i\in\lbPerturbIndices\setSeparator\lbPrefix{\tau}{i}=\lbPrefix{\lbOtherOrdering}{i}\bigr\}}$ belong to the same set~${\lbPerturbIndicesLower}$ or~${\lbPerturbIndicesUpper}$.  
Hence there exist~${a\in\{1,2\}}$ and~${J\subseteq\lbPerturbIndices_a}$ with~${|J|=h}$ such that~${\lbPrefix{\tau}{i}=\lbPrefix{\lbOtherOrdering}{i}}$ for every~${i\in J}$. Fix such an~$a$ and~$J$.

If~${a=1}$, write~${J=\{j_1,\ldots,j_h\}}$, and if~${a=2}$, write~${J=\{n-j_h,\ldots,n-j_1\}}$, where~${j_1<\ldots<j_h}$ and~${j_r\in\lbPerturbIndicesLower}$ for each~${r\in[h]}$. 
By symmetry, it is enough to analyze the case with~${a=1}$; for~${a=2}$, taking complements changes equality of the prefixes at positions~${n-j_r}$ into equality of the last~${j_r}$ vertices, and reading the ordering from right to left gives the same probabilities.

Suppose therefore that~${a=1}$. 
Write~${\tau=(v_1,\ldots,v_n)}$ and~${\lbOtherOrdering=(v_1',\ldots,v_n')}$, and let~${j_0\coloneqq0}$ and~${j_{h+1}\coloneqq k}$.
Then~${\lbPrefix{\tau}{i}=\lbPrefix{\lbOtherOrdering}{i}}$ for every~${i\in J}$ if and only if 
\[
\bigl\{v_i\setSeparator j_r<i\leq j_{r+1}\bigr\}
\;=\;
\bigl\{v_i'\setSeparator j_r<i\leq j_{r+1}\bigr\}
\qquad
\text{for every }r\in\{0,\ldots,h\}.
\]
Since~${(v_1',\ldots,v_k')}$ is uniformly ordered, all of these equalities hold simultaneously with probability
\[
\pr{
\lbPrefix{\tau}{i}
=
\lbPrefix{\lbOtherOrdering}{i}
\text{ for every }i\in J
}
\;=\;
\frac{\prod_{r=0}^{h}(j_{r+1}-j_r)!}{k!}.
\]
Each difference satisfies~${j_{r+1}-j_r \geq \lbSpacing}$ and their sum is~${\sum_{r=0}^{h}(j_{r+1}-j_r)=k}$. Therefore the product~${\prod_{r=0}^{h}(j_{r+1}-j_r)!}$ is maximized when~$h$ of the differences equal~$\lbSpacing$ and the remaining difference equals~${k-h\lbSpacing}$. 
This holds as, if~${x\geq y>\lbSpacing}$, replacing~${x,y}$ by~${x+1,y-1}$ multiplies~${x!y!}$ by~${(x+1)/y\geq1}$. 
Hence~${\prod_{r=0}^{h}(j_{r+1}-j_r)!\leq(k-h\lbSpacing)!(\lbSpacing!)^h}$. 
Moreover, as~${2\leq\ell<k}$, the definitions of~$h$ and~$\lbSpacing$ give~${k/9\leq h\lbSpacing<k/2}$ and~${\lbSpacing\leq k/(\ell+1)}$.
Therefore
\[
\pr{
\lbPrefix{\tau}{i}
=
\lbPrefix{\lbOtherOrdering}{i}
\text{ for every }i\in J
}
\;\leq\;
\frac{
(k-h\lbSpacing)!
(\lbSpacing!)^h
}{k!}
\;\leq\;
\left(
\frac{\lbSpacing}
{k-h\lbSpacing+1}
\right)^{h\lbSpacing}
\;\leq\;
\left(\frac{2}{\ell+1}\right)^{k/9}.
\]
There are~${2\binom{\ell}{h}}$ choices of~$a$ and~$J$.
Thus, by a union bound,
\[
\pr{\lbSharedPrefixEvent}
\;\leq\;
2\binom{\ell}{h}
\left(\frac{2}{\ell+1}\right)^{k/9}
\;\leq\;
2^{\ell+1}
\left(\frac{2}{\ell+1}\right)^{k/9}
\;\leq\; 
2^{-\Omega(k\log\ell)}.
\]
Since~${n=2k}$ and~$n$ is sufficiently large,~${\pr{\lbSharedPrefixEvent}\leq 2^{-\Omega(n\log\ell)}}$. 
The result follows as~$\lbOtherOrdering$ is uniform over~$\lbOrderings$, so at most~${|\lbOrderings| 2^{-\Omega(n\log\ell)}}$ orderings satisfy~$\lbSharedPrefixEvent$. 
\end{proof}

\subsection{A hard distribution over graphs}
\label{subsec:distributional-lower-bound} 
Let~${\lbDistribution}$ be the following distribution over graphs.
First sample~${\uniformSample{\tau}{\lbOrderings}}$.
With probability~${1/2}$, return the base graph~${\lbBaseGraph{\tau}}$, and with probability~${1/2}$, sample~${\uniformSample{i}{\lbPerturbIndices}}$ and return the perturbed graph~${\lbPerturbedGraph{\tau}{i}}$.

We use the distribution~$\lbDistribution$ to prove~\Cref{thm:lower}.
By Yao's principle, it is enough to show that a deterministic algorithm that outputs a minimum cut with probability at least~${9/10}$ over~$\lbDistribution$ must make~${\Omega(n\log\ell)}$ comparisons.

\begin{lemma}[Correctness on many graphs]
\label{lemma:many-correct-base-and-perturbed-graphs}
Let~$\alg$ be a deterministic algorithm that outputs a minimum cut with probability at least~${9/10}$ when the input is drawn from~$\lbDistribution$.
Then there are at least~${2|\lbOrderings|/5}$ orderings~$\tau$ for which~$\alg$ is correct on the base graph~${\lbBaseGraph{\tau}}$ and on at least~${\ell+1}$ of the perturbed graphs~${\lbPerturbedGraph{\tau}{i}}$.
\end{lemma}

\begin{proof}
The distribution~$\lbDistribution$ returns a base graph~${\lbBaseGraph{\tau}}$ with probability~${1/2}$.
Since~$\alg$ succeeds with probability at least~${9/10}$ over~$\lbDistribution$, it follows that~$\alg$ is correct on~${\lbBaseGraph{\tau}}$ for at least a~${4/5}$ fraction of the orderings~${\tau\in\lbOrderings}$.

Similarly, when~$\lbDistribution$ returns a perturbed graph, the pair~${(\tau,i)}$ is uniform over~${\lbOrderings\times\lbPerturbIndices}$.
Thus~$\alg$ outputs the minimum cut of~${\lbPerturbedGraph{\tau}{i}}$ for at least a~${4/5}$ fraction of these pairs.
Since~${|\lbPerturbIndices|=2\ell}$,
\[
\ex{\left|
\left\{
i\in\lbPerturbIndices
\setSeparator
\alg\text{ is correct on }\lbPerturbedGraph{\tau}{i}
\right\}
\right|}
\;\geq\;
2\ell\cdot\frac45
\;=\;
\frac{8\ell}{5},
\]
where the expectation is over~$\tau$. 
Let~$p$ be the fraction of orderings~$\tau$ for which~$\alg$ is correct on at least~${\ell+1}$ perturbed graphs.
For the remaining~${1-p}$ fraction, it is correct on at most~$\ell$ perturbations, so~${8\ell/5\leq p(2\ell)+(1-p)\ell=(1+p)\ell}$.
Hence~${p\geq3/5}$.

Thus~$\alg$ is correct on the base graph for at least a~${4/5}$ fraction of the orderings~${\tau\in\lbOrderings}$, and on at least~${\ell+1}$ perturbations for at least a~${3/5}$ fraction of the orderings.
The two properties therefore hold simultaneously for at least a~${2/5}$ fraction of the orderings~${\tau\in\lbOrderings}$.
\end{proof}

\begin{lemma}[Deterministic comparison lower bound under~$\lbDistribution$]
\label{lemma:distributional-lower-bound}
Let~$\alg$ be a deterministic comparison algorithm that makes at most~$q$ comparisons on every $n$-vertex input graph.
If~$\alg$ outputs a minimum cut with probability at least~${9/10}$ when the input is drawn from~$\lbDistribution$, then
\[
q
\;=\;
\Omega(n\log\ell).
\]
\end{lemma}

\begin{proof}
By~\Cref{lemma:many-correct-base-and-perturbed-graphs}, there are at least~${2|\lbOrderings|/5}$ orderings~$\tau$ for which~$\alg$ is correct on~${\lbBaseGraph{\tau}}$ and on at least~${\ell+1}$ of its perturbed graphs.
Fix one such~$\tau$.

By~\Cref{lemma:base_graph_properties}, the cut output by~$\alg$ on~${\lbBaseGraph{\tau}}$ is the prefix cut at some position~$j$.
By~\Cref{lemma:one-edge-perturbation}, the unique minimum cut of~${\lbPerturbedGraph{\tau}{i}}$ is the prefix cut at position~$i$.
Thus, among the perturbations on which~$\alg$ is correct, only the one with~${i=j}$ can cause~$\alg$ to output the same cut as on~${\lbBaseGraph{\tau}}$. 
Since~$\alg$ is correct on at least~${\ell+1}$ perturbed graphs, there is a set~${J\subseteq\lbPerturbIndices}$ with~${|J|=\ell}$ such that, for every~${i\in J}$,~$\alg$ is correct on~${\lbPerturbedGraph{\tau}{i}}$ and outputs different cuts on~${\lbBaseGraph{\tau}}$ and~${\lbPerturbedGraph{\tau}{i}}$. 

For the counting argument, we record slightly more information than the transcript of~$\alg$.
Whenever~$\alg$ compares two cuts~$X$ and~$Y$ and finds that they have equal cut values, we also record the result of comparing~$X$ with~$\lbTargetCut$.
We call this the \emph{augmented transcript}.
Since~$\alg$ makes at most~$q$ comparisons, the augmented transcript contains at most~$2q$ comparison outcomes.

Now fix~${i\in J}$.
Since~$\alg$ is deterministic and has different outputs on~${\lbBaseGraph{\tau}}$ and~${\lbPerturbedGraph{\tau}{i}}$, there is a first comparison~${X,Y}$ whose answer differs on the two graphs. 
By~\Cref{lemma:one-edge-perturbation}\ref{item:perturbed_strict_comparisons}, the perturbed graph preserves strict comparison answers, so we must have~${\cutGraphSet{\lbBaseGraph{\tau}}{X}=\cutGraphSet{\lbBaseGraph{\tau}}{Y}}$. 
Thus by~\Cref{lemma:one-edge-perturbation}\ref{item:perturbed_changed_comparisons},~$X$ and~$Y$ represent prefix cuts of~$\tau$, and one of them represents the prefix cut~$\lbPrefix{\tau}{i}$.
Renaming~$X$ and~$Y$ if necessary, suppose that~${X\in\{\lbPrefix{\tau}{i},V\setminus\lbPrefix{\tau}{i}\}}$.
Since~$X$ and~$\lbTargetCut$ represent prefix cuts of~$\tau$,~\Cref{lemma:base_graph_properties}\ref{item:base_graph_min_cuts} gives~${\cutGraphSet{\lbBaseGraph{\tau}}{X}=\cutGraphSet{\lbBaseGraph{\tau}}{\lbTargetCut}}$, so the augmented transcript contains this equality. 

Now let~${\lbOtherOrdering\in\lbOrderings}$ be such that~$\alg$ has the same augmented transcript on~${\lbBaseGraph{\lbOtherOrdering}}$.
Then
\[
\cutGraphSet{\lbBaseGraph{\lbOtherOrdering}}{X}
\;=\;
\cutGraphSet{\lbBaseGraph{\lbOtherOrdering}}{Y}
\;=\;
\cutGraphSet{\lbBaseGraph{\lbOtherOrdering}}{\lbTargetCut}.
\]
Since the cut represented by~$\lbTargetCut$ is a minimum cut of every base graph, the equalities above imply that~$X$ and~$Y$ also represent minimum cuts of~${\lbBaseGraph{\lbOtherOrdering}}$.
By~\Cref{lemma:base_graph_properties}\ref{item:base_graph_min_cuts}, they therefore represent prefix cuts of~$\lbOtherOrdering$. 
Since~$X$ represents the same cut as~${\lbPrefix{\tau}{i}}$, the set~${\lbPrefix{\tau}{i}}$ also represents a prefix cut of~$\lbOtherOrdering$.
Hence, for some~${j\in[n-1]}$,
\[
\lbPrefix{\tau}{i}
\;\in\;
\bigl\{
\lbPrefix{\lbOtherOrdering}{j},
V\setminus\lbPrefix{\lbOtherOrdering}{j}
\bigr\}.
\]
Since~${i\in\lbPerturbIndices}$, we have~${i\neq k}$.
If~${i<k}$, then~${\lbPrefix{\tau}{i}}$ is a nonempty proper subset of~$\lbTargetCut$, whereas~${V\setminus\lbPrefix{\lbOtherOrdering}{j}}$ either contains~${V\setminus\lbTargetCut}$ or is contained in~${V\setminus\lbTargetCut}$.
If~${i>k}$, then~${\lbPrefix{\tau}{i}}$ contains all of~$\lbTargetCut$, whereas~${V\setminus\lbPrefix{\lbOtherOrdering}{j}}$ either omits an element of~$\lbTargetCut$ or is contained in~${V\setminus\lbTargetCut}$.
Thus~${\lbPrefix{\tau}{i}=\lbPrefix{\lbOtherOrdering}{j}}$. 
Both sets have cardinality~$i$, so~${j=i}$ and we have~${\lbPrefix{\tau}{i}=\lbPrefix{\lbOtherOrdering}{i}}$.

Applying this argument to each~${i\in J}$ gives~${\bigl|\bigl\{i\in\lbPerturbIndices\setSeparator\lbPrefix{\tau}{i}=\lbPrefix{\lbOtherOrdering}{i}\bigr\}\bigr|\geq|J|=\ell}$.  
Hence, by~\Cref{lemma:common-cuts}, at most~${|\lbOrderings|2^{-\Omega(n\log\ell)}}$ orderings~$\lbOtherOrdering$ have the same augmented transcript as~$\tau$. 
This argument applies to at least~${2|\lbOrderings|/5}$ orderings~$\tau$, so the base graphs for these orderings produce at least
\[
\frac{2|\lbOrderings|/5}
{|\lbOrderings|2^{-\Omega(n\log\ell)}}
\;=\;
2^{\Omega(n\log\ell)}
\]
different augmented transcripts. 
Each augmented transcript contains at most~${2q}$ comparison outcomes, each with three possible values.
Thus there are fewer than~${3^{2q+1}}$ possible augmented transcripts, and therefore~${3^{2q+1}\geq 2^{\Omega(n\log\ell)}}$.
Taking logarithms gives~${q=\Omega(n\log\ell)}$. 
\end{proof}

\subsection{Proof of \texorpdfstring{\Cref{thm:lower}}{Theorem 1.4}}
We can now use the deterministic bound in~\Cref{lemma:distributional-lower-bound} to prove~\Cref{thm:lower}.

\randomizedlowerbound*

\begin{proof}
By~\Cref{lemma:base_graph_properties,lemma:one-edge-perturbation}, every graph drawn from~$\lbDistribution$ is complete and has positive integer edge weights at most~${\lbWeightRange}$. 
Let~$\alg$ be a randomized comparison algorithm that outputs a minimum cut with probability at least~${2/3}$ on every graph in the support of~$\lbDistribution$, and let~$q$ be the worst-case expected number of comparisons made by~$\alg$ on these graphs. 

Fix a graph~$G$ in the support of~$\lbDistribution$.
Run~$\alg$ independently three times on~$G$ and return a cut with the smallest value among the three returned cuts, using two additional comparisons.
Since each run returns a minimum cut with probability at least~${2/3}$, all three runs fail with probability at most~${(1/3)^3=1/27}$.
Thus the resulting algorithm outputs a minimum cut of~$G$ with probability at least~${26/27}$ and makes at most~${3q+2}$ comparisons in expectation.

Now stop this algorithm after~${20(3q+2)}$ comparisons and return an arbitrary cut if it has not terminated.
By Markov's inequality, this happens with probability at most~${1/20}$.
Hence, for every graph~$G$ in the support of~$\lbDistribution$, the truncated algorithm makes at most~${20(3q+2)}$ comparisons and outputs a minimum cut with probability at least~${1-1/27-1/20 > 9/10}$. 

Therefore, when~$G$ is drawn from~$\lbDistribution$, the truncated randomized algorithm succeeds with probability greater than~${9/10}$. 
By averaging over the random coins, we can fix them so that the resulting deterministic algorithm succeeds with probability at least~${9/10}$ over~$\lbDistribution$. 
By the deterministic lower bound in~\Cref{lemma:distributional-lower-bound},~${20(3q+2)=\Omega(n\log\ell)}$, and thus~${q=\Omega(n\log\ell)}$. 
\end{proof}

%%%%%%%%%%%%%%%%%%%
\section{Discussion}\label{sec:discussion}
%%%%%%%%%%%%%%%%%%%
Our results leave several open questions. 
\begin{enumerate}
\item \textbf{Symmetric~SFM.}
By~\Cref{thm:intro-sfm,thm:lower}, the comparison complexity of symmetric~SFM lies between~${\Omega(n\log n)}$ and~${O(n^3)}$.
Can this gap be reduced?
A deterministic algorithm using~${o(n^3)}$ comparison queries would also improve the best-known deterministic value-oracle bound, which has been~${O(n^3)}$ for almost thirty years~\cite{Queyranne,Nagamochi10,IwataKonno}. 
We therefore suspect that improving the lower bound might be more approachable. 
A comparison-query lower bound of~${\Omega(n^c)}$ for any constant~${c>1}$ would be very interesting; no such lower bound is known in the value-oracle model, but since comparison queries reveal less information, this may make lower bounds easier to prove.

\item \textbf{Weighted minimum cut with integer weights.}
By~\Cref{thm:bounded-weight-min-cut-upper-bound}, we can find a minimum cut of a graph with edge weights in~${\weightedEdgeWeights}$ using
\[
\weightedMinCutUpperBound
\]
cut comparisons.
Can the dependence on~$\B$ be reduced to~${\poly(\log \B)}$? Or removed entirely? 
In the value-oracle model, there is a randomized algorithm that finds a weighted minimum cut using~${\softO(n)}$ queries~\cite{MukhopadhyayNanongkai} for any weight range, but we do not know how to implement this using comparisons.

\item \textbf{The weight range in the randomized lower bound.} 
By~\Cref{cor:epsilon}, every randomized algorithm that outputs a minimum cut with probability at least~${2/3}$ makes~${\Omega(n\log n)}$ expected comparisons in the worst case, even when the edge weights are subexponential. 
However, for polynomially bounded edge weights, the corresponding lower bound in~\Cref{cor:poly} is only~${\Omega(n\log\log n)}$.  
Can this be strengthened to~${\Omega(n\log n)}$?
More generally, can the weight--query tradeoff in~\Cref{thm:lower} be improved?

\end{enumerate}

\paragraph{AI disclosure.}
We used GPT-5.6~Sol, GPT-6 Astra, Antigravity~2.0, and Opus~5 for proof exploration, proof checking, and preparing the exposition.  
Collectively, these tools materially affected the technical presentation throughout and portions of the development in~\Cref{sec:symmetric,sec:bounded,sec:lower,app:weighted-graph-contraction,app:lower_bound_proofs}. The authors verified the correctness and originality of all content, including every proof and every reference, and we take full responsibility for the paper.

\begingroup
\sloppy
\bibliographystyle{alphaurl}
\bibliography{refs}
\endgroup

\appendix
\crefname{appendix}{appendix}{appendices}
\Crefname{appendix}{Appendix}{Appendices}
\crefalias{section}{appendix}
\crefalias{subsection}{appendix}

\section{Finding a minimum-cut-preserving contraction}
\label{app:weighted-graph-contraction}

This appendix contains the proof of~\Cref{lemma:contraction}, which we restate for convenience.

\GraphContractionUsingComparisons*

In the remainder of this section, fix an undirected weighted graph~${G=(V,w)}$ on~$n$ vertices with edge weights in~$\{0,1,\ldots,\B\}$ for a known positive integer~$\B$.   
The \emph{support graph} of~$G$ is the unweighted graph on~$V$ whose edges are the positive-weight edges of~$G$.

The proof of~\Cref{lemma:contraction} has three steps.
First, we recover a subgraph~$\supp$ of the support graph of~$G$ while leaving only~${O(n\B)}$ total unrecovered weight~(see~\Cref{lemma:support}).
Next, we use an expander decomposition to partition~$\supp$ into blocks~${V=V_1\sqcup\ldots\sqcup V_{\partitionIndex}}$ with few recovered edges between different blocks~(see~\Cref{thm:expander}).
Finally, we use comparisons to choose a set~${\contractionSet{i}\subseteq V_i}$ inside each nonsingleton block~$V_i$, and contract each nonempty~${\contractionSet{i}}$ to one supervertex~(see~\Cref{lemma:external}). 
We show that the resulting contracted graph satisfies the properties in~\Cref{lemma:contraction}.

This strategy of explicitly recovering part of the graph, then finding a suitable contraction or expander decomposition, has been used throughout the cut-query literature~(see, e.g., \cite{RSW,StarContraction,AnandSaranurakWang,KennethMordochKrauthgamer2025,KennethMordochKrauthgamer2026a,KennethMordochKrauthgamer2026b}).
However, these algorithms either assume that the graph is unweighted or that we know the exact weighted cut values. 
In the comparison-oracle model, we do not know the cut values, and generally cannot even recover the graph up to scaling~\cite[Appendix~D.4.2]{comparison}. 
In the remainder of this section, we show how to use cut comparisons to guarantee that the properties in~\Cref{lemma:contraction} hold, without ever using any explicit cut values.

\subsection{Recovering most of the support graph}
\label{subsec:recovering-support-graph}
Since each edge weight in~$G$ is at most~$\B$, the total edge weight of~$G$ could be as large as~$O(n^2 B)$. 
In this subsection, we show how to recover enough of the support graph of~$G$ to reduce the unrecovered total weight to~${O(n\B)}$. 
The algorithm never uses the edge weights explicitly; instead, it constructs an explicit set of edges incident to each vertex such that the unrecovered edges incident to the vertex have total weight less than~$\B$. We use the following weighted version of~\cite[Observation~D.1]{comparison}.

\begin{lemma}[Testing whether a set contains half the incident weight]
\label{lemma:half-incident-weight}
Let~${G=(V,w)}$ be an undirected weighted graph.
Fix~${u\in V}$ with~${\degreeGraph{u}{G}>0}$, and let~${S\subseteq V\setminus\{u\}}$.
Then we can determine whether
\[
\weightVertexSet{u}{S}
\;<\; \frac{\degreeGraph{u}{G}}{2}
\]
using at most one cut comparison.
\end{lemma}

\begin{proof}
Let~${d\coloneqq\degreeGraph{u}{G}>0}$. 
If~${S=\emptyset}$, then we know that~${\weightVertexSet{u}{S}=0<d/2}$ without making any comparisons.
Similarly, if~${S=V\setminus\{u\}}$, then we know that~${\weightVertexSet{u}{S}=d\geq d/2}$ without any comparisons. 
Otherwise,~$S$ and~${S\cup\{u\}}$ are nonempty proper subsets of~$V$. 

By~\Cref{lemma:cut-marginal},~${\cutGraphSet{G}{S\cup\{u\}}-\cutGraphSet{G}{S}=d-2\weightVertexSet{u}{S}}$, and thus~${\cutGraphSet{G}{S\cup\{u\}}>\cutGraphSet{G}{S}}$ if and only if~${\weightVertexSet{u}{S}<d/2}$. 
We can therefore determine whether~${\weightVertexSet{u}{S}<d/2}$ by comparing~${\cutGraphSet{G}{S\cup\{u\}}}$ with~${\cutGraphSet{G}{S}}$. 
\end{proof}

We now use this test to recover a set of edges incident to a given vertex~$u$ such that the unrecovered edges have total weight less than~$\B$.

\begin{lemma}[Recovering edges incident to one vertex]
\label{lemma:recovering-at-one-vertex}
Let~${G=(V,w)}$ be an undirected graph on~$n$ vertices with edge weights in~${\{0,1,\ldots,\B\}}$. Fix~${u\in V}$ with~${\degreeGraph{u}{G}>0}$.
Then using~${O(n)}$ comparisons, we can construct a set~$\uEdges$ of edges incident to~$u$ such that every edge in~$\uEdges$ has positive weight in~$G$ and the edges incident to~$u$ that are not in~$\uEdges$ have total weight
\[
\missingWeightSum{\uEdges}
 \;<\;
 \B.
\]
\end{lemma}

\begin{proof} 
Initialize~${\uEdges\coloneqq\emptyset}$, and let~${d\coloneqq\degreeGraph{u}{G}>0}$.
We construct~${\uEdges}$ by iterating twice over~${V\setminus\{u\}}$.

For the first iteration, fix an arbitrary order~${v_1,\ldots,v_{n-1}}$ of~${V\setminus\{u\}}$.
Starting with~${X=\emptyset}$, we process the~${v_i}$ in this order and maintain the invariant~${\weightVertexSet{u}{X}<d/2}$.
When the iteration reaches~${v_i}$, we determine whether~${\weightVertexSet{u}{X\cup\{v_i\}}<d/2}$ using~\Cref{lemma:half-incident-weight}. 
If~${\weightVertexSet{u}{X\cup\{v_i\}}<d/2}$, then add~${v_i}$ to~$X$. 
Otherwise, discard~${v_i}$ and add~${\{u,v_i\}}$ to~${\uEdges}$.

The final set satisfies~${\weightVertexSet{u}{X}<d/2}$, as this invariant is preserved at every step of the iteration.
Moreover, whenever we add~${\{u,v_i\}}$ to~$\uEdges$, we have ${\weightVertexSet{u}{X}<d/2}$ before considering~$v_i$ and ${\weightVertexSet{u}{X\cup\{v_i\}}\ge d/2}$.
Therefore,
\[
 \edgeWeight{u}{v_i}
 \;=\;
 \weightVertexSet{u}{X\cup\{v_i\}}
 -
 \weightVertexSet{u}{X}
 \;>\;
 0,
\]
so every edge added to~$\uEdges$ has positive weight in~$G$.
If~$v_i$ is added to~$X$, then the comparison query certifies that the invariant~${\weightVertexSet{u}{X}<d/2}$ continues to hold; if~$v_i$ is discarded, then~$X$ does not change.

Let~$\supportVerticesOne$ be the final vertex set after the first iteration. The invariant gives
\begin{align}
 \weightVertexSet{u}{V\setminus(\supportVerticesOne\cup\{u\})}
 \;=\;
 \weightVertexSet{u}{V\setminus\{u\}}
 -
 \weightVertexSet{u}{\supportVerticesOne}
 \;>\;
 d-d/2
 \;=\;
 d/2.
 \label{eq:total_weight_u_to_residual_set}
\end{align}
For the second iteration, start again with~${X=\emptyset}$. 
This time, process the vertices in~${V\setminus(\supportVerticesOne\cup\{u\})}$ first and the vertices in~$\supportVerticesOne$ second, using the same rule as in the first iteration. Let~$\supportVerticesTwo$ be the final vertex set. 
By~\eqref{eq:total_weight_u_to_residual_set}, the total weight from~$u$ to~${V\setminus(\supportVerticesOne\cup\{u\})}$ is greater than~$d/2$, so the invariant forces the second iteration to discard at least one vertex~${v_j\in V\setminus(\supportVerticesOne\cup\{u\})}$.

Immediately before~$v_j$ is considered, the current set satisfies~${X\subseteq \supportVerticesTwo\setminus \supportVerticesOne}$.
Since the edge weights are nonnegative,~${\weightVertexSet{u}{X}\le\weightVertexSet{u}{\supportVerticesTwo\setminus\supportVerticesOne}}$.
Moreover, because~$v_j$ is discarded,~${\weightVertexSet{u}{X\cup\{v_j\}}\ge d/2}$, and therefore 
\[
 d/2-\weightVertexSet{u}{\supportVerticesTwo\setminus\supportVerticesOne}
 \;\le\;
 d/2-\weightVertexSet{u}{X}
 \;\le\;
 \weightVertexSet{u}{X\cup\{v_j\}}-\weightVertexSet{u}{X}
 \;=\;
 \edgeWeight{u}{v_j}
 \;\le\;
 \B.
\]
Every vertex in~${V\setminus\{u\}}$ outside~${\supportVerticesOne\cap\supportVerticesTwo}$ is discarded in at least one iteration. 
Hence, every edge incident to~$u$ that is not in~$\uEdges$ has its other endpoint in~${\supportVerticesOne\cap\supportVerticesTwo}$.
It follows that 
\[
\missingWeightSum{\uEdges}
 \;\le\;
 \weightVertexSet{u}{\supportVerticesOne\cap\supportVerticesTwo}
 \;=\;
  \weightVertexSet{u}{\supportVerticesTwo}
  -  \weightVertexSet{u}{\supportVerticesTwo\setminus\supportVerticesOne} 
    \;<\; 
    \B.
\]
Each iteration uses at most~${n-1}$ comparisons, so the procedure uses at most~${2(n-1)=O(n)}$ comparisons.
\end{proof}

Applying~\Cref{lemma:recovering-at-one-vertex} to each vertex allows us to recover a subgraph of the support graph of~$G$ such that the total weight of the unrecovered edges is at most~${n\B/2}$.

\begin{corollary}[Recovering most of the support graph]
\label{lemma:support}
Let~${G=(V,w)}$ be an undirected graph on~${n\geq 2}$ vertices with edge weights in~${\{0,1,\ldots,\B\}}$, and suppose that~${\minimumGraphDegree{G}>0}$.
Using~${O(n^2)}$ comparisons, we can construct a known unweighted graph~${\supp=(V,\suppEdges)}$ with the following properties.
\begin{enumerate}[label=(\roman*)]
\item\label{item:positive-edge-weight} Every edge in~$\suppEdges$ has positive weight in~$G$.
\item\label{item:bounded-unrecodered-weight} For every~${u\in V}$, the unrecovered edges incident to~$u$ have total weight
\[
\missingWeightSum{\suppEdges}
 \;<\;
 \B.
\]
\end{enumerate} 
Consequently, the total weight of the unrecovered edges is~${O(n\B)}$.
\end{corollary}

\begin{proof}
For each~${u\in V}$, use~\Cref{lemma:recovering-at-one-vertex} to obtain a set~$\uEdges$ of edges incident to~$u$, and let~${\suppEdges \coloneqq \bigcup_{u\in V}\uEdges}$. 
The procedure in~\Cref{lemma:recovering-at-one-vertex} uses~${O(n)}$ comparisons for each vertex, so the total number of comparisons is~${O(n^2)}$. 
Every edge in~${\suppEdges}$ belongs to some set~${\uEdges}$ and therefore has positive weight by~\Cref{lemma:recovering-at-one-vertex}.
Moreover, for every~${u\in V}$,
\[
\missingWeightSum{\suppEdges}
 \;\le\;
\missingWeightSum{\uEdges}
 \;<\;
 \B.
\]
Summing this inequality over~${u\in V}$ counts every unrecovered edge twice, so the total unrecovered weight is less than~${n\B/2=O(n\B)}$.
\end{proof}

\subsection{Computing an expander decomposition of the recovered graph}\label{sec:expander}
To compute a suitable contraction of~$G$, we first use~\Cref{lemma:support} to recover an explicit unweighted graph~${\supp=(V,\suppEdges)}$. 
We then compute an expander decomposition of~$\supp$ in polynomial time using no additional comparisons. 
This is the only place where our algorithm explicitly needs to know~$\B$; otherwise, the value of~$\B$ is used only in the analysis. 
We use the expander decomposition of~Li and~Saranurak~\cite{LiSaranurak}. 
We note that there are other deterministic expander decompositions that would suffice for polynomial time~(see, e.g., references in~\cite{LiSaranurak})---we choose~\cite{LiSaranurak} primarily for convenience.  

For disjoint sets~${S,T\subseteq V}$, let~${\edgesBetween{\supp}{S}{T}}$ denote the set of edges in~$\supp$ with one endpoint in~$S$ and the other in~$T$.  
In the remainder of the appendix, fix an absolute constant~${\expanderConstant>144}$; it is enough to take~${\expanderConstant=145}$.

\begin{theorem}[Expander decomposition of the recovered support graph~{\cite{LiSaranurak}}]\label{thm:expander}
Let~${\supp=(V,\suppEdges)}$ be a known unweighted graph on~${n\ge2}$ vertices.
For any~${\B\in\N_{\ge 1}}$, we can find a partition ${V=V_1\sqcup\ldots\sqcup V_{\partitionIndex}}$ in polynomial time such that the following properties hold.
\begin{enumerate}[label=(\roman*)]
\item\label{item:expander-inside} For every~$i$ and every nonempty~${X\subsetneq V_i}$,
\[
\frac{
\bigl|\edgesBetween{\supp}{X}{V_i\setminus X}\bigr|
}{
\min\{|X|,|V_i\setminus X|\}
} 
 \;\ge\;
 \expanderConstant\B .
\]
\item\label{item:expander-boundary} At most~${\softO(n\B)}$ edges of~$\supp$ join different blocks.
\end{enumerate}
\end{theorem}

\begin{proof}
If~${\B\ge n}$ or~${|\suppEdges|\le1}$, then the partition into singleton blocks satisfies both properties. 
Suppose therefore that~${\B<n}$ and~${|\suppEdges|\ge2}$.

Let~$C_0$ be a sufficiently large absolute constant. 
Apply the deterministic expander decomposition of~\cite[Corollary~2.5]{LiSaranurak} with~${r=1}$,~${\varepsilon=1}$, edge weights~$1$, and demand~${\expanderDemand\coloneqq  \lceil \expanderConstant\B\log^{C_0}(n\B+2)\rceil}$ at every vertex.  
Then for every nonempty~${X\subsetneq V_i}$, the cited corollary gives
\[
 \bigl|\edgesBetween{\supp}{X}{V_i\setminus X}\bigr|
 \;\geq\;
 \frac{\expanderDemand\min\{|X|,|V_i\setminus X|\}}
      {\log^{O(1)}|\suppEdges|\cdot\log \expanderDemand}
      \;\geq\;
       \frac{\expanderConstant\B\log^{C_0}(n\B+2) \cdot \min\{|X|,|V_i\setminus X|\}}
      {\log^{O(1)}|\suppEdges|\cdot\log \expanderDemand}
      .
\]
Since~${\log|\suppEdges|=O(\log n)}$ and~${\log \expanderDemand = O(\log n)}$, we can choose~$C_0$ sufficiently large to absorb both the fixed polylogarithmic exponent and the hidden absolute constant, giving
\[
\bigl|\edgesBetween{\supp}{X}{V_i\setminus X}\bigr| 
\;\geq\;
\expanderConstant\B \min\{|X|,|V_i\setminus X|\},
\]
which proves~(i). Moreover,~\cite[Corollary~2.5]{LiSaranurak} shows that the number of edges between blocks is at most~${\varepsilon\sum_{v\in V}D=nD=\softO(n\B)}$, which proves~(ii). Finally, for~${r=\varepsilon=1}$ and~${U=\poly(n)}$, the running time in~\cite[Corollary~2.5]{LiSaranurak} is polynomial in~$n$.  
\end{proof}

\subsection{Choosing which vertices to contract}

To preserve minimum cuts, we only contract vertices whose incident weight lies mostly inside their expander block. 
Although we cannot determine~${\degreeGraph{v}{G}}$ or~${\weightVertexSet{v}{V\setminus V_i}}$ explicitly using comparisons, we can use~\Cref{lemma:half-incident-weight} and a small hash family to classify these vertices with a constant-factor gap. 
In the proof of~\Cref{lemma:contraction}, we then contract each set~${\contractionSet{i}}$ found by the next result into a supervertex.

\begin{lemma}[Classifying vertices by external incident weight]\label{lemma:external}
Let~${G=(V,w)}$ be an undirected weighted graph on~${n \geq 2}$ vertices with edge weights in~${\{0,1,\ldots,\B\}}$, and let~${V=V_1\sqcup\ldots\sqcup V_{\partitionIndex}}$ be a known partition.
Suppose that~${\degreeGraph{v}{G}\ge\expanderConstant\B}$ for every vertex~$v$ in every nonsingleton block.  
Then using~${O(n^2)}$ comparisons, we can compute a set~${\contractionSet{i}\subseteq V_i}$ for every nonsingleton block~$V_i$ such that the following properties hold for all~${v\in V_i}$.
\begin{enumerate}[label=(\roman*)]
\item\label{item:external-in} 
If~${v\in \contractionSet{i}}$, then ${\weightVertexSet{v}{V\setminus V_i}<\degreeGraph{v}{G}/3}$.
\item\label{item:external-out} 
If~${v\notin \contractionSet{i}}$, then ${\weightVertexSet{v}{V\setminus V_i}>\degreeGraph{v}{G}/6}$.
\end{enumerate}
\end{lemma}

\begin{proof}
Let~$\ell$ be the smallest integer with~${3^\ell\ge n}$, and identify each graph vertex with a distinct vector in~${\F_3^\ell}$.
We use the following standard pairwise independent hash family.
Let~${\hashFamily\coloneqq\F_3^\ell\times\F_3}$, and for each~${h=(a,b)\in\hashFamily}$, define~${h(x)\coloneqq\innerProd{a}{x}+b}$.
By the choice of~$\ell$, we have~${|\hashFamily|=3^{\ell+1}\le9n=O(n)}$.

Now fix a block~${V_i}$ with~${|V_i|\geq 2}$.
Fix a vertex~${v\in V_i}$, and write~${d\coloneqq\degreeGraph{v}{G}}$. As~$V_i$ is not a singleton block, we have~${d \geq \expanderConstant\B > 0}$. 
For each~${h\in\hashFamily}$, define
\begin{align}
     \hashQuerySet{h}
 \;\coloneqq\;
 (V\setminus V_i)
 \cup
 \{x\in V_i\setminus\{v\}\setSeparator h(x)=0\}.
 \label{eq:hash_query_set_definition}
\end{align} 
Then~${\hashQuerySet{h}\subseteq V\setminus\{v\}}$, so we can use~\Cref{lemma:half-incident-weight} to determine whether~${\weightVertexSet{v}{\hashQuerySet{h}}<d/2}$ using at most one comparison.  
Put~$v$ in~$\contractionSet{i}$ if and only if~${\weightVertexSet{v}{\hashQuerySet{h}}<d/2}$ for more than half of the functions~${h\in\hashFamily}$. 

Repeat this procedure for every vertex in every nonsingleton block.
Each vertex requires at most~${|\hashFamily|=O(n)}$ comparisons, so we can construct all of the sets~${\contractionSet{i}}$  using~${O(n^2)}$ comparisons and polynomial time. 

We now show that properties~(i) and~(ii) hold for the fixed vertex~$v$.  
Sample~$h$ uniformly at random from~$\hashFamily$.
This randomness is only for the analysis; the algorithm is deterministic and enumerates every pair in~$\hashFamily$.
By pairwise independence, for distinct graph vertices~$x$ and~$y$,
\[
 \pr{h(x)=0}
 \;=\;
 \frac{1}{3}
 \qquad\text{and}\qquad
 \pr{h(x)=h(y)=0}
 \;=\;
 \frac{1}{9},
\]
where the probabilities are over the choice of~$h$.
Let~$\sampledIncidentWeight \coloneqq \weightVertexSet{v}{\hashQuerySet{h}}$.
By~\eqref{eq:hash_query_set_definition}, 
\[
 \sampledIncidentWeight
 \;=\;
 \externalIncidentWeight
 +
 \sum_{x\in V_i\setminus\{v\}}
 \edgeWeight{v}{x}\, \indicator{h(x)=0}.
\]
Each vertex in~${V_i\setminus\{v\}}$ belongs to~${\hashQuerySet{h}}$ with probability~${1/3}$.
Therefore,
\begin{equation}
\ex{\sampledIncidentWeight}
\;=\;
\externalIncidentWeight+\frac13\left(d-\externalIncidentWeight\right)
\;=\;
\frac d3+\frac{2\externalIncidentWeight}{3}.
\label{eq:expected_sampled_incident_weight}
\end{equation} 
Each indicator has variance~${2/9}$, and the covariance terms vanish by pairwise independence, so  
\[
\var{\sampledIncidentWeight}
\;=\;
\frac29\sum_{x\in V_i\setminus\{v\}}\edgeWeight{v}{x}^2
\;\le\;
\frac{2\B}{9}\sum_{x\in V_i\setminus\{v\}}\edgeWeight{v}{x}
\;\le\;
\frac{2\B d}{9}.
\]
We first prove~(i). 
Suppose that~${\externalIncidentWeight\ge d/3}$.
Then by~\eqref{eq:expected_sampled_incident_weight},
\[
 \ex{\sampledIncidentWeight}-d/2
 \;=\;
 \frac{2\externalIncidentWeight}{3}-\frac{d}{6}
 \;\ge\;
 \frac{2}{3}\cdot \frac{d}{3} - \frac{d}{6}
 \;=\;
 \frac {d}{18}.
\]
Chebyshev's inequality and~${d\ge\expanderConstant\B}$ give

\[
 \pr{\sampledIncidentWeight<d/2}
 \;\le\;
 \pr{\left|\sampledIncidentWeight-\ex{\sampledIncidentWeight}\right|\ge \frac{d}{18}}
 \;\le\;
 \frac{\var{\sampledIncidentWeight}}{(d/18)^2}
 \;\le\;
 \frac{2\B d}{9} \cdot \frac{18^2}{d^2}
 \;=\;
 \frac{72\B}{d}
 \;\le\;
 \frac{72}{\expanderConstant}
 \;<\;
 \frac12.
\]
Since~${\sampledIncidentWeight=\weightVertexSet{v}{\hashQuerySet{h}}}$ and we chose~$h$ uniformly at random, it follows that fewer than half of the functions~${h\in\hashFamily}$ satisfy~${\weightVertexSet{v}{\hashQuerySet{h}}<d/2}$. Thus, if~${\externalIncidentWeight\ge d/3}$, then~${v\notin\contractionSet{i}}$. Taking the contrapositive proves~(i).

We next prove~(ii). Suppose that~${\externalIncidentWeight\le d/6}$. 
Then
\[
d/2-\ex{\sampledIncidentWeight}
\;=\;
\frac d6-\frac{2\externalIncidentWeight}{3}
\;\ge\;
\frac d{18}.
\]
A similar application of Chebyshev's inequality gives
\[
\pr{\sampledIncidentWeight\ge d/2}
\;\le\;
\pr{\left|\sampledIncidentWeight-\ex{\sampledIncidentWeight}\right|\ge \frac{d}{18}}
\;\le\;
\frac{\var{\sampledIncidentWeight}}{(d/18)^2} 
\;<\;
\frac12.
\]
Thus, if~$\externalIncidentWeight\le d/6$, then more than half of the functions~${h\in\hashFamily}$ satisfy~${\sampledIncidentWeight<d/2}$, so~${v\in\contractionSet{i}}$. Taking the contrapositive proves~(ii).
\end{proof}

\subsection{Constructing the contraction}
\label{sec:contraction}
We now combine the results above to prove~\Cref{lemma:contraction}.

\GraphContractionUsingComparisons*

\begin{proof} 
First, use~\Cref{lemma:support} to recover an unweighted graph~${\supp=(V,\suppEdges)}$ such that every edge of~$\supp$ has positive weight in~$G$ and the total weight of unrecovered edges is~${O(n\B)}$.
Next, apply the expander decomposition in~\Cref{thm:expander} to~$\supp$ to partition~$V$ into blocks~${V_1,\ldots,V_{\partitionIndex}}$.

Fix a nonsingleton block~${V_i}$ and a vertex~${v\in V_i}$.
Applying~\Cref{thm:expander}\ref{item:expander-inside} to~${\{v\}}$ gives
\[
\frac{
\bigl|\edgesBetween{\supp}{\{v\}}{V_i\setminus\{v\}}\bigr|
}{
\min\{|\{v\}|,|V_i\setminus\{v\}|\}
}
\;=\;
\bigl|\edgesBetween{\supp}{\{v\}}{V_i\setminus\{v\}}\bigr|
\;\ge\;
\expanderConstant\B,
\]
so~$v$ has at least~${\expanderConstant\B}$ neighbors in~${\supp[V_i]}$.
Every recovered edge has positive integer weight in~$G$, and thus~${\degreeGraph{v}{G}\ge\expanderConstant\B}$.
We may therefore apply~\Cref{lemma:external} to obtain a set~${\contractionSet{i}\subseteq V_i}$ for every nonsingleton block~$V_i$.
For each remaining singleton block~$V_i$, set~${\contractionSet{i}\coloneqq\emptyset}$, and define
\[
\uncontractedVertices
\;\coloneqq\;
V\setminus\bigcup_{i=1}^{\partitionIndex}\contractionSet{i}.
\] 
Let~$\partition$ consist of every nonempty~${\contractionSet{i}}$ together with the singleton sets~${\{v\}}$ for each~${v\in\uncontractedVertices}$.
These sets are pairwise disjoint and cover~$V$, so~$\partition$ is a partition. 
We show that~$\partition$ satisfies the required properties.

For the query complexity, the procedures in~\Cref{lemma:support,lemma:external} each use~${O(n^2)}$ comparisons.  
The expander decomposition and all remaining computation take polynomial time and use no further comparisons.
Therefore we can construct~$\partition$ using~${O(n^2)}$ comparisons and polynomial time.

We first prove property~(i).
Let~${\delta\coloneqq\minimumGraphDegree{G}}$.
If some singleton is a minimum cut, then~(i) holds. 
Suppose therefore that no singleton is a minimum cut, and fix an arbitrary minimum cut~$S$. 
Every singleton has cut value at least~$\delta$, so~${\cutGraphSet{G}{S}<\delta}$ and both~$S$ and~${V\setminus S}$ contain at least two vertices.

Suppose for a contradiction that~$S$ is \emph{not} a union of blocks of~$\partition$. 
Every vertex in~${\uncontractedVertices}$ forms a singleton block, so there must be some~$i$ such that~${\contractionSet{i}}$ contains vertices in both~$S$ and~${V\setminus S}$.
Choose~${x\in\contractionSet{i}\cap S}$ and~${y\in\contractionSet{i}\setminus S}$.
By~\Cref{lemma:external}\ref{item:external-in},~${\weightVertexSet{x}{V\setminus V_i}<\degreeGraph{x}{G}/3}$. 
Moreover, as~$S$ is a minimum cut, it follows by~\Cref{lemma:cut-marginal} that~${2\weightVertexSet{x}{S\setminus\{x\}}-\degreeGraph{x}{G}\;=\;\cutGraphSet{G}{S\setminus\{x\}}-\cutGraphSet{G}{S}\;\ge\;0}$. Therefore~${\weightVertexSet{x}{S\setminus\{x\}}\ge\degreeGraph{x}{G}/2}$, and it follows that
\[ 
\weightVertexSet{x}{(S\cap V_i)\setminus\{x\}}
 \;\ge\;
\weightVertexSet{x}{S\setminus\{x\}}
-
\weightVertexSet{x}{V\setminus V_i}  
 \;>\;
\frac{\degreeGraph{x}{G}}{2}
-
\frac{\degreeGraph{x}{G}}{3}
\;=\;
\frac{\degreeGraph{x}{G}}{6}. 
\]
By~\Cref{lemma:support}\ref{item:bounded-unrecodered-weight}, the unrecovered edges incident to~$x$ have total weight less than~$\B$. 
The recovered edges from~$x$ to~${(S\cap V_i)\setminus\{x\}}$ therefore have total weight at least
\[ 
\weightVertexSet{x}{(S\cap V_i)\setminus\{x\}}-\B\;>\;\frac{\degreeGraph{x}{G}}{6}-\B\;>\;\frac{\degreeGraph{x}{G}}{12},
\]
where the last inequality uses~${\degreeGraph{x}{G}\ge\expanderConstant\B}$ and~${\expanderConstant>144}$.
Each recovered edge has weight at most~$\B$, so dividing both sides by~$\B$ gives
\[
|S\cap V_i|
\;\geq\;
\bigl|(S\cap V_i)\setminus\{x\}\bigr|\;>\;\frac{\degreeGraph{x}{G}}{12 \B}
\;\geq\;
\frac{\delta}{12\B}.
\]  
Applying the same argument to~$y$ and~${V\setminus S}$ gives
\[
|V_i\setminus S|
\;>\;
\frac{\degreeGraph{y}{G}}{12\B}
\;\ge\;
\frac{\delta}{12\B}.
\]
Since~${S\cap V_i}$ is a nonempty proper subset of~$V_i$, we can apply~\Cref{thm:expander}\ref{item:expander-inside} to~${S\cap V_i}$, giving
\[
\bigl|\edgesBetween{\supp}{S\cap V_i}{V_i\setminus S}\bigr|\;\ge\;\expanderConstant\B\min\{|S\cap V_i|,|V_i\setminus S|\}\;>\;\frac{\expanderConstant\delta}{12}\;>\;\delta.
\]
Every edge in~$\supp$ has positive weight in~$G$, and every edge counted above crosses the cut~$S$.
Hence~${\cutGraphSet{G}{S}>\delta}$, contradicting our assumption that~${\cutGraphSet{G}{S}<\delta}$.
Therefore~$S$ is a union of blocks of~$\partition$. 
Since~$S$ was an arbitrary minimum cut, every minimum cut is a union of blocks of~$\partition$, which proves~(i).

We next prove property~(ii).
Let~$\crossingWeight$ be the total weight in~$G$ of edges joining different expander blocks.
By~\Cref{thm:expander}\ref{item:expander-boundary}, at most~${\softO(n\B)}$ recovered edges join different blocks, and each such edge has weight at most~$\B$.
Moreover, by~\Cref{lemma:support}, the total weight of the unrecovered edges is~${O(n\B)}$.
Therefore
\[
\crossingWeight\;\le\;\B\cdot\softO(n\B)+O(n\B)\;=\;\softO(n\B^2).
\]
Now fix~${v\in\uncontractedVertices}$, and let~$V_i$ be the expander block containing~$v$.
If~$V_i$ is not a singleton, then~${v\notin\contractionSet{i}}$, so~\Cref{lemma:external}\ref{item:external-out} gives~${\degreeGraph{v}{G}<6\weightVertexSet{v}{V\setminus V_i}}$.
If~$V_i$ is a singleton, then~${\degreeGraph{v}{G}=\weightVertexSet{v}{V\setminus V_i}}$.
Thus, in either case,~${\degreeGraph{v}{G}\le6\weightVertexSet{v}{V\setminus V_i}}$. 
Summing over all~${v\in\uncontractedVertices}$ gives
\[
\sum_{v\in\uncontractedVertices}\degreeGraph{v}{G}
\;\le\;
6\sum_{i=1}^{\partitionIndex}\sum_{v\in\uncontractedVertices\cap V_i}\weightVertexSet{v}{V\setminus V_i}
\;\leq\;
6\sum_{i=1}^{\partitionIndex}\sum_{v\in V_i}\weightVertexSet{v}{V\setminus V_i}
\;=\;
12\crossingWeight
\;=\;
\softO(n\B^2).
\]
It remains to bound the total edge weight of~${\contractedG}$. 
Every edge of~${\contractedG}$ comes from an edge of~$G$ that either joins different expander blocks or lies inside one expander block and has an endpoint in~${\uncontractedVertices}$. 
The total edge weight of~${\contractedG}$ is therefore at most
\[
\crossingWeight
+
\sum_{v\in\uncontractedVertices}\degreeGraph{v}{G}
\;\le\;
13\crossingWeight
\;=\;
\softO(n\B^2).
\] 
This proves property~(ii).
\end{proof}

  %%%%%%%%%%%%%%%%%%
\section{Proofs omitted from \texorpdfstring{\Cref{sec:lower}}{Section 4}}
\label{app:lower_bound_proofs} 
%%%%%%%%%%%%%%%%%%
This appendix contains the details of the lower-bound construction omitted from~\Cref{sec:lower}. 
We briefly recall the definitions from~\Cref{subsec:path-construction}. 
Let~$n$ be even and sufficiently large, write~${n=2k}$, and let~${\lbSpacing\coloneqq\lfloor k/(\ell+1)\rfloor}$. 
The indices of the edges that we perturb are given by
\[
\lbPerturbIndices
\;\coloneqq\;
\{\lbSpacing,\,2\lbSpacing,\,\ldots,\,\ell\lbSpacing\}
\cup
\{n-\ell\lbSpacing,\,\ldots,\,n-2\lbSpacing,\,n-\lbSpacing\}. 
\]
Let~${\lbBoundarySet\subseteq[n-1]}$ and write~${\lbBoundarySet=\{b_1<\ldots<b_m\}}$.
Define its signed indicator~${\lbBoundaryVector{\lbBoundarySet}\in\{-1,0,1\}^{n-1}}$ by
\[
\lbBoundaryVectorEntry{\lbBoundarySet}{i}
\;\coloneqq\;
\begin{cases}
(-1)^{r}, & \text{if }i=b_r\text{ for some }r\in[m],\\
0, & \text{if }i\notin\lbBoundarySet.
\end{cases}
\]
For a symmetric~$(n-1)\times(n-1)$ matrix~${\lbEncodingMatrix}$, let
\[
\lbQuadraticForm{\lbEncodingMatrix}{\lbBoundarySet}
\;\coloneqq\;
(\lbBoundaryVector{\lbBoundarySet})^\top
\lbEncodingMatrix
\lbBoundaryVector{\lbBoundarySet}
\;=\;
\sum_{r=1}^{m}
\sum_{s=1}^{m}
(-1)^{r+s}\lbEncodingMatrix_{b_r b_s}.
\]

\subsection{Constructing the matrix \texorpdfstring{$\lbEncodingMatrix$}{M}}
\label{subsec:constructing_the_encoding_matrix_M}

This subsection contains the proof of~\Cref{lemma:encoding-perturbed-edges}.
We use the following two symmetric matrices. 
Define~${\lbPairMatrix{\lbBoundarySet} \in \{-1,0,1\}^{2\ell \times 2\ell}}$, with rows and columns indexed by~$\lbPerturbIndices$, by 
\[
\lbPairMatrixEntry{\lbBoundarySet}{p}{q}
\;\coloneqq\;
\begin{cases}
\lbBoundaryVectorEntry{\lbBoundarySet}{p}\,
\lbBoundaryVectorEntry{\lbBoundarySet}{q},
& \text{if }p\neq q,\\
0,
& \text{if }p=q.
\end{cases}
\] 
Additionally, define~${\lbSingletonMatrix\in\Z^{(n-1)\times(n-1)}}$ by
\[
\lbSingletonMatrixEntry{i}{j}
\;\coloneqq\;
\begin{cases}
-\bigl(n-|i-j|\bigr)(2n^2+p),
& \text{if }i\neq j\text{ and }\{i,j\}\cap\lbPerturbIndices=\{p\},\\
0,
& \text{otherwise}.
\end{cases}
\]
We first show that the pair~${\lbCertificatePair{\lbBoundarySet}}$ determines~${\lbBoundarySet\cap\lbPerturbIndices}$.

\begin{lemma}[The pair~${\lbCertificatePair{\lbBoundarySet}}$ determines~${\lbBoundarySet\cap\lbPerturbIndices}$] 
\label{lemma:perturbed-edge-certificate}
For every~${\lbBoundarySet,\lbOtherBoundarySet\subseteq[n-1]}$ with~${|\lbBoundarySet|,|\lbOtherBoundarySet|\geq2}$, if~${\lbCertificatePair{\lbBoundarySet}=\lbCertificatePair{\lbOtherBoundarySet}}$, then~${\lbBoundarySet\cap\lbPerturbIndices=\lbOtherBoundarySet\cap\lbPerturbIndices}$.  
\end{lemma}

\begin{proof}
First, suppose that~${|\lbBoundarySet\cap\lbPerturbIndices|\geq2}$. Then~${\lbPairMatrix{\lbBoundarySet}}$ determines~${\lbBoundarySet\cap\lbPerturbIndices}$, as~${p\in\lbBoundarySet\cap\lbPerturbIndices}$ if and only if~${\lbPairMatrixEntry{\lbBoundarySet}{p}{q}\neq0}$ for some~${q\in\lbPerturbIndices\setminus\{p\}}$. 
Suppose therefore that~${|\lbBoundarySet\cap\lbPerturbIndices|\leq1}$. Then every entry of~${\lbPairMatrix{\lbBoundarySet}}$ is~$0$, so we consider the quadratic form~${\lbQuadraticForm{\lbSingletonMatrix}{\lbBoundarySet}}$.

If~${\lbBoundarySet\cap\lbPerturbIndices=\emptyset}$, then~${\lbQuadraticForm{\lbSingletonMatrix}{\lbBoundarySet}=0}$.
Otherwise, suppose that~${\lbBoundarySet\cap\lbPerturbIndices=\{p\}}$ for some~${p=b_r}$. 
The only nonzero terms in~${\lbQuadraticForm{\lbSingletonMatrix}{\lbBoundarySet}}$ are those involving~$p$.
Since~${\lbBoundaryVectorEntry{\lbBoundarySet}{{b_s}}=(-1)^s}$, expanding these terms gives
\begin{align}
\frac{\lbQuadraticForm{\lbSingletonMatrix}{\lbBoundarySet}}
{2(2n^2+p)}
\;=\;
\sum_{s=1}^{r-1}
(-1)^{r-s-1}
\bigl(n-p+b_s\bigr)
+
\sum_{s=r+1}^{m}
(-1)^{s-r-1}
\bigl(n-b_s+p\bigr).
\label{eq:quadratic_form_pairs}
\end{align}
Since~${b_1<\ldots<b_{r-1}<p}$, the terms in the first sum alternate in sign and decrease in absolute value when read from~${s=r-1}$ down to~${s=1}$. Thus, if~${r>1}$, then the summand with~${s=r-1}$ is positive, and the first sum satisfies
\[
0
\;<\;
\sum_{s=1}^{r-1}
(-1)^{r-s-1}
\bigl(n-p+b_s\bigr)
\;\leq\;
n-p+b_{r-1}
\;\leq\;
n-1.
\] 
Similarly, we have~${p<b_{r+1}<\ldots<b_m}$, so if~${r<m}$, then the second sum satisfies
\[
0
\;<\;
\sum_{s=r+1}^{m}
(-1)^{s-r-1}
\bigl(n-b_s+p\bigr)
\;\leq\;
n-b_{r+1}+p
\;\leq\;
n-1.
\]
Since~${m\geq2}$, at least one of these two sums is nonempty. 
Therefore their sum is an integer between~$1$ and~${2n-2}$, and it follows by~\eqref{eq:quadratic_form_pairs} that~${\lbQuadraticForm{\lbSingletonMatrix}{\lbBoundarySet}=2(2n^2+p)t}$ for some integer~${1\leq t<2n}$.  
Now suppose that~${2(2n^2+p)t=2(2n^2+p')t'}$ for some~${p,p'\in\lbPerturbIndices}$ and~${1\leq t,t'<2n}$.
Rearranging gives~${2n^2(t-t')=p't'-pt}$.
Since~${1\leq p,p'<n}$ and~${1\leq t,t'<2n}$, we have~${|p't'-pt|<2n^2}$.
It follows that~${t=t'}$, and hence~${p=p'}$. Therefore, if~${\lbBoundarySet\cap\lbPerturbIndices=\{p\}}$, then~${\lbQuadraticForm{\lbSingletonMatrix}{\lbBoundarySet}}$ determines~$p$. 

Combining the cases above, the pair~$\lbCertificatePair{\lbBoundarySet}$ determines~${\lbBoundarySet\cap\lbPerturbIndices}$. Thus for every~${\lbBoundarySet,\lbOtherBoundarySet\subseteq[n-1]}$ with~${|\lbBoundarySet|,|\lbOtherBoundarySet|\geq2}$, if~${\lbCertificatePair{\lbBoundarySet}=\lbCertificatePair{\lbOtherBoundarySet}}$, then~${\lbBoundarySet\cap\lbPerturbIndices=\lbOtherBoundarySet\cap\lbPerturbIndices}$. 
\end{proof}

We now use the pair from~\Cref{lemma:perturbed-edge-certificate} to prove~\Cref{lemma:encoding-perturbed-edges}.

\encodingCrossingsWithPerturbedPathIndices*

\begin{proof}
Let~${\lbCertificate{\lbBoundarySet}\coloneqq\lbCertificatePair{\lbBoundarySet}}$ be the pair from~\Cref{lemma:perturbed-edge-certificate}. 
There are at most~${3^{2\ell}}$ possible matrices~${\lbPairMatrix{\lbBoundarySet}}$, as the entries are determined by the~${2\ell}$ values~${\lbBoundaryVectorEntry{\lbBoundarySet}{p}\in\{-1,0,1\}}$ for~${p\in\lbPerturbIndices}$.
Also, every entry of~${\lbSingletonMatrix}$ has magnitude smaller than~${2n^3}$, so~${|\lbQuadraticForm{\lbSingletonMatrix}{\lbBoundarySet}|<2n^5}$.
Hence there are at most
\[
\lbCertificateCount
\;\coloneqq\;
3^{2\ell}\cdot4n^5
\;=\;
\lbWeightRange
\]
possible values of~${\lbCertificate{\lbBoundarySet}}$. 
For each pair~${p,q\in\lbPerturbIndices}$ with~${p<q}$, let~${\lbProjectionCoeff{p}{q}=\lbProjectionCoeff{q}{p}\in\{0,1,\ldots,\lbCertificateCount^2\}}$ be a coefficient to be fixed later. 
Define~${\lbProjectionMatrix\in\Z^{(n-1)\times(n-1)}}$ by
\[
\lbProjectionMatrix_{ij}
\;\coloneqq\;
\begin{cases}
\lbProjectionCoeff{i}{j},
& \text{if }i\neq j\text{ and }i,j\in\lbPerturbIndices,\\
0,
& \text{otherwise},
\end{cases}
\]
and let~${\lbInitialEncodingMatrix\coloneqq\lbSingletonMatrix+\lbProjectionMatrix}$.
Then~$\lbInitialEncodingMatrix$ is symmetric with zero diagonal. 
Moreover,~${|\lbPerturbIndices|=2\ell<n-1}$, so there exist~${p\in\lbPerturbIndices}$ and~${j\in[n-1]\setminus\lbPerturbIndices}$.
For these indices,~${\lbProjectionMatrix_{pj}=0}$ and~${\lbSingletonMatrixEntry{p}{j}\neq0}$, so~${\maxNorm{\lbInitialEncodingMatrix}\geq1}$.
Also,
\[
\maxNorm{\lbInitialEncodingMatrix}
\;\leq\;
2n^3+\lbCertificateCount^2
\;=\;
\lbWeightRange.
\]
Let~${\lbEncodingScale\coloneqq 4n^2\maxNorm{\lbInitialEncodingMatrix}}$, and define the final matrix~$\lbEncodingMatrix$ by
\begin{align}
\lbEncodingMatrix_{ij}
\;&\coloneqq\;
i(n-j) \lbEncodingScale 
+
2\lbInitialEncodingMatrix_{ij}
\qquad\text{for }1\leq i<j\leq n-1,
\label{eq:final_encoding_matrix}
\end{align}
and set~${\lbEncodingMatrix_{ji}\coloneqq\lbEncodingMatrix_{ij}}$ and~${\lbEncodingMatrix_{ii}\coloneqq0}$.
Thus~$\lbEncodingMatrix$ is a symmetric integer matrix with zero diagonal. 
Since~${i(n-j)<n^2}$, we have 
\[
\maxNorm{\lbEncodingMatrix} 
\;\leq\;
(4n^4+2)\maxNorm{\lbInitialEncodingMatrix} 
\;=\;
\lbWeightRange.
\]
Also,~${\lbEncodingMatrix_{1,n-1}\geq\lbEncodingScale-2\maxNorm{\lbInitialEncodingMatrix}>0}$,
so~${\maxNorm{\lbEncodingMatrix}\geq1}$.

It remains to choose the coefficients~${\lbProjectionCoeff{p}{q}}$.
Fix~${\lbBoundarySet,\lbOtherBoundarySet\subseteq[n-1]}$ with~${|\lbBoundarySet|,|\lbOtherBoundarySet|\geq2}$, and suppose that~${\lbCertificate{\lbBoundarySet}\neq\lbCertificate{\lbOtherBoundarySet}}$.  
By the definition of~$\lbInitialEncodingMatrix$,
\[
\lbQuadraticForm{\lbInitialEncodingMatrix}{\lbBoundarySet}
\;=\;
\lbQuadraticForm{\lbSingletonMatrix}{\lbBoundarySet}
+
2
\sum_{\substack{p,q\in\lbPerturbIndices\\p<q}}
\lbProjectionCoeff{p}{q}
\lbPairMatrixEntry{\lbBoundarySet}{p}{q}.
\]
Subtracting this identity for~$\lbOtherBoundarySet$ from the identity for~$\lbBoundarySet$ gives
\begin{align}
\lbQuadraticForm{\lbInitialEncodingMatrix}{\lbBoundarySet}
\;=\;
\lbQuadraticForm{\lbInitialEncodingMatrix}{\lbOtherBoundarySet}
\qquad\iff\qquad
\lbQuadraticForm{\lbSingletonMatrix}{\lbOtherBoundarySet}
-
\lbQuadraticForm{\lbSingletonMatrix}{\lbBoundarySet}
\;=\;
2
\sum_{\substack{p,q\in\lbPerturbIndices\\p<q}}
\lbProjectionCoeff{p}{q}
\left(
\lbPairMatrixEntry{\lbBoundarySet}{p}{q}
-
\lbPairMatrixEntry{\lbOtherBoundarySet}{p}{q}
\right).
\label{eq:iff_condition_for_equal_quadratic_forms}
\end{align}
If~${\lbPairMatrix{\lbBoundarySet}=\lbPairMatrix{\lbOtherBoundarySet}}$, then as~${\lbCertificate{\lbBoundarySet}\neq\lbCertificate{\lbOtherBoundarySet}}$, we must have~${\lbQuadraticForm{\lbSingletonMatrix}{\lbBoundarySet} \neq \lbQuadraticForm{\lbSingletonMatrix}{\lbOtherBoundarySet}}$. It follows that~${\lbQuadraticForm{\lbInitialEncodingMatrix}{\lbBoundarySet} \neq \lbQuadraticForm{\lbInitialEncodingMatrix}{\lbOtherBoundarySet}}$ for every choice of~${\lbProjectionCoeff{p}{q}}$. 
Suppose instead that~${\lbPairMatrix{\lbBoundarySet}\neq\lbPairMatrix{\lbOtherBoundarySet}}$, and choose~${p'<q'}$ in~$\lbPerturbIndices$ such that~${\lbPairMatrixEntry{\lbBoundarySet}{p'}{q'}\neq\lbPairMatrixEntry{\lbOtherBoundarySet}{p'}{q'}}$.
Suppose that we fix the~${\binom{2\ell}{2}-1}$ coefficients~${\lbProjectionCoeff{p}{q}}$ with~${(p,q)\neq(p',q')}$. 
Then by~\eqref{eq:iff_condition_for_equal_quadratic_forms}, if~${\lbQuadraticForm{\lbInitialEncodingMatrix}{\lbBoundarySet}=\lbQuadraticForm{\lbInitialEncodingMatrix}{\lbOtherBoundarySet}}$, then these coefficients determine at most one possible value of~${\lbProjectionCoeff{p'}{q'}}$.  
Thus, for these fixed sets~$\lbBoundarySet$ and~$\lbOtherBoundarySet$, at most 
\[
(\lbCertificateCount^2+1)^{\binom{2\ell}{2}-1}
\]
choices of all of the coefficients satisfy~${\lbQuadraticForm{\lbInitialEncodingMatrix}{\lbBoundarySet}=\lbQuadraticForm{\lbInitialEncodingMatrix}{\lbOtherBoundarySet}}$. 

There are fewer than~${\lbCertificateCount^2/2}$ unordered pairs of distinct values of~$\lbCertificate{\lbBoundarySet}$.
Hence the total number of coefficient choices for which there exist~${\lbBoundarySet,\lbOtherBoundarySet}$ with~${\lbCertificate{\lbBoundarySet}\neq \lbCertificate{\lbOtherBoundarySet}}$ and~${\lbQuadraticForm{\lbInitialEncodingMatrix}{\lbBoundarySet}=\lbQuadraticForm{\lbInitialEncodingMatrix}{\lbOtherBoundarySet}}$ is at most 
\[ 
\frac{\lbCertificateCount^2}{2}
\cdot 
(\lbCertificateCount^2+1)^{\binom{2\ell}{2}-1}
\;<\;
(\lbCertificateCount^2+1)^{\binom{2\ell}{2}}.
\]
The right-hand side is equal to the total number of possible choices of the~${\lbProjectionCoeff{p}{q}}$. 
Thus, by the pigeonhole principle, there exists a choice of the~${\lbProjectionCoeff{p}{q}}$ such that, for all~${\lbBoundarySet,\lbOtherBoundarySet\subseteq[n-1]}$ with~${|\lbBoundarySet|,|\lbOtherBoundarySet|\geq 2}$, 
\begin{align}
    \lbQuadraticForm{\lbInitialEncodingMatrix}{\lbBoundarySet}
\;=\;
\lbQuadraticForm{\lbInitialEncodingMatrix}{\lbOtherBoundarySet}
\qquad\implies\qquad
\lbCertificate{\lbBoundarySet}
\;=\;
\lbCertificate{\lbOtherBoundarySet}.
\label{eq:quadratic_form_equals_implies_equal_certificates}
\end{align} 
Fix such a choice.
Suppose that~${\lbQuadraticForm{\lbEncodingMatrix}{\lbBoundarySet}=\lbQuadraticForm{\lbEncodingMatrix}{\lbOtherBoundarySet}}$.
By~\eqref{eq:final_encoding_matrix}, we have~${2\lbQuadraticForm{\lbInitialEncodingMatrix}{\lbBoundarySet} \;\equiv\; 2\lbQuadraticForm{\lbInitialEncodingMatrix}{\lbOtherBoundarySet} \pmod{\lbEncodingScale}}$. 
Since~$\lbInitialEncodingMatrix$ has zero diagonal,~${|\lbQuadraticForm{\lbInitialEncodingMatrix}{\lbBoundarySet}|<n^2\maxNorm{\lbInitialEncodingMatrix}}$ for every~${\lbBoundarySet\subseteq[n-1]}$. 
Thus the difference of the two sides of the congruence has magnitude smaller than~${\lbEncodingScale }$, so
${\lbQuadraticForm{\lbInitialEncodingMatrix}{\lbBoundarySet}=\lbQuadraticForm{\lbInitialEncodingMatrix}{\lbOtherBoundarySet}}$. 
It follows by~\eqref{eq:quadratic_form_equals_implies_equal_certificates} that~${\lbCertificate{\lbBoundarySet}=\lbCertificate{\lbOtherBoundarySet}}$.
Hence by~\Cref{lemma:perturbed-edge-certificate},~${\lbBoundarySet\cap\lbPerturbIndices=\lbOtherBoundarySet\cap\lbPerturbIndices}$.

As~$\lbBoundarySet$ and~$\lbOtherBoundarySet$ were arbitrary, for every~${\lbBoundarySet,\lbOtherBoundarySet\subseteq[n-1]}$ with~${|\lbBoundarySet|,|\lbOtherBoundarySet|\geq2}$,
\[
\lbQuadraticForm{\lbEncodingMatrix}{\lbBoundarySet}
\;=\;\lbQuadraticForm{\lbEncodingMatrix}{\lbOtherBoundarySet}
\qquad\implies\qquad
\lbBoundarySet\cap\lbPerturbIndices
\;=\;
\lbOtherBoundarySet\cap\lbPerturbIndices.
\]
For the final property, set~${\lbEncodingMatrix_{0j}\coloneqq\lbInitialEncodingMatrix_{0j}\coloneqq0}$ and~${\lbEncodingMatrix_{i n}\coloneqq\lbInitialEncodingMatrix_{i n}\coloneqq0}$ whenever these entries occur.
Then, for~${1\leq r<s-1\leq n-1}$,~\eqref{eq:final_encoding_matrix} gives 
\begin{align*}
\lbEncodingMatrix_{r,s-1}
+
\lbEncodingMatrix_{r-1,s}
-
\lbEncodingMatrix_{r-1,s-1}
-
\lbEncodingMatrix_{r,s}
&\;=\;
\lbEncodingScale 
+
2\left(
\lbInitialEncodingMatrix_{r,s-1}
+
\lbInitialEncodingMatrix_{r-1,s}
-
\lbInitialEncodingMatrix_{r-1,s-1}
-
\lbInitialEncodingMatrix_{r,s}
\right).
\end{align*}
The term in parentheses has magnitude at most~${4\maxNorm{\lbInitialEncodingMatrix}}$.
Hence the right-hand side is at least~${(4n^2-8)\maxNorm{\lbInitialEncodingMatrix}>0}$ and is an even integer.
This proves all the required properties of~$\lbEncodingMatrix$.
\end{proof}

\subsection{Properties of the base graphs}
\label{subsec:properties_of_the_base_graph}
We first restate the definition of~$\lbBaseGraph{\tau}$, then prove~\Cref{lemma:base_graph_properties}.

\baseGraphDefinition*

We prove the following result.

\baseGraphProperties*

We prove each of the three properties in~\Cref{lemma:base_graph_properties} individually. 
In the remainder of this section, fix a path~${\tau=(v_1,\ldots,v_n)\in\lbOrderings}$ and let~${\lbBaseGraph{\tau}=(V,w)}$ be the base graph in~\Cref{def:base_graph}. 

\begin{lemma}[Edge weights of the base graph; \Cref{lemma:base_graph_properties}\ref{item:base_graph_edge_weights}]
\label{lemma:base_graph_edge_weights}
Every edge weight in~$\lbBaseGraph{\tau}$ is a positive even integer at most~${\lbWeightRange}$.
\end{lemma}

\begin{proof}
We defined~${\lbWeightScale\coloneqq2n^2\maxNorm{\lbEncodingMatrix}}$.
By~\Cref{lemma:encoding-perturbed-edges},~${1\leq\maxNorm{\lbEncodingMatrix}\leq\lbWeightRange}$, so~${\lbWeightScale}$ is a positive even integer with~${\lbWeightScale=\lbWeightRange}$. 
Thus by~\eqref{eq:nonpath_edge_weights}, each nonpath edge weight~$\lbEdgeWeight{r}{s}$ is a positive even integer and~${\lbEdgeWeight{r}{s}\leq4\maxNorm{\lbEncodingMatrix}=2\lbWeightScale/n^2=\lbWeightRange}$. 

For a path edge~${\{v_i,v_{i+1}\}}$, the two sums in~\eqref{eq:path_edge_weight} contain exactly the nonpath edges crossing the prefix cut~${\lbPrefix{\tau}{i}}$.
There are~${i(n-i)-1<n^2/4}$ such edges, so their total weight is smaller than~${\lbWeightScale/2}$.
Thus~${0<\lbEdgeWeight{i}{i+1}<n^3\lbWeightScale=\lbWeightRange}$.
Each summand in~\eqref{eq:path_edge_weight} is an even integer, so~$\lbEdgeWeight{i}{i+1}$ is also an even integer.
\end{proof}

We now compute the cut values in~$\lbBaseGraph{\tau}$ in terms of the quadratic form~$\lbQuadraticForm{\lbEncodingMatrix}{\cdot}$.

\begin{lemma}[Cut value from the path boundary]
\label{lemma:base_graph_cut_function}
For every nonempty proper set~${X\subsetneq V}$,
\begin{align*}
\cutGraphSet{\lbBaseGraph{\tau}}{X}
\;=\;
|\lbBoundary{\tau}{X}|n^3\lbWeightScale
+
\lbQuadraticForm{\lbEncodingMatrix}{\lbBoundary{\tau}{X}}. 
\end{align*} 
Moreover,~${\left|\lbQuadraticForm{\lbEncodingMatrix}{\lbBoundary{\tau}{X}}\right|<\lbWeightScale/2}$. 
\end{lemma}

\begin{proof}
Write~${\lbBoundary{\tau}{X}=\{b_1<\ldots<b_m\}}$ and, for every nonpath edge~${\{v_r,v_s\}}$, let
\[
\lbBoundaryCount
\;\coloneqq\;
|\lbBoundary{\tau}{X}\cap\{r,\ldots,s-1\}|.
\] 
The edge~${\{v_r,v_s\}}$ crosses~$X$ if and only if~${\lbBoundaryCount}$ is odd.
Thus, by the definition of the cut function,
\[
\cutGraphSet{\lbBaseGraph{\tau}}{X}
\;=\;
\sum_{i\in\lbBoundary{\tau}{X}}
\lbEdgeWeight{i}{i+1}
+
\sum_{r=1}^{n-2}
\sum_{s=r+2}^{n}
\indicator{\lbBoundaryCount\text{ is odd}}
\lbEdgeWeight{r}{s}.
\]
Summing~\eqref{eq:path_edge_weight} for each~${i \in \lbBoundary{\tau}{X}}$ gives
\[
\sum_{i\in\lbBoundary{\tau}{X}}
\lbEdgeWeight{i}{i+1}
\;=\;
mn^3\lbWeightScale
-
\sum_{r=1}^{n-2}
\sum_{s=r+2}^{n}
\lbBoundaryCount\lbEdgeWeight{r}{s}.
\]
Since~${-t+\indicator{t\text{ is odd}}=-2\lfloor t/2\rfloor}$ for every nonnegative integer~$t$, it follows that 
\begin{align}
\cutGraphSet{\lbBaseGraph{\tau}}{X}
\;&=\;
m n^3\lbWeightScale
-
2
\sum_{r=1}^{n-2}
\sum_{s=r+2}^{n}
\left\lfloor\frac{\lbBoundaryCount}{2}\right\rfloor
\lbEdgeWeight{r}{s}.
\label{eq:cut_value_with_pairs}
\end{align}
By~\eqref{eq:nonpath_edge_weights}, for every~${1\leq\alpha<\beta\leq m}$, 
\begin{align}
\sum_{r=1}^{b_\alpha}
\sum_{s=b_\beta+1}^{n}
\lbEdgeWeight{r}{s}
\;=\;
\sum_{r=1}^{b_\alpha}
\sum_{s=b_\beta+1}^{n}
\bigl(
\lbEncodingMatrix_{r,s-1}
+
\lbEncodingMatrix_{r-1,s}
-
\lbEncodingMatrix_{r-1,s-1}
-
\lbEncodingMatrix_{r,s}
\bigr)
\;=\;
\lbEncodingMatrix_{b_\alpha b_\beta}.
\label{eq:sum_weights_telescoped}
\end{align}
Since~$\lbEncodingMatrix$ is symmetric with zero diagonal, 
\begin{align}
\lbQuadraticForm{\lbEncodingMatrix}{\lbBoundary{\tau}{X}}
\;=\;
2
\sum_{\alpha=1}^{m-1}
\sum_{\beta=\alpha+1}^{m}
(-1)^{\alpha+\beta}\lbEncodingMatrix_{b_\alpha b_\beta}
\;=\;
2
\sum_{\alpha=1}^{m-1}
\sum_{\beta=\alpha+1}^{m}
(-1)^{\alpha+\beta}
\sum_{r=1}^{b_\alpha}
\sum_{s=b_\beta+1}^{n}
\lbEdgeWeight{r}{s},
\label{eq:quadratic_form_expanded_edge_weights}
\end{align}
where the second equality follows from~\eqref{eq:sum_weights_telescoped}. 

Fix a nonpath edge~${\{v_r,v_s\}}$. 
If~${\lbBoundaryCount=0}$, then the coefficient of~$\lbEdgeWeight{r}{s}$ in~\eqref{eq:quadratic_form_expanded_edge_weights} is~$0$. 
Otherwise, since the boundary indices in~${\{r,\ldots,s-1\}}$ are consecutive in~${b_1<\ldots<b_m}$, suppose that~${\lbBoundary{\tau}{X}\cap\{r,\ldots,s-1\}=\bigl\{b_{\gamma+1},\,b_{\gamma+2},\,\ldots,\,b_{\gamma+\lbBoundaryCount}\bigr\}}$.  
By~\eqref{eq:quadratic_form_expanded_edge_weights}, the coefficient of~$\lbEdgeWeight{r}{s}$ is
\begin{align}
2
\sum_{\alpha=\gamma+1}^{\gamma+\lbBoundaryCount-1}
\sum_{\beta=\alpha+1}^{\gamma+\lbBoundaryCount}
(-1)^{\alpha+\beta}
\;=\;
2
\sum_{\alpha=1}^{\lbBoundaryCount-1}
\sum_{\beta=\alpha+1}^{\lbBoundaryCount}
(-1)^{\alpha+\beta}
\;=\;
\left(
\sum_{\alpha=1}^{\lbBoundaryCount}
(-1)^\alpha
\right)^2
-
\lbBoundaryCount
\;=\;
-2
\left\lfloor
\frac{\lbBoundaryCount}{2}
\right\rfloor.
\label{eq:nonpath_edge_coefficient}
\end{align}   
Substituting~\eqref{eq:nonpath_edge_coefficient} into~\eqref{eq:quadratic_form_expanded_edge_weights}, we obtain
\begin{align*}
\lbQuadraticForm{\lbEncodingMatrix}{\lbBoundary{\tau}{X}}
\;=\;
-2
\sum_{r=1}^{n-2}
\sum_{s=r+2}^{n}
\left\lfloor
\frac{\lbBoundaryCount}{2}
\right\rfloor
\lbEdgeWeight{r}{s}. 
\end{align*}
It follows by~\eqref{eq:cut_value_with_pairs} that~${\cutGraphSet{\lbBaseGraph{\tau}}{X} = m n^3\lbWeightScale + \lbQuadraticForm{\lbEncodingMatrix}{\lbBoundary{\tau}{X}} = |\lbBoundary{\tau}{X}| n^3\lbWeightScale + \lbQuadraticForm{\lbEncodingMatrix}{\lbBoundary{\tau}{X}}}$. Finally, since~$\lbEncodingMatrix$ has zero diagonal,
\[
\left|
\lbQuadraticForm{\lbEncodingMatrix}{\lbBoundary{\tau}{X}}
\right|
\;=\;
\left|
\sum_{\alpha=1}^{m}
\sum_{\beta=1}^{m}
(-1)^{\alpha+\beta}
\lbEncodingMatrix_{b_\alpha b_\beta}
\right|
\;\leq\;
m(m-1)\maxNorm{\lbEncodingMatrix}
\;<\;
n^2\maxNorm{\lbEncodingMatrix}
\;=\;
\frac{\lbWeightScale}{2}.  \qedhere
\]
\end{proof}

\begin{corollary}[Minimum cuts of the base graph; \Cref{lemma:base_graph_properties}\ref{item:base_graph_min_cuts}]
\label{cor:base_graph_minimum_cuts}
The minimum cuts of~$\lbBaseGraph{\tau}$ are exactly the prefix cuts of~$\tau$.
\end{corollary}

\begin{proof}
Fix a nonempty proper set~${X\subsetneq V}$.
If~$X$ represents a prefix cut, then~${|\lbBoundary{\tau}{X}|=1}$ and~${\lbQuadraticForm{\lbEncodingMatrix}{\lbBoundary{\tau}{X}}=0}$, so~\Cref{lemma:base_graph_cut_function} gives~${\cutGraphSet{\lbBaseGraph{\tau}}{X}=n^3\lbWeightScale}$.
If~$X$ represents a non-prefix cut, then~${|\lbBoundary{\tau}{X}|\geq2}$, so~\Cref{lemma:base_graph_cut_function} gives~${\cutGraphSet{\lbBaseGraph{\tau}}{X}>2n^3\lbWeightScale-\lbWeightScale/2>n^3\lbWeightScale}$.
 Hence the minimum cuts of~$\lbBaseGraph{\tau}$ are exactly the prefix cuts of~$\tau$.
\end{proof}

\begin{lemma}[Non-prefix cuts with equal cut values in the base graph cross the same perturbed edges; \Cref{lemma:base_graph_properties}\ref{item:base_graph_boundary_intersections}]
\label{lemma:nonprefix_cuts_with_equal_cut_values_cross_the_same_perturbed_edges}
If~${\emptyset\subsetneq X,Y\subsetneq V}$ represent two non-prefix cuts and satisfy~${\cutGraphSet{\lbBaseGraph{\tau}}{X}=\cutGraphSet{\lbBaseGraph{\tau}}{Y}}$, then 
\[
\lbBoundary{\tau}{X}\cap\lbPerturbIndices
\;=\;
\lbBoundary{\tau}{Y}\cap\lbPerturbIndices.
\]
\end{lemma}

\begin{proof} 
Since~${\cutGraphSet{\lbBaseGraph{\tau}}{X}=\cutGraphSet{\lbBaseGraph{\tau}}{Y}}$, it follows by~\Cref{lemma:base_graph_cut_function} that
\[ 
|\lbBoundary{\tau}{X}|n^3\lbWeightScale+\lbQuadraticForm{\lbEncodingMatrix}{\lbBoundary{\tau}{X}}
\;=\;
|\lbBoundary{\tau}{Y}|n^3\lbWeightScale+\lbQuadraticForm{\lbEncodingMatrix}{\lbBoundary{\tau}{Y}},
\]
and therefore 
\begin{align}
    \lbQuadraticForm{\lbEncodingMatrix}{\lbBoundary{\tau}{X}} 
\;\equiv\;
\lbQuadraticForm{\lbEncodingMatrix}{\lbBoundary{\tau}{Y}} \pmod{\lbWeightScale}.
\label{eq:boundary_quadratic_forms_modulo}
\end{align} 
By the upper bound in~\Cref{lemma:base_graph_cut_function},~${\left|\lbQuadraticForm{\lbEncodingMatrix}{\lbBoundary{\tau}{X}}\right|<\lbWeightScale/2}$ and~${\left|\lbQuadraticForm{\lbEncodingMatrix}{\lbBoundary{\tau}{Y}}\right|<\lbWeightScale/2}$. 
Therefore
\[
|\lbQuadraticForm{\lbEncodingMatrix}{\lbBoundary{\tau}{Y}}-\lbQuadraticForm{\lbEncodingMatrix}{\lbBoundary{\tau}{X}}|
\;<\;
\lbWeightScale, 
\]
and it follows by~\eqref{eq:boundary_quadratic_forms_modulo} that~${\lbQuadraticForm{\lbEncodingMatrix}{\lbBoundary{\tau}{Y}}=\lbQuadraticForm{\lbEncodingMatrix}{\lbBoundary{\tau}{X}}}$. 

Also, since~$X$ and~$Y$ represent non-prefix cuts,~${|\lbBoundary{\tau}{X}|,|\lbBoundary{\tau}{Y}|\geq2}$. 
Therefore~$\lbBoundary{\tau}{X}$ and~$\lbBoundary{\tau}{Y}$ satisfy the hypotheses of~\Cref{lemma:encoding-perturbed-edges}, so~${\lbBoundary{\tau}{X}\cap\lbPerturbIndices=\lbBoundary{\tau}{Y}\cap\lbPerturbIndices}$. 
\end{proof}

The result of~\Cref{lemma:base_graph_properties} is now immediate by combining~\Cref{lemma:base_graph_edge_weights,cor:base_graph_minimum_cuts,lemma:nonprefix_cuts_with_equal_cut_values_cross_the_same_perturbed_edges}.

\end{document}